\documentclass[12pt,a4paper]{article}
\usepackage{epsfig}
\usepackage{natbib,color}
\usepackage{amsmath}
\usepackage{array}
\usepackage{enumerate}
\usepackage{amsfonts}
\usepackage{amssymb}
\usepackage{graphicx}
\usepackage{amsthm}
\usepackage{booktabs}
\usepackage{chngcntr}
\usepackage{apptools}
\usepackage{tikz}
\usetikzlibrary{calc}
\usetikzlibrary{decorations.pathreplacing}
\usepackage[a4paper, lmargin=0.1666\paperwidth, rmargin=0.1666\paperwidth, tmargin=0.1111\paperheight, bmargin=0.1111\paperheight]{geometry}
\theoremstyle{definition}
\newtheorem{definition}{Definition}[section]
\newtheorem{theo}{Theorem}
\newtheorem{lemm}{Lemma}
\newtheorem{prop}{Proposition}
\newtheorem{coro}{Corollary}

\newtheorem{rema}{Remark}

\AtAppendix{\counterwithin{lemm}{section}}
\AtAppendix{\counterwithin{theo}{section}}
\AtAppendix{\counterwithin{prop}{section}}
\AtAppendix{\counterwithin{definition}{section}}
\AtAppendix{\counterwithin{rema}{section}}

\newcommand{\be}{\begin{eqnarray}}
\newcommand{\ee}{\end{eqnarray}}
\newcommand{\bea}{\begin{eqnarray*}}
\newcommand{\eea}{\end{eqnarray*}}

\newcommand{\bq}{\begin{quote}}
\newcommand{\eq}{\end{quote}}
\newcommand{\beq}{\begin{equation}}
\newcommand{\eeq}{\end{equation}}
\newcommand{\beqn}{\begin{eqnarray*}}
\newcommand{\eeqn}{\end{eqnarray*}}

\newcommand{\ru}{\mbox{Uniform}}

\newcommand{\bt}{\begin{itemize}}
\newcommand{\et}{\end{itemize}}

\theoremstyle{definition}
\newtheorem{example}{Example}[section]

\theoremstyle{definition}
\newtheorem{ass}{Assumption}

\title{Living on the Edge:\\ An Unified Approach to Antithetic Sampling}
\author{Roberto Casarin\setcounter{footnote}{1}\footnotemark{}\hspace{15pt} %
Radu V. Craiu\setcounter{footnote}{2}\footnotemark{}\\ %
Lorenzo Frattarolo\setcounter{footnote}{3}\footnotemark{}%
\hspace{15pt}
 Christian P. Robert \setcounter{footnote}{4}\footnotemark{}  
\\\\%
{\centering \small\setcounter{footnote}{1}\footnotemark{} University Ca' Foscari of Venice}\\%
{\centering \small\setcounter{footnote}{2}\footnotemark{} University of Toronto}\\
 {\centering \small\setcounter{footnote}{3}\footnotemark{} European Commission Joint Research Centre}\\
 {\centering \small\setcounter{footnote}{4}\footnotemark{} CEREMADE, Universit\' e Paris-Dauphine and  University of Warwick}\\
}
\graphicspath{{./figs/}}
\begin{document}
\maketitle
\begin{abstract} 
 We identify recurrent ingredients in the antithetic sampling literature leading to a unified sampling framework. We introduce a new class of antithetic schemes that includes the most used antithetic proposals. This perspective enables the derivation of new properties of the sampling schemes: i) optimality in the Kullback-Leibler sense; ii) closed-form multivariate Kendall's $\tau$ and Spearman's $\rho$; iii)ranking in concordance order and iv) a central limit theorem that characterizes stochastic behavior of  Monte Carlo estimators when the sample size tends to infinity. Finally, we provide applications to  Monte Carlo integration and Markov Chain Monte Carlo Bayesian estimation. 
\end{abstract}
{\it Keywords: Antithetic variables, Countermonotonicity, Monte Carlo, Negative dependence, Variance reduction}

\section{Introduction}

The Monte Carlo method is at the core of model-based scientific exploration.  In its simplest form it relies on approximating an integral $\mathfrak{I}=\int f(\mathbf{x})\pi(d\mathbf{x})$ with $\hat{\mathfrak{I}}_d = \frac{1}{ d}  \sum_{i=1}^d f(\mathbf{X}_i)$ when $\pi$ is a probability measure, $f:\mathbb{R}^p\mapsto \mathbb{R}$ is a integrable function with respect to $\pi$, $d$ is the Monte Carlo sample size and $\mathbf{X}_1,\ldots, \mathbf{X}_d$ are independent, identically distributed (henceforth, iid) samples from $\pi$. 

In modern computational problems, sampling from the distribution $\pi$ may be expensive,  in terms of either computational effort or time,  so techniques needed to reduce the Monte Carlo sample size $d$, while maintaining the desired precision in estimation, are essential. A relevant class is represented by the {\it variance reduction techniques} that use statistical properties induced by the sampling design to reduce $Var(\hat{\mathfrak{I}}_d)$. For instance, in the case $p=1$, if the independence condition between samples $X_1,\ldots, X_d$ is dropped then 
\beq
\hbox{Var}(\hat{\mathfrak{I}}_d)= \dfrac{1}{d}\sum^{d}_{i=1}\hbox{Var}\left(f\left(X_i\right)\right) + \dfrac{1}{d^2}\sum_{i \neq  j} \hbox{Cov}\left( f\left(X_i \right),f\left(X_j\right)\right),
\label{antivar}
\eeq
and the variance is reduced, compared to independent sampling, if the average covariance is negative.

Antithetic sampling designs aim at minimizing the covariances between samples while preserving their marginal distribution.  A historical perspective on the strategies for antithetic sampling  \citep[e.g., ][]{hammersley_mauldon_1956,ham-mo} allows us to better understand  the rationale behind various constructions and to establish useful  relationships with the results available from related fields, such as stochastic orders \citep[e.g., ][]{barlow1975statistical}, optimal transport \citep[e.g., ][]{Gaffke1981}, Fr\'echet classes \cite[e.g.][]{Whitt76}, and group transformation \cite[e.g.][]{Andreasson1972}. Our historical review identifies some key recurrent ingredients used to propose a unified framework for antithetic sampling. We introduce a new class of antithetic constructions that also includes some of the, to our knowledge, most used antithetic proposals, which are reviewed later in this section. Moreover, this new perspective enables the derivation of new properties of the sampling schemes for $p$ stochastically independent replications  ($p\geq 1$): \textit{i)} optimality in the Kullback-Leibler sense; \textit{ii)} closed-form multivariate Kendall's $\tau$ and Spearman's $\rho$; \textit{iii)} ranking in concordance order, and \textit{iv)} a central limit theorem that characterizes stochastic behavior when $d$ tends to infinity.

The pairwise  antithetic coupling  introduced by  \cite{ham-mo} 
achieves variance reduction by generating $d/2$ (we assume $d$ is even in \eqref{antivar}) iid pairs of negatively correlated random variables. The joint distribution of any pair, say $(X_1,X_2)$, achieves the lower Fr\'echet bound \citep[see][]{Frechet35}. This is achieved by sampling $X_{1} \sim \pi$ and sampling the second variable $X_{2}$ via  \[X_{2}=F_\pi^{-1}\left(1- F_\pi\left(X_{1}\right)\right),\] where $F_\pi$ is the cumulative distribution function (CDF) for the distribution $\pi$. This procedure minimizes the correlation  for any monotonic $f$.  
The explanation of this reduction, as correctly pointed out by \cite{Whitt76},  can be found in \cite{hoeff} and \cite{Frechet1951} and it is due to the rearrangement inequality \citep[see chapter X of][]{hardy1934}. Interestingly, the construction cannot be unambiguously extended beyond pairs because the lower Fr\'echet bound of all $k$-variate distributions is itself a distribution only when $k=2$.

In the following, we establish a relationship between variance reduction and stochastic orders. This relationship allows generalizing variance reduction to higher dimensions $d>2$. 

\begin{definition}[Concordance Order \citep{joe1992}] Let $\mathbf{X}$ and $\mathbf{Y}$ be random vectors with CDFs $F$ and $G$ and survival function $\bar{F}$ and $\bar{G}$, respectively. Then $\mathbf{Y}$ is more concordant than $\mathbf{X}$ (written $\mathbf{X}\prec_{C}\mathbf{Y}$) if
\begin{equation}
\begin{array}{ccc}
F\leq G&\hbox{ and }&\bar{F}\leq \bar{G}
\end{array} 
\end{equation}
\end{definition}
The pairwise antithetic construction being based on the lower Fr\'echet bound is minimal in concordance order. If the univariate marginals are the same in the two-dimensional case, then the concordance order is equivalent to the correlation order, which is a stochastic order induced by the covariance among coordinate-wise monotonic functions. The correlation order was introduced in \cite{barlow1975statistical} an studied further in \cite{dhaene1996dependency}. In this paper, we follow the multivariate definition and notation from \cite{lu2004generalized}.

\begin{definition}[Correlation Order \citep{lu2004generalized}]
Denote with $\mathcal{G}^{d}$ the set of real-valued bounded coordinate-wise non-decreasing Borel functions on $\mathbb{R}^d$, $d\in \mathbb{N}$ and with $\mathcal{D}=\left\{1,\ldots,d\right\}$ the set of coordinate indexes. The d-variate random vector $\mathbf{X}$ is less correlated than the d-variate random vector $\mathbf{Y}$, written $\mathbf{X}\prec_{corr} \mathbf{Y}$,
if they have the same univariate marginal CDFs $F_i$, $i=1,\ldots,d$ and if for every pair of disjoint finite subsets $\mathcal{I},\mathcal{J}\subset \mathcal{D}$ and all functions $f\in\mathcal{G}^{\left\vert \mathcal{I} \right\vert}, g\in\mathcal{G}^{\left\vert \mathcal{J} \right\vert}$, for which the covariances exist:
\[\mathbb{C}\mathrm{ov}\left(f\left(X_\mathcal{I}\right),g\left(X_\mathcal{J}\right)\right)\leq \mathbb{C}\mathrm{ov}\left(f\left(Y_\mathcal{I}\right),g\left(Y_\mathcal{J}\right)\right) \]
\end{definition}
Since concordance and correlation orders are equivalent in dimension two \citep[see, for example, ][]{dhaene1996dependency}, antithetic coupling minimizes the covariance among monotone functions.

If we drop the monotonicity assumption about $f$ in \eqref{antivar}, the discussion and derivation of lower bounds for the variance are more complex and less general. If $p=1$ and $d=2$, for non-monotonic, bounded $f$, \cite{hammersley_mauldon_1956} prove that the lower bound of the variance can be attained only by a multivariate transformation of a single standard uniform random variable which, almost surely, is coordinate-wise monotonic. The proof relies on two main ingredients. First, the monotonic transformation introduces an approximate representation of the class of bivariate distributions with uniform marginals. The candidate member of the class is approximated by partitioning the unit square in sub-squares of side $1/n$.  The approximation is a doubly stochastic matrix in which each element corresponds to a sub-square and has a value equal to the mass assigned to the corresponding sub-square.
This construction relies on the bijective rearrangement \citep{Puccetti2015} also known as measure-preserving transformation \citep[e.g.][]{brown1966,vitale1990}. In the interpretation of \cite{vitale1990}, for every random vector $\left(U_1 , U_2\right)$ on the unit square, with standard uniform marginals, there is a sequence of bijective maps $f_n$ such that $\left(U_1 , f_n\left(U_1\right)\right)$ weakly converges to $\left(U_1 , U_2\right)$. Bijective rearrangement and the induced stochastic dependence \citep{DURANTE2012} are relevant to our discussion of the antithetic constructions for $d>2$.
 

The second ingredient is the Birkhoff-Von Neumman's decomposition \citep{birkhoff1946,vonneumann1953} of doubly-stochastic matrices in which they are represented as convex combinations of permutation matrices. \cite{handscomb_1958} extends the latter result to characterize the extremal points of multi-stochastic arrays as higher-dimension permutation arrays and to provide a generalization for $d>2$ of the results in \cite{hammersley_mauldon_1956}.

Those early results based on discretization give sufficient conditions to characterize the transformations needed to obtain a minimal variance. Unfortunately, the characterizations are not constructive and do not provide feasible random sampling algorithms. Moreover, the existence of the optimal transformation minimizing the variance is not guaranteed.

This led earlier researchers to propose feasible, yet sub-optimal sampling solutions, including  Hammersley and Morton's  \citep{ham-mo} proposal for $d>2$. \cite{Andreasson1972}, \cite{AndreassonDahlquist1972} and \cite{Roach1977} follow a group theoretic approach. In particular, \cite{Roach1977}, build on \cite{Andreasson1972}, \cite{AndreassonDahlquist1972} and \cite{tukey_1957},  draws a parallel with systematic sampling. Its solutions for $d=2$, in the case of non-monotonic $f$s, relies on  discretization, optimal transport and a branch and bound algorithm. The group theoretic approach was also used in \cite{Fishman1983} to obtain a reinterpretation of the original \cite{ham-mo} proposal for $d>2$. Their construction was named rotation sampling and is described next. 
\begin{example}[Rotation sampling]\label{ex:rotsam}
\begin{eqnarray}\label{rotsamp}\nonumber
U_1&=&U\sim \mathcal{U}[0,1]\\
U_{l}&=& \left(\dfrac{l-1}{d} + U\right)\mathrm{mod}\,1, \quad l\in\left\{2,\ldots,d\right\}
\end{eqnarray}
\end{example}
This proposal is a particular case of our stochastic representation. 

The extension to unbounded functions of the theorems in \cite{hammersley_mauldon_1956} and \cite{handscomb_1958}  can be found in \cite{Wilson1979} for $d\geq2$ and $p=1$ and in \cite{Wilson1983}  for $d\geq2$,  $p>=1$. The latter paper combines discretization and bijective rearrangement with the optimal transport assignment problem to prove the results. Bijective rearrangement and Monge-Kantorowitch transportation problem are used in \cite{Gaffke1981} to obtain minimum variance constructions for $f$ equal to the identity function. The authors are the first to realize that the \cite{ham-mo} bivariate antithetic vector has an almost sure constant sum, which is one of the main ingredients of our unified approach.

The relationship between constant sum and variance reduction is trivial. Random vectors of dimension $d\geq 2$ with constant sum achieve the smallest variance for the sum of their components. Then they minimize the variance in the case $p=1$ and $f$ equal to the identity function. Beyond that, a recent stream of papers proves that the constant sum vectors are among the minimal vectors with respect to the concordance order. In particular, one possible generalization of the constant sum constraint is the following one.
\begin{definition}[$l$-countermonotonic] A $d$-dimensional random vector $\mathbf{U}$ with uniform marginals, is said to be $l$-countermonotonic ($l$-CTM), if there exist some 
$\mathcal{L}\subseteq\mathcal{D}$ with $\left\vert\mathcal{L}\right\vert=l$, a family $\left\{g_{l}\right\}_{l\in\mathcal{L}}$ of strictly increasing continuous functions $\left[0,1\right]\mapsto \mathbb{R}$ and some $k\in \mathbb{R}$ such that :
\begin{equation}\label{lcounter}
\sum_{l\in\mathcal{L}}g_{l}\left(U_l\right)=k\hbox{   a.s.}
\end{equation}
\end{definition}
Theorem 2 and Proposition 1 in \cite{LEE2014} show that the antithetic vector is the only element of the $2$-CTM class
and it is minimal in the concordance and correlation orders. The equivalence among orders does not hold in the multivariate case, but several relationships have been identified. For instance,  \cite{lu2004generalized} have showed that $\mathbf{X}\prec_{corr} \mathbf{Y}$ implies $\mathbf{X}\prec_{C} \mathbf{Y}$. 

In addition, conditions for achieving minimality in the concordance order were linked to $l$-CTM random vectors. If $\mathbf{V}$ is $d$-CTM or $(d-1)$-CTM and $\mathbf{U}\prec_{C} \mathbf{V}$ then $\mathbf{U}$ has the same distribution as  $\mathbf{V}$  \citep{lee2017multivariate}. The same result has been shown to hold for $\prec_{corr}$ order 
  \citep[see][remark 3.1]{ahn2020minimal}.  Since $d$-CTM and $(d-1)$-CTM are among the minimal elements in the correlation order, they are promising candidates for constructing variance reduction techniques in any dimension $d$. In this paper, we study some of the existing sampling methods and propose new constructions for $d$-CTM vectors of \ru(0,1) random variables with a.s. constant sum, that is $g_{l}\left(U_l\right)=U_l$, $l\in \mathcal{D}$ in Eq \eqref{lcounter}. This subclass is known in the literature as strict $d$-CTM \citep{LEE2014}. \cite{Gaffke1981} recognize that \cite{ham-mo} is strict $2$-CTM and provide the first strict $3$-CTM construction given below.
\begin{example}[\cite{Gaffke1981} strict $3$-CTM]\label{ex:rushsample}
\begin{eqnarray}\label{rushsample}\nonumber
U_1 &=& U,  \qquad U \sim U\left[0,1\right]
\\
U_2 &=& U + \frac{1}{2}\mathbb{I}_{\left[0,1/2\right]}\left(U\right) -\frac{1}{2}\mathbb{I}_{\left[1/2,1\right]}\left(U\right)  \\\nonumber
U_3 &=& -2U + \mathbb{I}_{\left[0,1/2\right]}\left(U\right) +2\mathbb{I}_{\left[1/2,1\right]}\left(U\right) 
\end{eqnarray}
\end{example}
For the case $d>3$, the authors propose to generate a sequence of independent random vectors using their representation in \eqref{rushsample} and the bivariate antithetic vector of \cite{ham-mo}. 
\begin{example}[\cite{Gaffke1981} strict $d$-CTM]\label{ex:GR}
Let $V_i$, $i=1,\ldots, \left\lfloor  (d-2)/2\right\rfloor +1$ independent random variables, and
\begin{eqnarray}\nonumber
U_{2i-1}=V_i, \quad U_{2i}=1-V_i,\quad i=1,\ldots, \left\lfloor  (d-2)/2\right\rfloor 
\end{eqnarray}
with  
\begin{eqnarray}
U_{d-1} =V_{\left\lfloor  (d-2)/2 \right\rfloor +1},\quad U_{d} = 1 - V_{\left\lfloor  (d-2)/2\right\rfloor}
\end{eqnarray}
if $d$ even, and
\begin{eqnarray}\nonumber
U_{d-2} &=& V_{\left\lfloor (d-2)/2 \right\rfloor +1}  \\\nonumber
U_{d-1} &=& V_{\left\lfloor  (d-2)/2 \right\rfloor +1} + 1/2\mathbb{I}_{\left[0,1/2\right]}\left(V_{\left\lfloor  (d-2)/2 \right\rfloor +1}\right) -1/2\mathbb{I}_{\left[1/2,1\right]}\left(V_{\left\lfloor  (d-2)/2 \right\rfloor +1}\right)  \\\label{GR3}
U_{d} &=& -2U + \mathbb{I}_{\left[0,1/2\right]}\left(V_{\left\lfloor  (d-2)/2 \right\rfloor +1}\right) +2\mathbb{I}_{\left[1/2,1\right]}\left(V_{\left\lfloor  (d-2)/2 \right\rfloor +1}\right)
\end{eqnarray}
if $d$ is odd. 
\end{example}

Almost contemporaneously, \cite{Arvi:John:82:VRT} put forward the apparently different proposal given in the following.

\begin{example}[\cite{Arvi:John:82:VRT} strict $d$-CTM]\label{ex:AV}
\begin{eqnarray}\label{AV}\nonumber
U_1 &=& U  \qquad U \sim \mathcal{U}[0,1]\\
U_i&=& \left( 2^{i-2} U_1 + 1/2 \right)\mathrm{mod}\,1, \quad i=\left\lbrace2,\ldots,d-1\right\rbrace \\\nonumber U_d &=& 1- \left( 2^{d-2} U_1 \right)\mathrm{mod}\,1.
\end{eqnarray}
\end{example}
We will show that for $d=3$, the two proposals in Examples \ref{ex:rushsample} and \ref{ex:AV} coincide and are special cases of our general stochastic representation. Both constructions yield vectors with a constant sum, but \cite{Gaffke1981} proposal's use of independent random variates for $d>3$ made us wonder about its efficiency, especially since the results in \cite{hammersley_mauldon_1956} and \cite{handscomb_1958} suggest that combinations of independent vectors could be sub-optimal. We compare different strict $d$-CTM constructions using the concordance order. According to \cite{ahn2020minimal},  all strict $d$-CTM  have minimal multivariate Kendall's $\tau$, but they can have different multivariate Spearman's $\rho$ values. For example, \cite{Gaffke1981}, and \cite{Arvi:John:82:VRT} proposals have the same values for Kendall's $\tau$, but different ones for Spearman's $\rho$.  

Other examples of strict $d$-CTM vectors, partially covered by our representation, can be found in \cite{KNOTT2006} , \cite{LEE2014}, and after a linear transformation also the construction in \cite{BUBENIK2007} and references therein, can be seen as strict $d$-CTM vectors. 

The range of application of CTM constructions has been extended to other marginal distributions. For instance, \cite{Ruschendorf2002}, expands the work in \cite{Gaffke1981} to random variables $\{Y_1,\ldots, Y_d\}$ with 
unimodal distributions using the   Levy-Shepp form of the Khinchine representation theorem \citep{levy1962,shepp1962}, as $Y_i= X V_i$ where $V_i \sim U(-1,1)$ for all $1\le i \le d$. Hence, a CTM construction for $\{V_1,\ldots, V_d\}$ implies constant sum for for $\{Y_1,\ldots, Y_d\}$. Trivially, such a  vector will achieve the smallest variance for the sum of its components. In those cases, the literature refers to these vectors as complete or joint mix, differentiating between having identical or different marginals  \citep[see][and references therein]{Puccetti2015}. Our general construction can be extend to non-uniform marginals following \cite{Ruschendorf2002}.

In \cite{Rubinstein1987}, a different extension of \cite{handscomb_1958} theorem was proposed by dropping the bijective condition for the rearrangement. They prove the existence of antithetic solutions that minimize the variance. According to the authors,  optimal antithetic solutions should be a function of only one uniform random variable, without restriction on the functional dependence. Unfortunately, dropping the bijective condition result in a tautological statement because, as shown in \cite{brown1966}, \cite{Whitt76}, \cite{vitale1990} and recently reformulated in theorem 1 of \cite{Puccetti2015}, every random vector can be expressed as a function of only one uniform random variable. 
This difficulty of narrowing down conditions for the existence of optimal antithetic variables is linked to the challenge of extending Birkoff-von Neumman representation to the continuous case. It is, in fact, well known that bijective rearrangements are only a sub-class of the extremal transformations. For example, the $d=2$ case is known as Birkhoff's problem 111 \citep{isbell1955}, and even if there exists a characterization  \citep{Lindenstrauss65}, the necessary and sufficient conditions in their most recognizable form \citep{moameni_2016} are of limited practical relevance. For a discussion and an example of an extremal non-bijective class in the multivariate case, refer to \cite{Durante2014}. Since a general characterization is out of reach, we solve the optimal transport problem for transformations in the extremal class and produce a  stochastic representation that depends on a single standard uniform.  

 Historically, given the impossibility of obtaining optimal and feasible antithetic plans, by mid 80's the literature shifted the focus  to negative dependence. In particular, a procedure  considered close to antithetic sampling,  but applicable to the general $d\geq 2$, $p\geq1$ is the Latin Hypercube sampling introduced in \cite{McKa:Beck:Cono:79:CTM}. 
 \begin{example}[\cite{McKa:Beck:Cono:79:CTM} Latin Hypercube]\label{ex:LH}
 Given a standard uniform $d$-dimensional random vector $\mathbf{U}$  and let $\mathcal{D}^{\pi}_t = \left(\pi_{t}\left(0\right),\ldots,\pi_{t}\left(d-1\right)\right)^{T}$ be a permutation of $\lbrace 0,1,\ldots,d-1\rbrace$ independent of  $\mathbf{U},\ldots,\mathbf{U}_1$ and 
\begin{eqnarray}\label{LH}
\mathbf{U}_1= \dfrac{1}{d}\left(\mathcal{D}^{\pi}_t + \mathbf{U}\right)
\end{eqnarray}
\end{example}
The simplicity of the method, the guarantee of asymptotic variance reduction \citep{Stei:87:LSP} and the availability of a central limit theorem \citep{Owen:92:CLT} made it one of the most common variance reduction strategies. The relationship with antithetic variates was studied in \cite{CraiuMeng2005}, where through the introduction of an iterative version of the method, the Iterated Latin Hypercube (ILH), it is shown that, in the iteration limit, the resulting random vector has an almost-sure constant sum. Our new representation allows comparing ILH and its combination with other antithetic proposals. Finally, we extend the central limit theorem in \cite{Owen:92:CLT}, showing the irrelevance of the starting distribution when $d$ goes to infinity, and the number of iterations is fixed.

The paper is structured as follows: Section \ref{sec:stochrep} contains the description of the unified representation, its interpretation, and its invariant transformations. Section \ref{sec:distribution} discusses distributional properties and concordance measures.
New and old illustrations of the unified representation and their ranking are presented in Section \ref{sec:examples} followed by the derivation of the general central limit theorem for Latin Hypercube in Section \ref{sec:CLT}. Numerical illustrations are presented in Section \ref{sec:numerics} and the paper ends with a discussion  of future directions for research in Section \ref{sec:discussion}.

\section{Sampling on Line Segments}\label{sec:stochrep}
We introduce a general method for constructing antithetic vectors whose components have a standard uniform, \ru(0,1), as marginal distribution.  Non-uniform  variables can be obtained using various transformations, e.g. the inverse CDF method or  the Levy-Shepp form of the Khinchine representation theorem \citep{levy1962,shepp1962}. We study conditions for achieving  $d$-CTM  and show that several known countermonotonic random vectors used in variance reduction, can be obtained as special cases of our general construction. 

\subsection{Standard Antithetic Construction}
Our method relies on sampling with equal probability on a collection $\mathcal{S}$ of  line segments in the $d$-dimensional Euclidean space. Since each segment is uniquely characterized by its endpoints or vertexes, the collection $\mathcal{S}$ can be equivalently represented by the set of vertex pairs that define the segments and their coordinates. This representation is efficient in large dimensions even when the segments share some of their vertexes.

More formally, let us define a vertex set $\mathcal{V}= \left\lbrace  1,\ldots,n\right\rbrace$ as a set of points in the $d$-dimensional hypercube, the coordinates of the $k$-th vertex as the column vector $\mathbf{x}_k \equiv \left(x_{1k},\ldots,x_{dk}\right)^{T} \in  \left[0,1\right]^d$ and the coordinate matrix $\mathbf{X}= \left(\mathbf{x}_1,\ldots,\mathbf{x}_n\right)$ as the collection of vertex coordinates. We assume there is an edge $e=(i,j)$ between $i$ and $j$, with $i<j$, if there is a segment joining the two vertices $i$ and $j$, and define the collection of segments by the edge set $\mathcal{E}=\left\lbrace \left(i,j\right) \in \mathcal{V}\times \mathcal{V}\right\rbrace$. Then $\mathcal{G} = \left\{\mathcal{V},\mathcal{E}\right\}$ is an undirected graph and $\mathcal{S} = \left\{\mathcal{G},\mathbf{X}\right\}$ is the collection of segments. 
The lexicographic order on vertex indexing induces an order on the edge set, such that the $k$-th element $e_k\in\mathcal{E}$ is uniquely associated to its couple of vertices, defining the map $\{1,\ldots,\left\vert\mathcal{E}\right\vert\}\mapsto\mathcal{E}$, $k\rightarrow\left(i\left(k\right),j\left(k\right)\right)$. 

Our stochastic construction relies on the graph representation $\mathcal{G}$ and requires a properly chosen vertex matrix $\mathbf{X}$ and two standard uniform random numbers\footnote{As in the antithetic couple case it is also possible to use only one standard uniform random number $W$  by setting $V= \left\{ \right\vert\mathcal{E} \left\vert W\right\}$.}
\begin{enumerate}
\item draw $V\sim \mathcal{U}[0,1]$ and $W \sim \mathcal{U}[0,1]$
\item choose with uniform probability on the edge set $\mathcal{E}$ the edge $e_K$ by computing $K = \left\lfloor \left\vert \mathcal{E}\right\vert W\right\rfloor +1$
\item obtain the random vertexes pair $\left(I,J\right)=\left(i\left(K\right),j\left(K\right)\right)$
\item obtain a random point on the segment joining vertexes $I$ and $J$ with uniform probability
\begin{eqnarray}\label{simplevar}
U_1&=& x_{1I} V + \left(1-V\right)x_{1J}\nonumber\\
&\vdots&\\
U_d&=& x_{dI} V + \left(1-V\right)x_{dJ}\nonumber
\end{eqnarray}
\end{enumerate} 

The following example shows that the standard antithetic method is a special case of our general sampling construction. 

\begin{example}\label{ex21}
Let us consider $d=2$ and sampling in the unit square on the diagonal joining the vertex $1$ of coordinates  $\mathbf{x}_1 \equiv \left(x_{11}=1,x_{21}=0\right)^{T}$ and the vertex $2$ of coordinates $\mathbf{x}_2 \equiv \left(x_{12}=0,x_{22}=1\right)^{T}$ (see Figure \ref{fig:antisegments}). Let $V\sim \mathcal{U}[0,1]$ and compute :
\begin{eqnarray*}
U_1&=& x_{11} V + \left(1-V\right)x_{12}= V\\
U_2&=& x_{21} V + \left(1-V\right)x_{22}= 1-V
\end{eqnarray*}

\begin{figure}[t]
\centering
\begin{tabular}{cc}
\begin{tikzpicture}[scale=1.5]
    \draw [<->,thick] (0,2.5) node (yaxis) [above] {$x_2$}
        |- (2.5,0) node (xaxis) [right] {$x_1$};
    \coordinate (c) at (2,2);
    \coordinate (A) at (0,2);
   \coordinate (B) at (2,0);
      \coordinate (D) at (1.5, 0.5);

  \draw[black] (A) -- (B);
    \draw (yaxis |- c) node[left] {$1$}
        -| (xaxis -| c) node[below] {$1$};
    \fill[black] (A) circle (1pt)  ; 
    \fill[black] (B) circle (1pt)  ;  
    \draw[black] (A)++(0,0.1)      node[above,right] {$\mathbf{x}_{2}$};
     \draw[black] (B)++(0,0.1)      node[below,right] {$\mathbf{x}_1$};
    \draw[black] (D) circle (1pt)      node[below,right] {$\mathbf{U}$};
    \draw [dashed](yaxis |- D) node[left] {$1-V$}
        -| (xaxis -| D) node[below] {$V$};
\end{tikzpicture}
&
\begin{tikzpicture}
\filldraw 
 (0,3)  circle (1pt) node[align=center,above] {$2$}--
 node[above,sloped] {$e_1$}
 (3,0) circle (1pt) node[align=left,below] {$1$};
 \end{tikzpicture}
\end{tabular}
\caption{Left: support set, i.e. the segment joining vertexes $1$ and $2$ of coordinates $x_1$ and $x_2$, for the antithetic sampling. Right: dependence graph $\mathcal{G}$ with $\mathcal{V}=\{1,2\}$ and $\mathcal{E}=\{e_1\}$.}\label{fig:antisegments}
\end{figure}
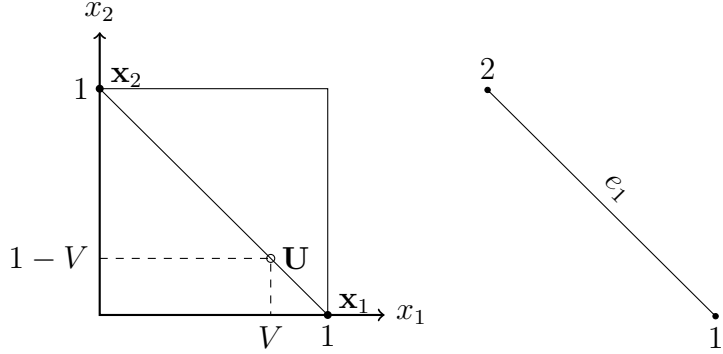

Thus, sampling one uniform antithetic couple $V$ and $1-V$ is equivalent to sampling on a segment (see left plot in Fig. \ref{fig:antisegments}) and the support of the samples can be summarized by the vertex coordinate matrix 
\begin{eqnarray}
\mathbf{X}= \left(\begin{array}{cc}  x_{11} & x_{12}\\ x_{21} & x_{22} \end{array}\right)
\end{eqnarray}
and the couple of vertexes $e_1=\left(1,2\right)$ of the segment we are sampling on (right plot). 

The marginal uniformity of the samples on the segment follows from the convex combination representation and the standard uniform assumption for $V$
\begin{eqnarray}
\begin{array}{cc}
U_l \sim \mathcal{U}\left[\min\left(x_{l1},x_{l2}\right),\max\left(x_{l1},x_{l2}\right)\right]& l=1,2
\end{array}
\end{eqnarray} 
In addition, the standard uniformity of $U_l$ follows from the assumption $\max\left(x_{l1},x_{l2}\right)=1$ and $\min\left(x_{l1},x_{l2}\right)=0$ for all $l=1,2$, and the $d$-CTM property from conditions on the vertex coordinates: $x_{11}+x_{21}=x_{12}+x_{22}=1$. 
\end{example}

Example \ref{ex21} illustrates the fact, which is easy to prove in full generality, that our construction generates samples with uniform probability on $\mathcal{S}$. However, it does not guarantee  that all the  marginal distributions  of the components $U_l$, for $l\in\mathcal{D}$, are  $\ru(0,1)$. In the following, we study the conditions on $\mathcal{G}$ and $\mathbf{X}$, such that the variables $U_l$, $l\in\mathcal{D}$ are conditionally uniform, given the choice of the segment, and marginally $\ru(0,1)$.

\subsection{Uniformity}
We provide conditions on the collection of segments $\mathcal{S}$ in $d$-dimensional hypercube to achieve standard marginal uniformity and $d$-CTM when using our construction. Before introducing the general result we introduce some notation and discuss the main assumptions.

For the general case, the following condition rules out atomic and mixed probability measures.
\begin{ass}[Admissibility]\label{assAdmiss}
The set $\mathcal{S}=\{\mathcal{G},\mathbf{X}\}$ is admissible if all segments in $\mathcal{S}$ are not contained in any of the $(d-1)$-hyperplanes that are parallel to a $(d-1)$-dimensional hyperface of the unit hypercube $[0,1]^{d}$.
\end{ass}

We provide some intuition for Assumption \ref{assAdmiss} through the following  $2$-dimensional example. 
\begin{example}\label{ex2by2}
Consider the collection  segments in the left plot of Figure \ref{fig:edgelabel} with
  coordinate matrix:
\begin{eqnarray}\label{excoord}
\mathbf{X}&= &\left(\begin{array}{cccc}  \alpha &\beta & \gamma&\alpha\\\beta &\alpha& \beta&\alpha\end{array}\right)
\end{eqnarray}
where $\alpha<\beta\leq \gamma\in\mathbb{R}$.
 According to the lexicographic order map $k\rightarrow\left(i\left(k\right),j\left(k\right)\right)$ the edge set $\mathcal{E}^*$ is: 
\begin{eqnarray*}
\left\{e^{*}_1=\left(1,2\right),e^{*}_2=\left(1,3\right),e^{*}_3=\left(1,4\right),e^{*}_4=\left(2,3\right),e^{*}_5=\left(2,4\right),e^{*}_6=\left(3,4\right)\right\}
\end{eqnarray*}
(see right plot). 
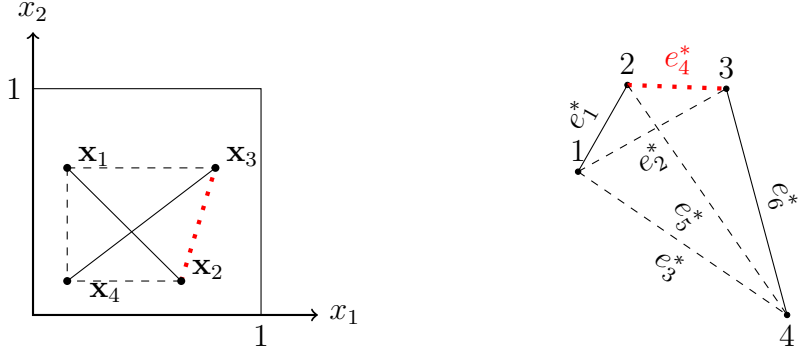
\begin{figure}[t]\label{fig:edgelabel}
\centering
\begin{tabular}{cc}
\begin{tikzpicture}[scale=1.5]
    \draw [<->,thick] (0,2.5) node (yaxis) [above] {$x_2$}
        |- (2.5,0) node (xaxis) [right] {$x_1$};
        
     \coordinate (c) at (2,2);
     \coordinate (X2) at (1.3,0.3);
     \coordinate (X1) at (0.3,1.3);
     \coordinate (X3) at (1.6,1.3);
     \coordinate (X4) at (0.3,0.3);

 \draw[black] (X1) -- (X2);
  \draw[black] (X3) -- (X4);
  \draw[loosely dotted,red,ultra thick] (X2) -- (X3);
  \draw[dashed] (X2) -- (X4);
  \draw[dashed] (X1) -- (X3);
  \draw[dashed] (X1) -- (X4);
    \draw (yaxis |- c) node[left] {$1$}
        -| (xaxis -| c) node[below] {$1$};
   \fill[black] (X2) circle (1pt);
      \fill[black] (X3) circle (1pt);
   \fill[black] (X4) circle (1pt);
   \fill[black] (X1) circle (1pt);
   \draw[black] (X2)++(0,0.1)      node[above,right] {$\mathbf{x}_2$};
     \draw[black] (X1)++(0,0.1)      node[below,right] {$\mathbf{x}_1$};
    \draw[black] (X3)++(0,0.1)      node[above,right] {$\mathbf{x}_3$};
     \draw[black] (X4)++(0.1,-0.1)      node[below,right] {$\mathbf{x}_4$};
\end{tikzpicture}
&\hspace{50pt}
\begin{tikzpicture}

     \coordinate (X2) at (0.5,1.15);
     \coordinate (X1) at (-0.15,0);
     \coordinate (X3) at (1.8,1.1);
     \coordinate (X4) at (2.6,-1.9);

\filldraw 
 (X2)  circle (1pt) node[align=center,above] {$2$}--
 node[above,sloped] {$e^*_1$}
 (X1) circle (1pt) node[align=left,above] {$1$};
 \filldraw 
 (X4)  circle (1pt) node[align=left,   below] {$4$}--node[above,sloped] {$e^*_6$}
 (X3) circle (1pt) node[align=left,   above] {$3$};

  \draw[loosely dotted,red,ultra thick] (X2) --node[above,sloped] {$e^*_4$} (X3);
  \draw[dashed] (X2) --node[below,sloped] {$e^*_5$} (X4);
  \draw[dashed] (X1)  --node[below,sloped] {$ e^*_2\quad$} (X3);
  \draw[dashed] (X1)  --node[below,sloped] {$e^*_3$} (X4);
\end{tikzpicture}
\end{tabular}
\caption{Left: support set of the sampling scheme. Right: admissible (solid) and not-admissible (dashed) edges. Edge labeling follows the vertexes lexicographic order.}
\end{figure}
Since the sampling method can concentrate the probability mass at some points or along some directions of the hypercube we need to impose some admissibility conditions to have  non-degenerate distributions. The first set of conditions excludes degenerate segments, that are segments with equal end-points. Thus, we rule out self-loops from the graph, i.e. edges from one vertex to itself.
The second set of admissibility conditions excludes edges where the distribution concentrates along some coordinates. The edge $e^{*}_2$ joining vertex 1 to 3 is not admissible since:
\begin{eqnarray*}
U_1&=&  x_{11}V + x_{13}\left(1-V\right)=\gamma - \left(\gamma-\alpha\right)V \\
U_2&=&  x_{21}V + x_{23}\left(1-V\right)=\beta
\end{eqnarray*}
and $U_2$ is almost surely constant conditionally on being on that edge. A similar remark applies to the edges $e^{*}_3$ and $e^{*}_5$, whereas $e^{*}_4$ (red dashed in Figure \ref{fig:edgelabel}) is only admissible if $\beta<\gamma$.
\end{example}
In summary, the admissible edge set is:
\begin{eqnarray*}
\mathcal{E}&=&\left\{\begin{array}{ccc}\left\{e_1=e^{*}_1=\left(1,2\right),e_2=e^{*}_4=\left(2,3\right),e_3=e^{*}_6=\left(3,4\right)\right\},&& \mbox{ if } \beta<\gamma\\&&\\\left\{e_1=e^{*}_1=\left(1,2\right),e_2=e^{*}_6=\left(3,4\right)\right\},&& \mbox{ if } \beta=\gamma\end{array}\right.
\end{eqnarray*}
In our construction, conditional uniformity is a necessary condition for standard marginal uniformity and the admissibility assumption implies the conditional uniformity. To enhance the paper's readability, all proofs are deferred to the Appendix.
\begin{lemm}[Conditional Uniformity]\label{notparrallel} 
Let $\mathcal{S}$ satisfy Assumption \ref{assAdmiss}. Conditionally on being on the $k$-th segment of edge  $e_k=\left(i\left(k\right),j\left(k\right)\right)$,  for each $l\in \mathcal{D}$, the random variable $U_l$ in \eqref{simplevar} is uniform on $\left[\alpha_{l,k},\beta_{l,k}\right]$ with $\alpha_{l,k}=\min\left(x_{l,i\left(k\right)},x_{l,j\left(k\right)}\right)$ and $\beta_{l,k}=\max\left(x_{l,i\left(k\right)},x_{l,j\left(k\right)}\right)$.
\end{lemm}

Another requirement for our construction is that sampling points are in the unit hypercube which is guaranteed by the following range condition.
\begin{ass}[Range]\label{assRange}
The range requirement is satisfied if $\max\{x_{lk},k=1,\ldots,n\}=1$ and $\min\{x_{lk},k=1,\ldots,n\}=0$.
\end{ass}

Since $U_{l}$ is a convex combination of $x_{l,i}$, $i\in\left\{1,\ldots,n\right\}$, Assumption \ref{assRange} is needed in light of the requirement that $U_{l} \in \left[0,1\right]$ for each $l\in\mathcal{D}$.
  
In order to introduce the third assumption we need some notation. The set of admissible edges depends on the partitions induced by the distinct elements, sorted in ascending order, in the rows of $\mathbf{X}$. For each row $\left(x_{l,1},\ldots,x_{l,n}\right)$, 
of $\mathbf{X}$ with $l=1,\ldots,d$ we define $\mathbf{a}_l=\left(a_{l,1},a_{l,2},\ldots,a_{l,n_l-1},a_{l,n_l}\right)$ the sequence of $n_l\leq n$ sorted distinct elements. The unique values define a partition, say  $\left\{A_{l,m}\right\}^{n_l}_{m=2}$ of the unit interval with elements:
\begin{equation}
A_{l,m}=\left[a_{l,m-1},a_{l,m}\right),\quad m\in\left\{2,\ldots,n_l\right\}.\label{partition}
\end{equation}

For each unique value $a_{lk}$, the position set $\mathcal{M}_{l,k} =\left\{ i\in\left\{1,\ldots,n\right\}: x_{l,i}=a_{l,k}\right\}$, $k=1,\ldots,n_l$ denotes the set of points that provides the unique projected values of $\mathbf{X}$'s columns into the $l$-th dimension. For each row, $\mathcal{M}_{l,k}$ must therefore satisfy:
\begin{equation*}
\mathcal{M}_{l,m}\cap\mathcal{M}_{l,m^{\prime}}=\emptyset,\quad 
\bigcup^{n_l}_{m=1} \mathcal{M}_{l,m} = \left\{1,\ldots,n\right\},
\end{equation*}
and one can represent the coordinates of the vertex $\mathbf{x}_{k}$ by using the sets of positions of the unique values:
\begin{eqnarray*}
x_{l,k}= \sum^{n_l}_{m=1}a_{l,m}\mathbb{I}_{\mathcal{M}_{l,m}}\left(k\right),\, l\in\mathcal{D},
\end{eqnarray*}
which will be used in the following to write the CM constraint.


For illustration purposes, let us focus on the first coordinate of Example \ref{ex2by2} in the case $\beta<\gamma$. The sequence of $n_l=3$ sorted distinct elements of the first row is $\mathbf{a}_1=\left(a_{1,1}=\alpha,a_{1,2}=\beta,a_{1,3}=\gamma\right)$ and the sets of positions of the unique values are $\mathcal{M}_{1,1} = \left\{1,4\right\}$, $\mathcal{M}_{1,2} = \left\{2\right\}$, $\mathcal{M}_{
1,3}= \left\{3\right\}$ which is a partition of $\{1,2,3,4\}$.

We denote with $\mathcal{G}_l\equiv\left\{\mathcal{V}_l,\mathcal{E}_l\right\}$ the projection of $\mathcal{G}$ on the $l$-coordinate induced by the unique values 
$\mathbf{a}_l$. $\mathcal{G}_l$ is graph obtained by assigning a node to each of the $n_{l}$ components of $\mathbf{a}_l$ and defining the edge set $\mathcal{E}_l=\{e_{l,1},\ldots,e_{l,|\mathcal{E}|}\}$ through the map $\{1,\ldots,\left\vert\mathcal{E}\right\vert\}\mapsto \mathcal{E}_l$,  $k\in\{1,\ldots,\left\vert\mathcal{E}\right\vert\} \rightarrow e_{l,k}=\left(m\left(k\right),m^{\prime}\left(k\right)\right)\in\mathcal{E}_l$ where
\begin{eqnarray*}
m\left(k\right)&=&\left\{ m\in\left\{1,\ldots,n_l\right\}: i\left(k\right)\in \mathcal{M}_{l,m}\right\},\\
m^{\prime}\left(k\right)&=&\left\{ m^{\prime}\in\left\{1,\ldots,n_l\right\}: j\left(k\right)\in \mathcal{M}_{l,m^{\prime}}\right\}.
\end{eqnarray*} 
We call $n^l_{\left(m,m^{\prime}\right)}$ the edges multiplicity between the projected nodes $m$ and $m^{\prime}$, by convention if $\left(m,m^{\prime}\right)$ is not in $\mathcal{E}_{l}$ then $n^l_{\left(m,m^{\prime}\right)}=0$. 

\begin{rema}\label{remadmis}
In the proposed notation,  Assumption \ref{assAdmiss} corresponds to $m\left(k\right)\neq m^{\prime}\left(k\right)$ when $e_{l,k}=\left(m\left(k\right),m^{\prime}\left(k\right)\right)$ for all $k=1,\ldots,|\mathcal{E}|$, $l=1,\ldots,d$. 
\end{rema}

In Figure \ref{fig:edgeproj} we report the projection of the graph of Figure \ref{fig:edgelabel} onto the $l$-th  coordinates given in  Equation \eqref{excoord} of Example \ref{ex2by2}. 

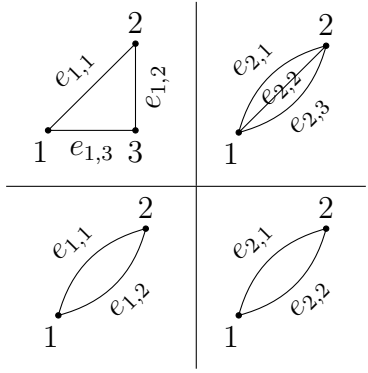
\begin{figure}[t]
\centering

\begin{tabular}{c|c}
\begin{tikzpicture}
\filldraw 
 (1,1.15)  circle (1pt) node[align=center,above] {$2$}--node[above,sloped] {$e_{1,1}$}
 (-0.15,0) ;
 \filldraw 
  (1,0) circle (1pt) node[align=center, below] {$3$}--node[below,sloped] {$e_{1,3}$}
  (-0.15,0) circle (1pt);
   \filldraw 
  (1,0) --node[below,sloped] {$e_{1,2}$}
   (1,1.15) ;
 \draw (-0.25,0)  node[align=center, below] {$1$};
\end{tikzpicture}
&
\begin{tikzpicture}
\filldraw 
 (1,1.15)  circle (1pt) node[align=center,above] {$2$}--node[sloped] {$e_{2,2}$}
 (-0.15,0) circle (1pt);
 \filldraw (1,1.15)  edge [bend right] node[above,sloped] {$e_{2,1}$}(-0.15,0);
  \filldraw (1,1.15)  edge [bend left] node[below,sloped] {$e_{2,3}$}(-0.15,0);
 \draw (-0.25,0)  node[align=center, below] {$1$};
\end{tikzpicture}\\\hline
\begin{tikzpicture}
\filldraw 
 (1,1.15)  circle (1pt) node[align=center,above] {$2$}
 (-0.15,0) circle (1pt);
 \filldraw (1,1.15)  edge [bend right] node[above,sloped] {$e_{1,1}$}(-0.15,0);
  \filldraw (1,1.15)  edge [bend left] node[below,sloped] {$e_{1,2}$}(-0.15,0);
 \draw (-0.25,0)  node[align=center, below] {$1$};
\end{tikzpicture}
&
\begin{tikzpicture}
\filldraw 
 (1,1.15)  circle (1pt) node[align=center,above] {$2$}
 (-0.15,0) circle (1pt);
 \filldraw (1,1.15)  edge [bend right] node[above,sloped] {$e_{2,1}$}(-0.15,0);
  \filldraw (1,1.15)  edge [bend left] node[below,sloped] {$e_{2,2}$}(-0.15,0);
 \draw (-0.25,0)  node[align=center, below] {$1$};
\end{tikzpicture}
\end{tabular}

\caption{ Graph projection for the graph in Example \ref{ex2by2} in the case $\beta<\gamma$ (top) and $\beta=\gamma$ (bottom), considering the first (left) or the second (right) coordinate.}\label{fig:edgeproj}
\end{figure}

Set $\alpha_{l,k}=a_{l,m_{\alpha}\left(k\right)}$ and $\beta_{l,k}=a_{l,m_{\beta}\left(k\right)}$, where $m_{\alpha}\left(k\right)$ and $m_{\beta}\left(k\right)$ are the minimum and the maximum value between $m\left(k\right)$ and $m^{\prime}\left(k\right)$, and define 
\begin{equation*}
\mathcal{K}_{l,m}\equiv\left\{k \in \left\{1,\ldots, \left\vert\mathcal{E}\right\vert\right\}: m_{\alpha}\left(k\right)+1\leq m\leq m_{\beta}\left(k\right)
\right\}
\end{equation*}
We state now the condition on the coordinates of the sampling construction for marginal standard uniformity.
\begin{ass}[Coordinate]\label{assCoord}

The following  set of $n_l-2$ equations in the variables $a_{l,m}$, $m=2,\ldots,n_l-1$ are satisfied 
\begin{eqnarray}\label{uniformseg}
F_{l,m}\left( \mathbf{a}_l\right)&=&\dfrac{1}{\left\vert\mathcal{E}\right\vert }\displaystyle \sum_{k\in \mathcal{K}_{l,m}}  \dfrac{1}{ a_{l,m_{\beta}\left(k\right)}-  a_{l,m_{\alpha}\left(k\right)}}-1= 0
\end{eqnarray} 
$m=2,\ldots,n_l-1$.
\end{ass}

The following example clarifies the relationship between Assumptions \ref{assRange}-\ref{assCoord} and standard uniformity.
\begin{example}
We discuss separately the two cases: $\beta=\gamma$ and $\beta<\gamma$ since they correspond to two different admissible edge sets.  If $\beta=\gamma$ 
the first component of $\mathbf{U}$ in the stochastic construction of Equation \eqref{simplevar} is: 
 \begin{eqnarray*}
 U_1=\left\{\begin{array}{ccc} \alpha + \left(\beta-\alpha\right)V &\hbox{if}& K=1\\\beta - \left(\beta-\alpha\right)V &\hbox{if}& K=2\end{array}\right. 
 \end{eqnarray*}
Since we are sampling uniformly on the edge set $\mathcal{E}$ it follows that: 
 \begin{eqnarray*}\mathbb{P}\left(K=1\right)=\mathbb{P}\left(K=2\right)=\dfrac{1}{\left\vert\mathcal{E}\right\vert}=\dfrac{1}{2}\end{eqnarray*}
and $U_1$ has marginal PDF
 \begin{eqnarray}
f\left(u_1\right) &=& \dfrac{1}{2}\dfrac{1}{\beta-\alpha} \mathbb{I}_{\left[\alpha,\beta\right]}\left(u_1\right)+\dfrac{1}{2}\dfrac{1}{\beta-\alpha} \mathbb{I}_{\left[\alpha,\beta\right]}\left(u_1\right)=\dfrac{1}{\beta-\alpha}\mathbb{I}_{\left[\alpha,\beta\right]}\left(u_1\right).
 \end{eqnarray}
We get a standard uniform random variable if $\alpha=0$ and $\beta=1$. By the same argument,  $U_2$ is a standard uniform random variable. In conclusion, by sampling on the two diagonals of the unit square with a mixture of an antithetic couple $\left(V,1-V\right)$ and a comonotonic couple $\left(V,V\right)$, the method can attain marginal standard uniformity (left plot in Figure \ref{antisegment1}).

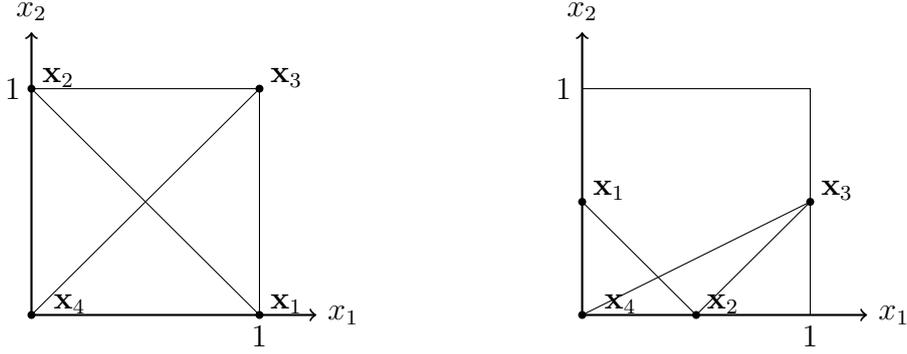
\begin{figure}[t]
\centering
\begin{tabular}{cc}
\begin{tikzpicture}[scale=1.5]
    \draw [<->,thick] (0,2.5) node (yaxis) [above] {$x_2$}
        |- (2.5,0) node (xaxis) [right] {$x_1$};
    \coordinate (c) at (2,2);
    \coordinate (A) at (0,2);
   \coordinate (B) at (2,0);
      \coordinate (C) at (2,2);
   \coordinate (D) at (0,0);

  \draw[black] (A) -- (B);
  \draw[black] (C) -- (D);
    \draw (yaxis |- c) node[left] {$1$}
        -| (xaxis -| c) node[below] {$1$};
   \fill[black] (A) circle (1pt);
      \fill[black] (C) circle (1pt);
   \fill[black] (D) circle (1pt);
   \fill[black] (B) circle (1pt);
   \draw[black] (A)++(0,0.1)      node[above,right] {$\mathbf{x}_2$};
     \draw[black] (B)++(0,0.1)      node[below,right] {$\mathbf{x}_1$};
    \draw[black] (C)++(0,0.1)      node[above,right] {$\mathbf{x}_3$};
     \draw[black] (D)++(0.1,0.1)      node[below,right] {$\mathbf{x}_4$};
\end{tikzpicture}
&\hspace{50pt}
\begin{tikzpicture}[scale=1.5]
    \draw [<->,thick] (0,2.5) node (yaxis) [above] {$x_2$}
        |- (2.5,0) node (xaxis) [right] {$x_1$};
    \coordinate (c) at (2,2);
    \coordinate (A) at (0,1);
   \coordinate (B) at (1,0);
      \coordinate (C) at (2,1);
   \coordinate (D) at (0,0);

  \draw[black] (A) -- (B);
  \draw[black] (C) -- (D);
  \draw[black] (B) -- (C);
    \draw (yaxis |- c) node[left] {$1$}
        -| (xaxis -| c) node[below] {$1$};
   \fill[black] (A) circle (1pt);
      \fill[black] (C) circle (1pt);
   \fill[black] (D) circle (1pt);
   \fill[black] (B) circle (1pt);
   \draw[black] (A)++(0,0.1)      node[above,right] {$\mathbf{x}_1$};
     \draw[black] (B)++(0,0.1)      node[below,right] {$\mathbf{x}_2$};
    \draw[black] (C)++(0,0.1)      node[above,right] {$\mathbf{x}_3$};
     \draw[black] (D)++(0.1,0.1)      node[below,right] {$\mathbf{x}_4$};
\end{tikzpicture}
\end{tabular}
\caption{Sampling on line segments. The support set of the 2 dimensional example with 4 vertexes. Cases: $\beta=\gamma$ (left) and $\beta<\gamma$ (right).}\label{antisegment1}
\end{figure}

If $\beta<\gamma$ the PDF of $U_1$ is:
\begin{eqnarray}\label{pdfexample}\nonumber
f\left(u_1\right) &=& 1/3\dfrac{1}{\beta-\alpha} \mathbb{I}_{\left[\alpha,\beta\right]}\left(u_1\right)+1/3\dfrac{1}{\gamma-\beta} \mathbb{I}_{\left[\beta,\gamma\right]}\left(u_1\right)+1/3\dfrac{1}{\gamma-\alpha} \mathbb{I}_{\left[\alpha,\gamma\right]}\left(u_1\right).
 \end{eqnarray}
A necessary and sufficient condition for $u_1\in \left[0,1\right]$ is  $\alpha=\min\{x_{1k}, k=1,\ldots,4\}=0$ and $\gamma=\max\{x_{1k},k=1,\ldots,4\}= 1$
which implies the PDF is piece-wise constant on the partition $\left[0,\beta\right)\subset[0,1]$ and $\left[\beta,1\right]\subset[0,1]$, induced by the unique values in the first row of $\mathbf{X}$, and is null on $\left[0,\beta\right)\cap\left[\beta,1\right]$. The value of $\beta$ such that the PDF takes the same value on the elements of the partition, i.e.:
 \begin{align}
\left\{
 \begin{array}{ccc}
 \dfrac{1}{\beta} + 1 =3 &\hbox{if}& u_1\in\left[0,\beta\right)\\
 \dfrac{1}{1-\beta} + 1 =3 &\hbox{if}& u_1\in\left[\beta,1\right]\\
 \end{array}
\right.
 \end{align}
is $\beta=1/2$ which implies $U_1\sim \mathcal{U}[0,1]$ and $U_2\sim \mathcal{U}[0,1/2]$. Thus, for $\beta<\gamma$ our construction is not standard uniform along all coordinates of the vector. The right plot of Figure \ref{antisegment1} shows that the range of $U_1$ is $[0,1]$ (horizontal axis) and the range of $U_2$ is $[0,1/2]$ (vertical axis) \footnote{Similar arguments can be applied to show that imposing standard uniformity for $U_2$ requires $\min\{x_{1k},k=1,\ldots,4\}=\alpha=0$  and  $\max\{x_{1k},k=1,\ldots,4\}=\beta=1$ which is not satisfied since $\beta<\gamma$.}
\end{example}

We are ready to state the main result of this section which proves the marginal standard uniformity of our construction. In addition, we provide a method to find $\mathbf{a}_l$ satisfying the condition in
Assumption \ref{assCoord}.

\begin{theo}[Marginal Standard Uniformity]\label{unif} 
Under Assumptions \ref{assAdmiss}-\ref{assCoord},  each  coordinate of the random vector $\mathbf{U}=(U_1,\ldots,U_d)$ in the stochastic representation \eqref{simplevar} has a $\ru(0,1)$ marginal distribution.
\end{theo}

\begin{theo}\label{theoopt} The system in \eqref{uniformseg} is equivalent to the following convex minimization problem 
\begin{eqnarray}\label{optuni}
\displaystyle\min_{\left(a_{l,2},\ldots,a_{l,n_{l}-1}\right)\in \left[0,1\right]^{n_l-2}}\Psi_l\left(\mathbf{a}_l\right),
\end{eqnarray}
with
\begin{eqnarray}\label{Psi}\nonumber
 \Psi_l\left(\mathbf{a}_l\right)  &=&-\dfrac{1}{2 \left\vert\mathcal{E}\right\vert}\displaystyle \sum^{n_l}_{ m,m^{\prime}=1}n^l_{\left(m,m^{\prime}\right)} \log\left\vert a_{l,m^{\prime}}-  a_{l,m}\right\vert.
\end{eqnarray}
\end{theo}
Since the optimization problem is convex, if a solution exists it is a global minimum.
In addition, since $\Psi_l\left(\mathbf{a}_l\right)$ is a sum of lower semi-continuous functions, it is lower semi-continuous. This, together with the compactness of the unit hypercube, guarantees the existence of a solution, by the lower version of Weierstrass theorem (see for example Theorem 2.43 in \cite{aliprantis2007infinite}).  

In the next theorem we show that our construction satisfies an optimality criteria involving  the Kullback–-Leibler (KL) divergence from the uniform distribution. We remind the reader that the KL divergence of the probability measure $\mathbb{P}$ with respect to the probability measure $\mathbb{Q}$ is 
\begin{eqnarray*}
D_{KL}\left(\mathbb{P}\vert\vert \mathbb{Q}\right)&=& \mathbb{E}_{\mathbb{P}}\left[\log\left(\dfrac{d\mathbb{P}}{d\mathbb{Q}}\right)\right].
\end{eqnarray*}
In our case  $\mathbb{P}$ is the joint measure of $U_l$ and $K$ given by the stochastic representation \eqref{simplevar}:
\begin{eqnarray*}
d\mathbb{P}\left( u_l,k\right)&=& \dfrac{1}{\left\vert\mathcal{E}\right\vert}f\left(u_l\right\vert \left.K=k\right)  du_l
\end{eqnarray*}
$\mathbb{Q}$ is the joint uniform independent measure in $U_l$ and $K$:
\begin{eqnarray*}
d\mathbb{Q}\left( u_l,k\right)&=&\dfrac{1}{\left\vert\mathcal{E}\right\vert} \mathbb{I}_{ \left[0,1\right]}\left(u_l\right)du_l\\
\end{eqnarray*}

\begin{theo}\label{KL} The minimization problem in \eqref{optuni} is equivalent to minimizing the Kullback-Leibler divergence of the joint measure of $U_l$ and $K$ given by the stochastic representation \eqref{simplevar} from a joint uniform independent measure in $U_l$ and $K$.
\end{theo}

The results in Theorems \ref{unif}-\ref{KL} allow us to find the $l$-th row $\mathbf{a}_l$, of the coordinates matrix $X$ such that $U_l$ in representation \eqref{simplevar} is standard uniform. Thus,  we have $d$ minimization problems each with a different  number of variables, $n_l$, where the dependence structure in $\mathbf{U}$ is  given by the graph $\mathcal{G}$ of the segments in $\mathcal{S}$ and marginally encoded by $\mathbf{a}_l$ and a projected graphs $\mathcal{G}_l$.
\begin{definition}\label{def:optuncostr}
Let $\mathcal{S}=(\mathcal{G},\mathbf{X})$ be the segment set with $\mathcal{G} = \left\{\mathcal{V},\mathcal{E}\right\}$. The convex minimization problem
\begin{eqnarray*}
\displaystyle\min_{\left\{a_{l,m},m=1,\ldots,n_l,\,l=1,\ldots,d\right\}\in \left[0,1\right]^{n_1+\ldots+n_d}}\sum^d_{l=1}\Psi_l\left(\mathbf{a}_l\right),
\end{eqnarray*}
is called standard uniform on $\mathcal{S}$ problem.
\end{definition}

\subsection{Strict Countermonotonicity on Segments}
Limiting the study to strict $d$-CTM leads to an unique value for the constant $k$ in \eqref{lcounter}:
\begin{eqnarray*}
k=\mathbb{E}\left[\sum^d_{j=1}U_j\right]=\sum^d_{j=1}\mathbb{E}\left[U_j\right]=\dfrac{d}{2}
\end{eqnarray*}
 The constant sum condition can be written as a linear restriction on the coordinates of the vertexes $\mathbf{x}_k$, that is
\begin{eqnarray*}\label{CMcond}
\sum^{d}_{l=1}U_l&=&\sum^{d}_{l=1} x_{lJ} +V\left[\sum^{d}_{l=1}x_{lI}-\sum^{d}_{l=1}x_{lJ}\right]=\dfrac{d}{2}
\end{eqnarray*}\\ and since the previous relationship should be valid for all $V$ and $\left(I,J\right)$ (i.e. for all $W$ in our setting) we obtain the condition that all vertexes should be in the hyperplane of constant sum, i.e. 
\begin{eqnarray}\begin{array}{ccc}\label{segdcm}
\displaystyle\sum^{d}_{l=1} x_{lk} =\sum^{d}_{l=1}\sum^{n_l}_{m=1}a_{l,m}\mathbb{I}_{\mathcal{M}_{l,m}}\left(k\right)=\dfrac{d}{2}&& k =1,\ldots,n.\end{array}
\end{eqnarray}

The convexity of the minimization problem in Eq. \eqref{def:optuncostr} is not altered by the inclusion of a linear constraint, nevertheless the constraint couples the coordinates and yields the following non-separable optimization problem in $n_d=\sum^d_{l=1}n_l$ variables.

\begin{definition}\label{mixonseg}
Let $\mathcal{S}=(\mathcal{G},\mathbf{X})$ be the segment set with $\mathcal{G} = \left\{\mathcal{V},\mathcal{E}\right\}$. The convex minimization problem:
\begin{eqnarray*}
\displaystyle\min_{\left\{a_{l,m},m=1,\ldots,n_l,\,l=1,\ldots,d\right\}\in \left[0,1\right]^{n_1+\ldots+n_d}}\sum^d_{l=1}\Psi_l\left(\mathbf{a}_l\right)
\end{eqnarray*}
subject to:
\begin{eqnarray*}
\sum^{d}_{l=1}\sum^{n_l}_{m=1}a_{l,m}\mathbb{I}_{\mathcal{M}_{l,m}}\left(k\right)=\dfrac{d}{2},&& k =1,\ldots,n
\end{eqnarray*}
will be referred as the strict $d$-CTM on $\mathcal{S}$ problem.
\end{definition}

Since all the constraints are affine the problem in the above definition represents an ordinary convex problem in the terminology of \cite{rockafellar1970convex} and can be solved using the method of Lagrange multipliers. Local minima are also global if they exist. Finally, the existence of a solution is guaranteed by lower semi-continuity of $\sum^d_{l=1}\Psi_l\left(\mathbf{a}_l\right)$ and the fact that the intersection of the unit hypercube with the hyperplane of constant sum is compact.

\subsection{Antithetic Vector Transformations}
We introduce some transformations of the random vector of representation \eqref{simplevar} which are useful to study the properties of our sampling on line segments. These transformations preserve uniformity and provide a general framework to the construction of antithetic vectors. The use of transformations to span the set of antithetic vectors was pioneered by \cite{Andreasson1972} and \cite{AndreassonDahlquist1972} and discussed further by \cite{Roach1977}. 
First, we define deterministic compositions and present reflections as a particular case. We provide sufficient conditions for preserving strict countermonotoniticy under deterministic compositions. Secondly, we introduce stochastic compositions as particular instances of probabilistic mixtures that preserve uniformity and the strict countermonotonicity property. In addition, we show that permutations preserve strict countermonotonicity. Finally, we  generalize of our stochastic representation \eqref{simplevar} inspired by the deterministic composition.

The notion of deterministic composition relies on the interpretation of the components of the stochastic representation \eqref{simplevar} as a deterministic map on the unit interval. Let $\mathbf{X}$ be a solution to the standard uniform on $\mathcal{S}^{(x)}=(\mathcal{G}^{(x)},\mathbf{X})$  with edge set $\mathcal{E}^{(x)}$ and let $T_{x,l,k} : \left[0,1\right] \mapsto  \left[0,1\right]$
\begin{equation}
     T_{x,l,k}\left(u\right)= x_{li(k)} u + x_{lj(k)} \left(1-u\right). 
\end{equation}
Then each element of \eqref{simplevar} can be written as
\begin{equation}
U_l= T_{x,l,K}\left(V\right)
\end{equation}
with $U_l$, $K$, and $V$ as in \eqref{simplevar}.
Let $\mathbf{Y}$ be  a different solution to the standard uniform $\mathcal{S}^{(y)}=(\mathcal{G}^{(y)},\mathbf{Y})$  with edge set $\mathcal{E}^{(y)}$. It induces a map $T_{y,l,k} : \left[0,1\right] \mapsto  \left[0,1\right]$
\begin{equation}
     T_{y,l,k}\left(u\right)= y_{li(k)} u + y_{lj(k)} \left(1-u\right). 
\end{equation}
The deterministic composition is defined as $T_{y,l,k^{\prime}}\circ T_{x,l,k} \left(u\right)=T_{y,l,k}\left(T_{x,l,k}(u)\right)$. The following lemma states that the random vector with components $T_{y,l,k^{\prime}}\circ T_{x,l,k}\left(U_l\right)$ can be represented as in \eqref{simplevar} and provides the set of segments  associated with the composition.

\begin{lemm}\label{detcompunif}
The deterministic composition $T_{z,l,k}\left(V\right)=T_{y,l,k^{\prime}}\circ T_{x,l,k} \left(V\right)$, $l\in\mathcal{D}$ is uniform on  the segments $\mathcal{S}'=(\mathcal{G}',\mathbf{Z})$ with graph $\mathcal{G}'=(\mathcal{V}',\mathcal{E}')$, vertex set $\mathcal{V}'=\mathcal{E}^{(y)}\times\mathcal{V}^{(x)}$, and edge set  $\mathcal{E}=\mathcal{E}^{(y)}\times\mathcal{E}^{(x)}$. The coordinate matrix $\mathbf{Z}$ has elements 
\begin{eqnarray*}
z_{l,m}=\left( y_{li\left(k\right)}x_{li^{\prime}} + y_{lj\left(k\right)}\left(1-x_{li^{\prime}}\right)\right),
\end{eqnarray*}
where the index $m\in \mathcal{V}'$ corresponds to the pair $\left(k,i^{\prime}\right)$ with $k\in\left\{1,\ldots,\left\vert \mathcal{E}^{(y)}\right\vert\right\}$ and $i^{\prime}\in \{1,\ldots,|\mathcal{V}^{(x)}|\}$ sorted in lexicographic order.
\end{lemm}
We provide a sufficient condition for the deterministic composition to satisfy the constant sum constraint and leave the necessary conditions for future research.
\begin{lemm}\label{CTMcomp}
 Let $\mathbf{X}$ be a solution of the strict countermonotonic problem on $\mathcal{S}^{(x)}=(\mathcal{G}^{(x)},\mathbf{X})$ 
and $\mathbf{Y}$ a solution of the standard uniform problem on $\mathcal{S}^{(y)}=(\mathcal{G}^{(y)},\mathbf{Y})$ with the additional constraints that for every $l\in\mathcal{D}$ and $k\in\left\{1,\ldots,\mathcal{E}^{(y)}\right\}$
\begin{eqnarray}\label{c1c2}\nonumber
|y_{li\left(k\right)}-y_{lj\left(k\right)}|= c_1,\quad \sum^{d}_{l=1}y_{lj\left(k\right)}=c_2.
\end{eqnarray}
Let $W_l = T_{y,l,k^{\prime}}\circ T_{x,l,k} \left(V\right)$, $l\in\mathcal{D}$
then $\sum^{d}_{l=1} W_l= d/2$ if and only if 
\begin{equation}
\left(1-c_1\right) \dfrac{d}{2}=c_2.
\end{equation}

\end{lemm}
\begin{rema}
 $\mathbf{Y}$ in Lemma \ref{CTMcomp} cannot be a  strict countermonotonic solution on $\mathcal{S}^{(y)}=(\mathcal{G}^{(y)},\mathbf{Y})$, because this implies $c_2= \dfrac{d}{2}$ and $c_1= 2$ while the maximum value of $c_1$ is $1$ for all solutions of standard uniform problems by the Assumption \ref{assRange}.
\end{rema}

Let $I_d$ denote the identity matrix and $\mathbf{e}_l$ its $l$-th column. Reflections are defined in the following.
\begin{definition}[Reflections]\label{def:reflection}
Let $\mathbf{U}\in [0,1]^{d}$ be a random vector, the transformed vector $\mathbf{U}\mapsto \mathbf{W}=R_{l,\frac{1}{2}}\left(\mathbf{U}\right)$ with $R_{l,\frac{1}{2}}\left(\mathbf{U}\right)= \left[I_{d}-2\left(\mathbf{e}_l\mathbf{e}^T_l\right)\right]\mathbf{U}+\mathbf{e}_l$, defines a reflection of $\mathbf{U}$ through an hyperplane with $l$-coordinate equal to 1/2.
Given a index subset $\mathcal{L}\subseteq \mathcal{D}$, the sequential reflection transformation $R_{\mathcal{L},\frac{1}{2}}(\mathbf{U})=[I_{d}-2(\sum_{l\in\mathcal{L}}\mathbf{e}_l\mathbf{e}^T_l)]\mathbf{U}-\sum_{l\in\mathcal{L}}\mathbf{e}_l$ is the transformation obtained by sequentially reflecting $\mathbf{U}$ for a sequence of indexes $l\in\mathcal{L}$ through a hyperplane with $l$-coordinate equal to 1/2. $R_{\mathcal{D},\frac{1}{2}}(\mathbf{U})= (I_{d}-2I_{d})\mathbf{U}+\mathbf{1}_d=(1-U_1,\ldots,1-U_l,\ldots,1-U_d)$ is the central inversion through the center of the unit hypercube.
\end{definition}


If $\mathbf{U}$ is uniform on segments, the reflection can be interpreted as a deterministic composition. Let $l\in \mathcal{D}$ and
\begin{eqnarray*}
y_{l,1}=\left\{\begin{array}{ccc}
 1 & \hbox{if} & l\in \mathcal{L}\\
0 & \hbox{if} &  \hbox{otherwise}
\end{array}\right.\\
y_{l,2}=\left\{\begin{array}{ccc}
 0 & \hbox{if} &  l\in \mathcal{L} \\
1 & \hbox{if} &  \hbox{otherwise}
\end{array}\right.
\end{eqnarray*}
Using the definition of reflection we can write the elements of $\mathbf{W}=R_{\mathcal{L},\frac{1}{2}}(\mathbf{U})$ as:
\begin{eqnarray*}
W_{l} = y_{l,1} U_{l} + y_{l,2} \left(1-U_{l}\right).
\end{eqnarray*}
If $\mathbf{U}$ is strict $d$-CTM on segments and since $c_1=1$ and $c_2=d-\left\vert \mathcal{L}\right\vert$,  according to Lemma \ref{CTMcomp}  $\mathbf{W}$ is $d$-CTM only if $\mathcal{L}= \mathcal{D}$.

The following notion of stochastic composition augments the class of  deterministic compositions and it is useful for constructing $d$-CTM vectors on segments.

\begin{lemm}[Stochastic Composition]\label{stochcomp}
Let $\mathbf{U}^{t}$ , $t=1,2$ be strict $d$-CTM vectors on $\mathcal{S}^{t}=(\mathcal{G}^{t},\mathbf{X}^{t})$, with $\mathcal{E}^t$ the edge set of $\mathcal{G}_t$, the random vector
\begin{equation*}
\mathbf{W} =\left\lbrace \begin{array}{ccc} 
\mathbf{U}^{1}, &\hbox{with probability}&
\dfrac{\left\vert\mathcal{E}^1\right\vert}{\left\vert\mathcal{E}^2\sqcup \mathcal{E}^1\right\vert}\\
\mathbf{U}^{2}, &\hbox{with probability}&
\dfrac{\left\vert\mathcal{E}^2\right\vert}{\left\vert\mathcal{E}^1\sqcup \mathcal{E}^2\right\vert}
\end{array}\right. ,    
\end{equation*}
defines a stochastic composition of $U^1$ and $U^2$ and is strict $d$-CTM on $\mathcal{S}^{\prime}=(\mathcal{G}^1\sqcup \mathcal{G}^2,\left(\mathbf{X}^1,\mathbf{X}^2\right))$, where $\sqcup$ denotes the disjoint union of graphs.
\end{lemm}

Another transformation used later on in this paper to obtain exchangeable random vectors is the following.
\begin{lemm}[Permutation]\label{permdctm}
Let $\mathbf{U}\in [0,1]^{d}$ be a random vector and $\left\lbrace \pi\left(1\right)\right.,$  $\ldots,$ $\left. \pi\left(d\right)\right\rbrace$  a permutation of $\mathcal{D}$ with $P_{\pi}= \left(\mathbf{e}_{\pi\left(1\right)},\ldots,\mathbf{e}_{\pi\left(d\right)}\right)^{T}$ the corresponding row-permutation matrix. $\mathbf{W}$ is a permutation of  $\mathbf{U}$ if $\mathbf{W}=P_{\pi}\mathbf{U}$. If $\mathbf{U}$ is strict $d$-CTM on $\mathcal{S}=(\mathcal{G},\mathbf{X})$ then $\mathbf{W}$ is strict $d$-CTM on $\mathcal{S}^{\prime}=(\mathcal{G},P_{\pi}\mathbf{X})$. 
\end{lemm}

\begin{rema}
From the properties of transformations in Lemma  \ref{stochcomp} and \ref{permdctm}  it follows that the stochastic combination of all the permutations of a strict countermonotonic vector is strict countermonotonic and exchangeable. 
\end{rema}

Using the inner composition we can rephrase standard uniform on segments as the deterministic composition of the maps induced by collection of segments $\mathcal{S}$ with the comonotonic random vector $U_l=V$ for each $l\in \mathcal{D}$ which is a uniform vector on the principal diagonal of the hypercube. In particular, it can be shown that our construction includes the multivariate shuffles of min  \citep{mikusinski2010some}.  This new interpretation of Eq. \eqref{simplevar} allows for a generalization to any random vectors with marginal standard uniforms. 
Given a $d$-dimensional $\mathbf{V}\sim F$ with $V_l\sim \mathcal{U}[0,1]$ for $l = 1,\ldots,d$ and the stochastic representation 
\begin{eqnarray}\label{gensimplevar}
U_1&=& x_{1I}V_1  + x_{1J}\left(1-V_1\right)\nonumber\\
&\vdots&\\\nonumber
U_d&=& x_{dI} V_d + x_{dJ}\left(1-V_d\right),
\end{eqnarray}
we have the following corollary  to Theorem \ref{unif}. 

\begin{coro}\label{genunif}
Under Assumptions \ref{assAdmiss}-\ref{assCoord} the $l$-th component of the random vector $\mathbf{U}=(U_1,\ldots,U_d)$ in the stochastic representation \eqref{gensimplevar} is a standard uniform random variable.
\end{coro}

\section{Distributional Properties}\label{sec:distribution}
\subsection{Distribution}
The joint cumulative distribution function of the stochastic representation $\mathbf{U}$ of Eq. \eqref{gensimplevar} can be written using the distributions of the reflections of $\mathbf{V}$. Given a index subset $\mathcal{L}$, we denote the distribution of the reflection  $R_{\mathcal{L},\frac{1}{2}}\left(\mathbf{V}\right)$ with:
\begin{eqnarray*}
F_{\mathbf{V},\mathcal{L}}\left(\mathbf{u}\right)= \mathbb{P}\left(R_{\mathcal{L},\frac{1}{2}}\left(\mathbf{V}\right)\leq \mathbf{u}\right).
\end{eqnarray*} 
The marginal distribution of $U_l$ conditionally on living on an edge $e_k$ is $\mathcal{U}[\alpha_{l,k},\beta_{l,k}]$ with CDF
\begin{equation}
F_{U_l\vert K}\left(u_l;k\right)=\frac{\max\{\alpha_{l,k},\min\{\beta_{l,k},u_l\}\}-\alpha_{l,k}}{\beta_{l,k}-\alpha_{l,k}}
\end{equation}
where $\alpha_{l,k}$, $\beta_{l,k}$ are defined in Lemma \ref{notparrallel}. For each $k=1,\ldots, n$ define the sets:
\begin{eqnarray*}
\mathcal{L}_{k}^{+} &=& \left\{l\in\left\{1,\ldots,d\right\}: x_{li\left(k\right)}-x_{lj\left(k\right)}\geq 0\right\}\\ \mathcal{L}_{k}^{-} &=& \left\{l\in\left\{1,\ldots,d\right\}: x_{li\left(k\right)}-x_{lj\left(k\right)}< 0\right\}.
\end{eqnarray*}

\begin{theo}\label{th:condjoint}
$\mathbf{U}=\left(U_1 ,\ldots,U_d\right)$ in representation \eqref{gensimplevar} conditioning on $K=k$ has cumulative distribution function $F_{\mathbf{U}\vert K}\left(
u_{1},\ldots,u_{d};k\right)$ given by
\begin{eqnarray*}\label{copula}\nonumber
\mathbb{P}\left(U_1 \leq u_1,\ldots,U_d\leq u_d\left\vert K=k \right.\right)
=F_{\mathbf{V},\mathcal{L}^{-}_{k}}\left(\mathbf{v}_k\right)
\end{eqnarray*}
where $\mathbf{v}_k=(v_{1,k},\ldots,v_{d,k})$ with $v_{l,k}=F_{U_l\left\vert K\right.}\left(u_l;k\right)$.
\end{theo}

The following result comes from summing over all possible values of $K$.
\begin{coro}\label{distr} The random vector 
$\left(U_1 ,\ldots,U_d\right)$ in representation \eqref{gensimplevar},  has distribution $F_{\mathbf{U}}\left(
u_{1},\ldots,u_{d}\right)$ given by
\begin{eqnarray*}\label{copula}\nonumber
\mathbb{P}\left(U_1 \leq u_1,\ldots,U_d\leq u_d\right)=\displaystyle\dfrac{1}{\left\vert \mathcal{E} \right\vert}\sum^{\left\vert \mathcal{E} \right\vert}_{k=1}F_{\mathbf{V},\mathcal{L}^{-}_{k}}\left(\mathbf{v}_k\right).
\end{eqnarray*}
\end{coro}

For random vectors in representation \eqref{simplevar} we are able 
to derive a closed form expression of the conditional distribution of $\mathbf{U}$.

\begin{coro}\label{th:permutecond} $\mathbf{U}=\left(U_1 ,\ldots,U_d\right)$ in representation \eqref{simplevar} conditioning on $K=k$ has  cumulative distribution function
\begin{eqnarray*}\label{copula}\nonumber
&&F_{\mathbf{U}\vert K}\left(
u_{1},\ldots,u_{d};k\right)=\max\left( v^{+}_{k}  + v^{-}_{k} -1 , 0 \right)
\end{eqnarray*}
where $v^{+}_{k}=\min\{v_{l,k}, l\in\mathcal{L}_{k}^{+}  \}$, $v^{-}_{k}=\min\{v_{l,k}, l\in\mathcal{L}_{k}^{-}\}$
and $v_{l,k}=F_{U_l\left\vert K\right.}\left(u_l;k\right)$.
The pairs of variables $U_l$ and $U_{l^{\prime}}$ $l\neq l^{\prime}$ have cumulative distribution function:
\begin{eqnarray}\resizebox{0.9\linewidth}{!}{$
F_{U_{l},U_{l^{\prime}}\left\vert K\right.}\left(u_l,u_{l^{\prime}};k\right)= \left\{\begin{array}{ccc}\min\left(F_{U_{l}
\left\vert K\right.}\left(u_l;k\right),F_{U_{l^{\prime}}\left\vert K\right.}\left(u_{l^{\prime}};k\right)\right)& \hbox{if}& l,l^{\prime}\in\mathcal{L}^{+}_k \\&&\\\min\left(F_{U_{l}
\left\vert K\right.}\left(u_l;k\right),F_{U_{l^{\prime}}\left\vert K\right.}\left(u_{l^{\prime}};k\right)\right)& \hbox{if}&  l,l^{\prime}\in\mathcal{L}^{-}_k\\&&\\\max\left(F_{U_{l}
\left\vert K\right.}\left(u_l;k\right)+F_{U_{l^{\prime}}\left\vert K\right.}\left(u_{l^{\prime}};k\right) -1,0\right)& \hbox{if}& l\in\mathcal{L}^{-}_k,l^{\prime}\in\mathcal{L}^{+}_k  \\&&\\\max\left(F_{U_{l}
\left\vert K\right.}\left(u_l;k\right)+F_{U_{l^{\prime}}\left\vert K\right.}\left(u_{l^{\prime}};k\right) -1,0\right)& \hbox{if}&  l\in\mathcal{L}^{+}_k,l^{\prime}\in\mathcal{L}^{-}_k \end{array}\right.$}
\end{eqnarray}
\end{coro}

From Corollary \ref{th:permutecond} and following the definition of Fr\'{e}chet bound \citep{Frechet1951}, the elements $U_l$ and $U_{l^{\prime}}$ of the stochastic representation \eqref{simplevar} are monotonic in the same direction if $l,l^{\prime}\in\mathcal{L}^{+}_k$  or $l,l^{\prime}\in\mathcal{L}^{-}_k$  and antithetic if $l\in\mathcal{L}^{-}_k ,l^{\prime}\in\mathcal{L}^{+}_k$ or $l\in\mathcal{L}^{+}_k ,l^{\prime}\in\mathcal{L}^{-}_k$, conditionally on living on the $k$-th segment.
\subsection{Multivariate Kendall' $\tau$ and Spearman's $\rho$ }\label{sec:mrho}
Multivariate Kendall' $\tau$  and Spearman's $\rho$ are  multivariate measures of concordance introduced in \cite{joe1992} as a generalization of the well known bivariate measures. \cite{fuchs2018characterizations,ahn2020minimal} show that  $d$-CTM vectors  have the same minimal multivariate Kendall's $\tau$ but different Spearman's $\rho$.  Then, Spearman's $\rho$ can be used to rank $d$-CTM vectors in concordance order.

Let $\mathbf{U}$ and $\mathbf{W}$ be independent random vectors with the same distribution $F_{\mathbf{U}}=F_{\mathbf{W}}=G$. The multivariate Kendall's $\tau$ is defined as:

\begin{eqnarray*}
\tau\left(F_{\mathbf{U}}\right)&=& \dfrac{2^{d}}{2^{d-1}-1} \left[\int_{\left[0,1\right]^d}F_{\mathbf{U}}\left(\mathbf{u}\right)dF_{\mathbf{U}}\left(\mathbf{u}\right)-\dfrac{1}{2^{d}}\right]
\\&=& \dfrac{2^{d}}{2^{d-1}-1} \left[\mathbb{E}\left[G\left(\mathbf{W}\right)\right]-\dfrac{1}{2^{d}}\right]
\\&=& \dfrac{2^{d}}{2^{d-1}-1} \left[\mathbb{P}\left(\mathbf{U}\leq\mathbf{W}\right)-\dfrac{1}{2^{d}}\right].
\end{eqnarray*}
\cite{fuchs2018characterizations} show that if $\mathbf{U}$ and $\mathbf{W}$ are $d$-CTM then $\tau\left(F_{\mathbf{U}}\right)$ attains its minimal value 
\begin{eqnarray}\label{taumin}
\tau\left(F_{\mathbf{U}}\right)=\tau_{\min}=-\dfrac{1}{2^{d-1}-1}.
\end{eqnarray}
We provide an analytical expression of the Kendall's $\tau$ for vectors with the stochastic representation of Eq. \ref{gensimplevar}.
\begin{prop}\label{prop:taubound}
Let $\mathcal{S}=(\mathcal{G},\mathbf{X})$ be the segment set with $\mathcal{G} = \left\{\mathcal{V},\mathcal{E}\right\}$. Let $\mathbf{U}$ and $\mathbf{W}$ two independent copies of a random vectors with the generalized line segments representation \eqref{gensimplevar} on $\mathcal{S}$ , then: 
\begin{eqnarray}\label{rhosegment}
\tau\left(F_{\mathbf{U}}\right)
 &=& \dfrac{2^{d}}{2^{d-1}-1} \left[\dfrac{1}{\left\vert\mathcal{E}\right\vert}\sum^{\left\vert\mathcal{E}\right\vert}_{k_{\mathbf{U}}=1} \mathbb{P}\left(R_{\mathbf{V},\mathcal{L}^{-}_{k_{\mathbf{U}}}}\left(\mathbf{V}\right)\leq\mathbf{Y}_{k_{\mathbf{U}}}\right)-\dfrac{1}{2^{d}}\right],
\end{eqnarray}
with $\mathbf{Y}_{k_{\mathbf{U}}}=\left(F_{U_1\left\vert K_{\mathbf{U}}\right.}\left(W_1;k_{\mathbf{U}}\right),\ldots, F_{U_d\left\vert K_{\mathbf{U}}\right.}\left(W_d;k_{\mathbf{U}}\right)\right)$.

\end{prop}

  We use the definition of Spearman's $\rho$ in \cite{joe1992}\footnote{The Spearman's $\rho$ following the definition in \cite{ahn2020minimal} is obtained as: $\frac{\rho\left(F_U\right)+\rho\left(F_W\right)}{2}$, with $\mathbf{W}= R_{\mathcal{D},\frac{1}{2}}\left(\mathbf{U}\right)$.} for distributions with standard uniform marginals. Let $F\left(\mathbf{U}\right)$ be the cumulative distribution function of the random vector $\mathbf{U}$; then the multivariate Spearman's $\rho$ is

\begin{eqnarray*}
\rho\left(F_{\mathbf{U}}\right)&=& \dfrac{2^d\left(d+1\right)}{2^d-\left(d+1\right)}\left(\int_{\left[0,1\right]^d}F_{\mathbf{U}}\left(\mathbf{u}\right)d\mathbf{u} -\dfrac{1}{2^d}\right)
\\&=& \dfrac{2^d\left(d+1\right)}{2^d-\left(d+1\right)}\left(\int_{\left[0,1\right]^d}\prod^{d}_{l=1}u_L dF_{\mathbf{U}}\left(\mathbf{u}\right)  -\dfrac{1}{2^d}\right)
\\&=& \dfrac{2^d\left(d+1\right)}{2^d-\left(d+1\right)}\left(\mathbb{E}\left[\prod^{d}_{l=1}U_l\right] -\dfrac{1}{2^d}\right).
\end{eqnarray*}

The attainable lower bound $\rho_{\min}$ can be computed using the lower bound for $\mathbb{E}\left[\prod^{d}_{l=1}U_l\right] $ given in Corollary 4.1 of \cite{Wang2011}. We report the values in Table \ref{tab:rhomin}. 

\begin{table}[htbp]
  \centering
  \caption{Minimum values of multivariate Spearman's $\rho$, using the lower bound in Corollary 4.1 of \cite{Wang2011}}\label{tab:rhomin}
    \begin{tabular}{|c|c|c|c|c|c|c|c|c|}
    \hline
  $d$ &  2     &  3     & 4     & 5     & 10    & 20    & 50    & 100   \\
  \hline
  $\rho_{\min}$&    -1    & -0.56  & -0.32  & -0.18  & -0.01 & -1.99$\cdot10^{-5}$  & -4.53$\cdot10^{-14}$ &-7.97$\cdot10^{-29}$  \\
   \hline
    \end{tabular}%
  \label{tab:addlabel}%
\end{table}%

\begin{prop}\label{prop:rhobound}
Let $\mathbf{U}$ have the generalized representation \eqref{gensimplevar},  $\mathbf{V}$ be reflection invariant, and $\xi^{*}\leq 1$. Then: 
\begin{itemize}
\item[(a)] \begin{eqnarray}\label{rhosegment}
\rho\left(F_{\mathbf{U}}\right)
 &=& \resizebox{0.7\linewidth}{!}{$\displaystyle\sum^{d}_{m=0}\sum_{\tiny\begin{array}{c}\mathcal{L}_m\subseteq\mathcal{D}\\ \left\vert \mathcal{L}_m\right\vert=m\end{array}}\xi_{\mathcal{L}_m}\rho\left(F_{\mathbf{V},\mathcal{D}\setminus\mathcal{L}_m}\right)
+\dfrac{2^d\left(d+1\right)}{2^d-\left(d+1\right)}\dfrac{1}{2^d}\left(\xi^{*} -1 \right)\label{rhobound},$}
\end{eqnarray}
with
\begin{eqnarray*}
\xi_{\mathcal{L}_m}&=&\dfrac{1}{\left\vert \mathcal{E}\right\vert}\sum^{\left\vert \mathcal{E}\right\vert}_{k=1}\left(\prod_{l\in \mathcal{L}_m}x_{l,i\left(k\right)}  \prod_{l\in \mathcal{D}\setminus\mathcal{L}_m}x_{l,j\left(k\right)}\right)\\
\xi^{*}&=&\sum^{d}_{m=0}\sum_{\tiny\begin{array}{c}\mathcal{L}_m\subseteq\mathcal{D}\\ \left\vert \mathcal{L}_m\right\vert=m\end{array}} \xi_{\mathcal{L}_m} =\dfrac{1}{\left\vert \mathcal{E}\right\vert}\sum^{\left\vert \mathcal{E}\right\vert}_{k=1} \prod^{d}_{l=1} \left(x_{l,i\left(k\right)} + x_{l,j\left(k\right)} \right) ;\label{eq:RhoSeg}\\
\end{eqnarray*}
\item[(b)]\begin{eqnarray}
 \rho\left(F_{\mathbf{U}}\right)&=&\xi^{*} \rho\left(F_{\mathbf{V}}\right)
 + \left(1-\xi^{*} \right)\left(-\dfrac{\left(d+1\right)}{2^d-\left(d+1\right)}\right);\label{eq:rhovref}
\end{eqnarray}
\item[(c)]
$ \rho\left(F_{\mathbf{U}}\right)\leq  \rho\left(F_{\mathbf{V}}\right)$.
\end{itemize}
\end{prop}

We provide a simplified formula for  $\rho\left(F_{\mathbf{V},\mathcal{D}\setminus\mathcal{L}_m}\right)$ in the case of two line segment constructions which will be studied later on in this paper. For the line segment representation \eqref{simplevar} we have:
\begin{eqnarray}\label{rhoseg}
\rho\left(F_{\mathbf{V},\mathcal{D}\setminus\mathcal{L}_m}\right)&=& \dfrac{2^d\left(d+1\right)}{2^d-\left(d+1\right)}\left(\mathbb{E}\left[V^{m}\left(1-V\right)^{d-m} \right] -\dfrac{1}{2^d}\right)
\\&=& \dfrac{2^d\left(d+1\right)}{2^d-\left(d+1\right)}\left(B\left(m+1,d-m+1\right) -\dfrac{1}{2^d}\right)
\end{eqnarray}
where $B\left(x,y\right)$ is the Euler's beta function.

For the representation \eqref{gensimplevar}, under the independence assumption $V_{l}$ $l=1,\ldots,d$ iid, the Spearman's $\rho$ is:
\begin{eqnarray}\nonumber
\rho\left(F_{\mathbf{V},\mathcal{D}\setminus\mathcal{L}_m}\right)&=&\dfrac{2^d\left(d+1\right)}{2^d-\left(d+1\right)}\left(\prod_{l\in \mathcal{L}_m}\mathbb{E}\left[V_l\right]\prod_{l\in \mathcal{D}\setminus\mathcal{L}_m}\mathbb{E}\left[\left(1-V_l\right) \right] -\dfrac{1}{2^d}\right)\\\label{rhoViid}
&=& \dfrac{2^d\left(d+1\right)}{2^d-\left(d+1\right)}\left( \dfrac{1}{2^d}-\dfrac{1}{2^d}\right)=0.
\end{eqnarray}
Finally note that the multivariate Kendall's $\tau$ and Spearman's $\rho$ are permutation invariant, thus permuted uniform vectors will have the same measures.

\section{Special cases}\label{sec:examples}
In this section, we discuss several examples starting from the new constructions proposed in this paper and then reviewing the constructions proposed in the literature which are special cases of our stochastic representations \eqref{simplevar} or \eqref{gensimplevar}. We use multivariate Spearman's $\rho$ to rank the proposals in concordance order\footnote{To keep the paper to a reasonable length, we report here the values of the multivariate Spearman's $\rho$ for dimensions $d$ from 2 to 5. We have evaluated the ordering of the constructions up to dimension $d=20$, and no changes were observed.}.

\subsection{Circulant Variates} 
\label{polytopalvar}
Obtaining the coordinate matrix $\mathbf{X}$ used in Eq. \ref{simplevar} can be costly in high dimensions especially when  numerical procedures are used to solve the optimization problem stated in Section \ref{mixonseg}. We propose suitable constraints on the segment set $\mathcal{S}=(\mathbf{X},\mathcal{G})$ to reduce the computational cost of our procedure. The proposed conditions on the coordinate matrix $\mathbf{X}$ allows for decoupling the CTM constraint. 

First, we  assume that the number of vertexes $n$ is equal to the dimension $d$ of the random vector and $\mathbf{X}=d/2 \tilde{\mathbf{X}}$ where $\tilde{\mathbf{X}}$ is doubly stochastic, and obtain:
\begin{eqnarray*}
\sum^d_{i=1}x_{i,k}&=&\sum^n_{k=1}x_{i,k}=\dfrac{d}{2}.
\end{eqnarray*}
This assumption allows us to simplify the optimization problem and to search for independent solutions for each row of $\mathbf{X}$. 

We assume further constraints on the matrix $\mathbf{X}$ and on the graph $\mathcal{G}$ such that the same optimization problem is solved for all rows of $\mathbf{X}$. We assume the first coordinates of the $d$ vertexes are arranged in increasing order $x_{11}\leq\ldots\leq x_{1d}$ and compute the $k$-th coordinates as the $k$-th circular permutation of the first ones:
\begin{eqnarray}\label{polyvertex}
\begin{array}{ccc}
x_{k1}&=&x_{1\,\left(k-1\right)\left(\mathrm{mod} \,d\right) +1  }\\
       &\vdots&\\
x_{ki}&=& x_{1\,\left(i-1 + \left(k-1\right)\right)\left(\mathrm{mod}\,d\right) +1 }\\
&\vdots&\\
x_{kd}&=& x_{1\,\left(d-1 + \left(k-1\right)\right)\left(\mathrm{mod}\,d\right) +1 }.
\end{array}
\end{eqnarray}  
The resulting coordinate matrix $\mathbf{X}$ is a circulant matrix with $i$-th row sum equal to the $i$-th column sum for all rows. Imposing $x_{11}\leq\ldots\leq x_{1d}$  implies the same set of $\mathbf{a}_{l}$  is used for all $l\in\mathcal{D}$, with the same multiplicities $ \left\vert\mathcal{M}_{l,m}\right\vert= \left\vert\mathcal{M}_{1,m}\right\vert$ but with different positions $m\in \mathcal{M}_{l,m}$ as effect of the circular permutation.  

Furthermore, we choose the edge set in such way to have the same projected graph for each set of coordinates, and assume a circulant graph  which is invariant by circular shifts of the vertexes. 
\begin{definition}[Circulant Graph] Given a subset $\mathcal{L}\subseteq \left\{1,\ldots,\left\lfloor \frac{d}{2}\right\rfloor\right\}$ then the $d$-vertex circulant graph $\mathcal{C}_{d}\left(\mathcal{L}\right)$ is a graph with vertexes $1,\ldots,d$ and edge set $\mathcal{E}_{d,\mathcal{L}}$ is such that $(i,j)\in \mathcal{E}_{d,\mathcal{L}}$  if either $\left\vert i-j\right\vert \in \mathcal{L}$ or $\left(d-\left\vert i-j\right\vert\right) \in \mathcal{L}$. 
\end{definition}


\begin{table}[htb!]
\label{tab:Circulant}
\centering
\begin{tabular}{ >{\centering\arraybackslash} m{2cm} | >{\centering\arraybackslash} m{3cm}| >{\centering\arraybackslash} m{4.6cm}  | >{\centering\arraybackslash} m{2cm} }
 Label & $\mathcal{G}$ & $\mathbf{x}_1$ & Spearman's $\rho$
\\\hline
 $\mathcal{C}_{2}\left(\left\{1\right\}\right)$&\begin{tikzpicture}[scale=0.9, every node/.style={draw,circle,inner sep=0pt}]
  \def\radius{1 cm}
  \coordinate (center) at (0,1);
  \pgfmathsetmacro{\d}{2}
  \pgfmathtruncatemacro{\nodes}{\d-1}
  \foreach \i in {0,...,\nodes}  
  {\pgfmathtruncatemacro{\nodelab}{\i+1}
   \node [minimum size=0.4cm,font=\scriptsize] (\i) at  (90+\i*360/\d:\radius){$\nodelab$};  }  
   \foreach \l in {1} {    
   \foreach \i in {0,...,\nodes} {    
   \pgfmathtruncatemacro {\j}{mod(\l+\i,\d)}
   \draw (\i) -- (\j);} 
   }
   \end{tikzpicture} & $\left(0,1\right)$& $-1$
\\\hline  $\mathcal{C}_{3}\left(\left\{1\right\}\right)$&
\begin{tikzpicture}[scale=0.9, every node/.style={draw,circle,inner sep=0pt}]
  \def\radius{1 cm}
  \coordinate (center) at (0,1);
  \pgfmathsetmacro{\d}{3}
  \pgfmathtruncatemacro{\nodes}{\d-1}
  \foreach \i in {0,...,\nodes}  
  {\pgfmathtruncatemacro{\nodelab}{\i+1}
   \node [minimum size=0.4cm,font=\scriptsize] (\i) at  (90+\i*360/\d:\radius){$\nodelab$};  }  
   \foreach \l in {1} {    
   \foreach \i in {0,...,\nodes} {    
   \pgfmathtruncatemacro {\j}{mod(\l+\i,\d)}
   \draw (\i) -- (\j);} 
   }
   \end{tikzpicture} &$\left(0,\frac{1}{2},1\right)$ & -0.5
   \\ \hline  $\mathcal{C}_{4}\left(\left\{1\right\}\right)$ & 
\begin{tikzpicture}[scale=0.9, every node/.style={draw,circle,inner sep=0pt}]
  \def\radius{1 cm}
  \coordinate (center) at (0,1);
  \pgfmathsetmacro{\d}{4}
  \pgfmathtruncatemacro{\nodes}{\d-1}
  \foreach \i in {0,...,\nodes}  
  {\pgfmathtruncatemacro{\nodelab}{\i+1}
   \node [minimum size=0.4cm,font=\scriptsize] (\i) at  (90+\i*360/\d:\radius){$\nodelab$};  }  
   \foreach \l in {1} {    
   \foreach \i in {0,...,\nodes} {    
   \pgfmathtruncatemacro {\j}{mod(\l+\i,\d)}
   \draw (\i) -- (\j);} 
   }
   \end{tikzpicture} & $\left(0,\frac{1}{3},\frac{2}{3},1\right) $ & -0.2840
   \\ \hline  $\mathcal{C}_{4}\left(\left\{1,2\right\}\right)$ & 
\begin{tikzpicture}[scale=0.9, every node/.style={draw,circle,inner sep=0pt}]
  \def\radius{1 cm}
  \coordinate (center) at (0,1);
  \pgfmathsetmacro{\d}{4}
  \pgfmathtruncatemacro{\nodes}{\d-1}
  \foreach \i in {0,...,\nodes}  
  {\pgfmathtruncatemacro{\nodelab}{\i+1}
   \node [minimum size=0.4cm,font=\scriptsize] (\i) at  (90+\i*360/\d:\radius){$\nodelab$};  }  
   \foreach \l in {1,2} {    
   \foreach \i in {0,...,\nodes} {    
   \pgfmathtruncatemacro {\j}{mod(\l+\i,\d)}
   \draw (\i) -- (\j);} 
   }
   \end{tikzpicture} & $\left(0,
\frac{1}{2}-\frac{1}{2\sqrt{5}}
,\frac{1}{2}+\frac{1}{2\sqrt{5}},
1\right) $ &  -0.2763
   \\ \hline  $\mathcal{C}_{4}\left(\left\{2\right\}\right)$ & 
\begin{tikzpicture}[scale=0.9, every node/.style={draw,circle,inner sep=0pt}]
  \def\radius{1 cm}
  \coordinate (center) at (0,1);
  \pgfmathsetmacro{\d}{4}
  \pgfmathtruncatemacro{\nodes}{\d-1}
  \foreach \i in {0,...,\nodes}  
  {\pgfmathtruncatemacro{\nodelab}{\i+1}
   \node [minimum size=0.4cm,font=\scriptsize] (\i) at  (90+\i*360/\d:\radius){$\nodelab$};  }  
   \foreach \l in {2} {    
   \foreach \i in {0,...,\nodes} {    
   \pgfmathtruncatemacro {\j}{mod(\l+\i,\d)}
   \draw (\i) -- (\j);} 
   }
   \end{tikzpicture} & $\left(0,0,1,1\right) $ & -0.2121
    \\ \hline  $\mathcal{C}_{5}\left(\left\{1\right\}\right)$ & 
\begin{tikzpicture}[scale=0.9, every node/.style={draw,circle,inner sep=0pt}]
  \def\radius{1 cm}
  \coordinate (center) at (0,1);
  \pgfmathsetmacro{\d}{5}
  \pgfmathtruncatemacro{\nodes}{\d-1}
  \foreach \i in {0,...,\nodes}  
  {\pgfmathtruncatemacro{\nodelab}{\i+1}
  \node [minimum size=0.4cm,font=\scriptsize] (\i) at  (90+\i*360/\d:\radius){$\nodelab$};  }  
  \foreach \l in {1} {    
  \foreach \i in {0,...,\nodes} {    
  \pgfmathtruncatemacro {\j}{mod(\l+\i,\d)}
  \draw (\i) -- (\j);} 
  }
  \end{tikzpicture} & $\left(0,1/4,2/4,3/4,1\right) $ & -0.1659
  \\ \hline   $\mathcal{C}_{5}\left(\left\{1,2\right\}\right)$ & 
\begin{tikzpicture}[scale=0.9, every node/.style={draw,circle,inner sep=0pt}]
  \def\radius{1 cm}
  \coordinate (center) at (0,1);
  \pgfmathsetmacro{\d}{5}
  \pgfmathtruncatemacro{\nodes}{\d-1}
  \foreach \i in {0,...,\nodes}  
  {\pgfmathtruncatemacro{\nodelab}{\i+1}
  \node [minimum size=0.4cm,font=\scriptsize] (\i) at  (90+\i*360/\d:\radius){$\nodelab$};  }  
  \foreach \l in {1,2} {    
  \foreach \i in {0,...,\nodes} {    
  \pgfmathtruncatemacro {\j}{mod(\l+\i,\d)}
  \draw (\i) -- (\j);} 
  }
  \end{tikzpicture} & $\left(0,\frac{1}{2}-\frac{\sqrt{3}}{2\sqrt{7}},\frac{1}{2},\frac{1}{2}+\frac{\sqrt{3}}{2\sqrt{7}}),1\right) $& -0.1577
 \\ \hline  $\mathcal{C}_{5}\left(\left\{2\right\}\right)$ & 
\begin{tikzpicture}[scale=0.9, every node/.style={draw,circle,inner sep=0pt}]
  \def\radius{1 cm}
  \coordinate (center) at (0,1);
  \pgfmathsetmacro{\d}{5}
  \pgfmathtruncatemacro{\nodes}{\d-1}
  \foreach \i in {0,...,\nodes}  
  {\pgfmathtruncatemacro{\nodelab}{\i+1}
  \node [minimum size=0.4cm,font=\scriptsize] (\i) at  (90+\i*360/\d:\radius){$\nodelab$};  }  
  \foreach \l in {2} {    
  \foreach \i in {0,...,\nodes} {    
  \pgfmathtruncatemacro {\j}{mod(\l+\i,\d)}
  \draw (\i) -- (\j);} 
  }
  \end{tikzpicture}  & $\left(0,0,\frac{1}{2},1,1\right) $ & -0.1385
\\\hline
\end{tabular}
\caption{Circulant graphs $\mathcal{C}_{d}\left(\mathcal{L}\right)$ of the Circulant Countermonotonic on Segments up to dimension 5.}
\end{table}
The circular symmetry imposed on the vertex coordinates and on the graph simplifies the optimization problem. Whatever the multiplicities  $\left\vert\mathcal{M}_{1,m}\right\vert$, under Assumptions \ref{assRange} and \ref{assAdmiss}, we obtain the following results:
\begin{eqnarray*}
\Psi_l\left(\mathbf{a}_l\right)&=&\Psi_1\left(\mathbf{a}_1\right)=\Phi_1\left(\mathbf{x}_1\right)= -\dfrac{1}{2 \left\vert\mathcal{E}_{d,\mathcal{L}}\right\vert}\displaystyle \sum_{ \left(i,j\right)\in \mathcal{E}_{d,\mathcal{L}} }\log\left\vert x_{1,i}-  x_{1,j }\right\vert, \\
&&\sum^{n_1}_{m=1} \left\vert \mathcal{M}_{1,k} \right\vert a_{1,m} = \sum^{d}_{i=1}x_{1,i}.
\end{eqnarray*}

 This rewriting of constraint and objective function allows one to minimize on $\mathbf{x}_1$  instead of $\mathbf{a}_l$. Minimization automatically excludes vectors that violate Assumption \ref{assAdmiss} given $\mathcal{E}_{d,\mathcal{L}}$, because for those cases the objective function is infinite.
 The following definitions introduce formally our new proposal called Circulant Variates (CCV).

\begin{definition}[Circulant Countermonotonic]\label{circmixonseg}
We call the circulant matrix $\mathbf{X}\in \mathbb{R}^d\times \mathbb{R}^n$ whose rows are obtained by the $d$ circular shifts of the first row, a solution of the Circulant Countermonotonic on Segments problem, on the circulant graph $\mathcal{C}_{d}\left(\mathcal{L}\right)$, if $x_{1,1}=0$, $x_{1,d}=1$ and  $\left\{x_{1,2},\ldots,x_{1,d-1}\right\}$ solve the convex minimization problem:

\begin{eqnarray*}
\displaystyle\min_{\left\{x_{l,m}\right\}^{d-1}_{m=2}\in \left[0,1\right]^{d-2}}\Phi_1\left(\mathbf{x}_1\right),
\end{eqnarray*}
with 
\begin{eqnarray*}
\Phi_1\left(\mathbf{x}_1\right)=-\dfrac{1}{2 \left\vert\mathcal{E}_{d,\mathcal{L}}\right\vert}\displaystyle \sum_{ \left(i,j\right)\in \mathcal{E}_{d,\mathcal{L}} }\log\left\vert x_{1,i}-  x_{1,j }\right\vert ,
\end{eqnarray*}
subject to:
\begin{eqnarray*}
\sum^{d}_{i=1}x_{1,i}=\dfrac{d}{2}.
\end{eqnarray*}

\end{definition}

\begin{definition}[Circulant Variates]
Let $\mathcal{S}_{d,\mathcal{L}}=\left\{\mathbf{X},C_{d}\left(\mathcal{L}\right)\right\}$ a collection of segments such that $\mathbf{X}$ is a solution  Circulant Countermonotonic on Segments problem on  $\mathcal{C}_{d}\left(\mathcal{L}\right)$. Variates obtained from the components of the $d$-dimensional random vector uniformly distributed on $\mathcal{S}_{d,\mathcal{L}}$ are  Circulant Variates (CCV).
 \end{definition}

The following corollary is the application of Theorem \ref{unif}  to CCV.
\begin{coro}\label{unifpoly} CCV are marginally standard uniform and constant in sum.
\end{coro} 
For $C_d\left(\left\{1\right\}\right)$ a solution can be derived as follows. The sum in each of the first $d-1$ equations of \eqref{simplevar} has only two terms, and considering the $m$-th equation we have
\begin{eqnarray}
\nonumber
\dfrac{1}{\left(x_{1m+1}-x_{1m}\right)}+\dfrac{1}{\left(x_{1d}-x_{11}\right)}&=& d.
\end{eqnarray}
Substituting the constraints $x_{11}=0$ and $x_{1d}=1$ , we find that the $x$'s are uniformly spaced on the unit interval and satisfy all the $d-2$ equations in \eqref{uniformseg}. 
The case of $C_{3}\left(\left\{1\right\}\right)$ was already studied in \cite{nelsen2012directional} . In their example, the probability mass has a distribution uniform on the edges of the triangle with vertexes $\mathbf{x}_1=\left(0, 1/2, 1\right)$, $\mathbf{x}_2=\left(1/2, 1, 0\right)$, $\mathbf{x}_3=\left(1, 0, 1/2\right)$. The other two vertexes are different 3-cycles of the first one. This implies that the row sum of $\mathbf{x}=(\mathbf{x}_1,\mathbf{x}_2,\mathbf{x}_3)$, is equal to its column sum and both are equal to $3/2$. \cite{LEE2014} show that  the construction in \cite{nelsen2012directional} is $3$-CTM.
Additionally, the exchangeable version of the $C_{d}\left(\left\{1\right\}\right)$ construction is distributionally equivalent to the degenerate random balanced sampling introduced in Eq. (8) of \cite{Ger}. They propose to generate $Z_1$ as an uniform random variable on $\left[-1,1\right]$ and obtain the remaining variables according to 
\begin{eqnarray*}
Z_l &=& c_l - \dfrac{Z_1}{d-1},\quad c_l = -1 +\dfrac{2 l -3 }{d-1}
\end{eqnarray*}
$l\in \mathcal{D}$, and then randomly permute the $Z_l$. They show that $\sum^{d}_{l=1}Z_l=0$ and that once permuted, the $Z_l$'s are uniformly distributed on $\left[-1,1\right]$. If we set $U_l= (Z_l + 1)/2$ then the permuted version can be written in terms of permuted  $C_{d}\left(\left\{1\right\}\right)$ construction: 
\begin{prop}\label{RBScirc}
The exchangeable version of $U_l= (Z_l + 1)/2$ $l\in \mathcal{D}$ has the same distribution of the exchangeable version of CCV with dependence graph $C_{d}\left(\left\{1\right\}\right)$.
\end{prop}

\subsection{Rotation Sampling}\label{rotsample}

 \cite{Fishman1983} rephrase the original \cite{ham-mo} proposal for $d>2$ obtaining an equivalent construction with the standard marginal uniformity that was missing in \cite{ham-mo}. Their construction was named rotation sampling because the modulo one arithmetic on which it is based is often associated with circular motion.
 
\begin{prop} \label{rotsampseg} The line segment stochastic representation  \eqref{simplevar} of the rotation sampling in the Example \ref{ex:rotsam} has $2d$ vertexes with coordinate matrix of elements: \begin{eqnarray*}
x_{l,m}&=& \left\{\begin{array}{ccc}  \dfrac{l+m-1}{d} &\hbox{if}& m < d+2-l 
\\\dfrac{l+m-1-d}{d} & \hbox{if}& m \geq d+2-l  \end{array}\right.
\\x_{l,d+m}&=& \left\{\begin{array}{ccc}  \dfrac{l+m-2}{d} &\hbox{if}& m < d+2-l 
\\\dfrac{l+m-2-d}{d} & \hbox{if}& m \geq d+2-l  \end{array}\right.
\end{eqnarray*}

and  edge set:
\begin{eqnarray*}
\mathcal{E}^{RS}=\left\lbrace \left(i,j\right) \in \{1,2,\ldots,d\}\times\{d+1,d+2,\ldots,2d\}\vert j=d+i 
\right\rbrace.
\end{eqnarray*}
\end{prop}

\begin{coro}\label{unifRS}A rotation sampling random vector has standard uniform marginals and is not $d$-CTM.
\end{coro} 

We report in Table \ref{tab:rhorotsamp} the analytic values of the Multivariate Spearman's $\rho$ for rotation sampling vectors up to dimension $d=5$.

\begin{table}[htbp]
  \centering
  \caption{Values of the multivariate Spearman's $\rho$ for rotation sampling.}\label{tab:rhorotsamp}
    \begin{tabular}{|c|c|c|c|c|}
    \hline
  $d$ &  2     &  3     & 4     & 5      \\
  \hline
  $\rho_{\min}$&    -0.5    & -0.33  & -0.2168  & -0.1372  \\
   \hline
    \end{tabular}%
  \label{tab:addlabel}%
\end{table}%

For $d=2$ this construction does not reduce to the usual antithetic variates. This is suggested by a value of multivariate Spearman's $\rho$ different from the value of $-1$ attained by the Fr\'{e}chet lower bound. Lacking the constant sum property,  the proposal has multivariate Spearman's $\rho$ larger than the one of the $d$-CTM proposals considered in this paper.

\subsection{Arvidsen and Johnson: A Fresh Look}\label{permdisp}
In the pioneering paper of \cite{Arvi:John:82:VRT}, the objective of variance reduction is obtained by designing the first standard uniform $d$-CTM construction \eqref{AV}. \cite{CraiuMeng2005} have shown that this construction is displacing the binary digits of $U_1$ and give the name permuted displacement to its exchangeable version. We show the following relationship between the two.
\begin{prop}\label{propAJ}
For $d=3$, the construction of \cite{Arvi:John:82:VRT} given in Example \ref{ex:AV} is equivalent to the antithetic proposal of \cite{Gaffke1981} given in \eqref{rushsample}.
\end{prop}

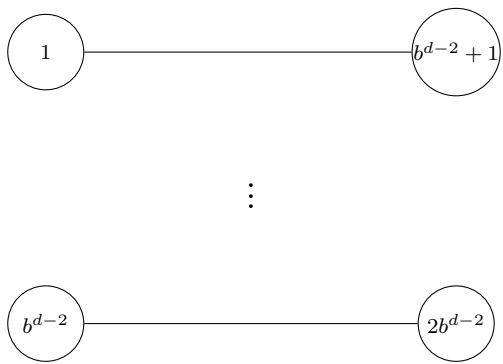
\begin{figure}[t]
\centering
\begin{tikzpicture}[scale=0.9]
  \node [{draw,circle,inner sep=0pt,minimum size=1cm,font=\scriptsize}] (a) at  (-3,2){$1$};
  \node [{draw,circle,inner sep=0pt,minimum size=1cm,font=\scriptsize}] (b) at (3,2){$b^{d-2}+1$};
  \draw (a) -- (b);
  \node [{draw,circle,inner sep=0pt,minimum size=1cm,font=\scriptsize}] (c) at  (3,-2){$2b^{d-2}$};
  \node [{draw,circle,inner sep=0pt,minimum size=1cm,font=\scriptsize}] (d) at (-3,-2){$b^{d-2}$};
  \draw (c) -- (d);
\node[scale=1] (f) at  (0,0){ {\large$\vdots$}};
   \end{tikzpicture}
\caption{$\mathcal{G}$ for $b$-based Arvidsen and Johnson construction}
\label{fig:AJ3d}
\end{figure}

The Ardvisen and Johnson' construction and the following family of general constructions admit a line segment representation. The generalization is useful to show that Ardvisen and Johnson is the only $d$-CTM within this family and attains consequently the minimal Spearman's $\rho$ within this family.

\begin{definition} \label{def:basedAV} Given $U_1\sim\mathcal{U}[0,1]$ and $b\in\mathbb{N}$ the base-$b$ Ardvisen and Johnson construction is:
\begin{eqnarray}\label{basedAV} 
\begin{array}{ccccc}U_i= \left( b^{i-2} U_1 + 1/b \right)\mathrm{mod}\,1 & \hbox{for}& i=\left\lbrace2,\ldots,d-1\right\rbrace & & U_d= 1- \left( b^{d-2} U_1 \right)\mathrm{mod}\,1
\end{array} 
\end{eqnarray}
\end{definition}

\begin{prop} \label{basedAVseg} The line segment stochastic representation  \eqref{simplevar} of the construction in \eqref{basedAV} has $2b^{d-2}$ vertexes with coordinate matrix $\mathbf{X}=(\mathbf{z}^{T},\mathbf{y}^{T})$ where:
\begin{eqnarray*} \begin{array}{ccccccccc}
\mathbf{y}_{1}&=& \left(0,\dfrac{1}{b^{d-2}},\ldots,\dfrac{b^{d-2}-1}{b^{d-2}} \right)^{T}, & 
\mathbf{z}_{1}&=& \mathbf{y}_{1} + \dfrac{1}{b^{d-2}}\mathbf{1}_{b^{d-1}}\\
\mathbf{y}_{k}&=& \left( b^{k-2} \mathbf{y}_{1}+ \dfrac{1}{b}\right)\mathrm{mod}\,1, & \mathbf{z}_{k}&=& \mathbf{y}_{k} + \dfrac{1}{b^{d-k}}\mathbf{1}_{b^{d-1}}\\
 \mathbf{y}_{d}&=& \mathbf{1}_{b^{d-1}}, & \mathbf{z}_{d}&=& \mathbf{y}_{d} - \mathbf{1}_{b^{d-1}}\\
 \mathbf{y}&=&\left(\mathbf{y}_{1},\ldots,\mathbf{y}_{d}\right) &
 \mathbf{z}&=&\left(\mathbf{z}_{1},\ldots,\mathbf{z}_{d}\right)\\
\end{array}
\end{eqnarray*}

and  edge set:
\begin{eqnarray*}
\mathcal{E}^{AJ}=\left\lbrace \left(i,j\right) \in \left\lbrace 1,\ldots,b^{d-2}\right\rbrace\times\left\lbrace 1,\ldots,b^{d-2}\right\rbrace \vert j=b^{d-2}+i 
\right\rbrace
\end{eqnarray*}
\end{prop}

\begin{coro}\label{CTMAJ} The base-$b$ Ardvisen and Johnson random vector has standard uniform marginals.  Only the case $b=2$ is $d$-CTM.
\end{coro} 

Table \ref{tab:rhoAJ} shows that the constructions in Definition \ref{def:basedAV} attain the lowest Spearman's $\rho$ for $b=2$. 

\begin{table}[htb!]
\label{tab:rhoAJ}
\centering
    \begin{tabular}{|c|c|ccccc|c|}
    \hline
    \multicolumn{2}{|c|}{Spearman's $\rho$} &\multicolumn{5}{c|}{$b$} & GR\\
\multicolumn{2}{|c|}{} & 1     & 2     & 3     & 4     & 5 & \\
\hline
$d$   & 2     & -1    & -1    & -1    & -1    & -1 & -1\\
          & 3     & -0.3333 & -0.5000 & -0.3333 & -0.2083 & -0.1200 & -0.5000\\
          & 4     & -0.0909 & -0.2822 & -0.1662 & -0.0869 & -0.0367 & -0.2525\\
          & 5     & 0.0154 & -0.1637 & -0.0933 & -0.0455 & -0.0165 & -0.1538\\
\hline
    \end{tabular}%
    \caption{Multivariate Spearman's $\rho$ for base-$b$ Ardvisen and Johnson random vectors and the multivariate proposal of \cite{Gaffke1981} in Example \ref{ex:GR}}
\end{table}
Table \ref{tab:rhoAJ} reports also the multivariate proposal of \cite{Gaffke1981} (GR) described in Example \ref{ex:GR}. The table shows that GR performs better than the non $d$-CTM proposal, but worse than the AJ proposal because it has independence as one of the main ingredients as discussed in the introduction. This also show the effectiveness of multivariate $\rho$ as a ranking measure for $d$-CTM vectors.

\subsection{Latin Hypercube Iterations}
 In this section, we will reconsider the Iterated Latin Hypercube construction introduced in \cite{CraiuMeng2005} establish a relationship with a $d$-dimensional generalization of the superstar introduced in \cite{Ger}, and with our strict countermonotonic on segments construction. \\
The latin hypercube sampling, introduced by \cite{McKa:Beck:Cono:79:CTM}, and further developed by \cite{CraiuMeng2005} with the goal of obtaining variance reduction  in MCMC sampling, consists in the following steps: for $t= 0,\ldots, T$ take an iid standard uniform $d$-dimensional random vector $\mathbf{U}_0$  and let $\mathcal{D}^{\pi}_t = \left(\pi_{t}\left(0\right),\ldots,\pi_{t}\left(d-1\right)\right)^{T}$ be a permutation of $\lbrace 0,1,\ldots,d-1\rbrace$ independent of  $\mathbf{U}_{0},\ldots,\mathbf{U}_{t-1}$ and 
\begin{eqnarray}\label{ILHiter}
\mathbf{U}_t= \dfrac{1}{d}\left(\mathcal{D}^{\pi}_t + \mathbf{U}_{t-1}\right)
\end{eqnarray}
If $t=1$ \eqref{ILHiter} corresponds to the original Latin Hypercube Sampling, and $t>1$ to the Iterated Latin Hypercube procedure introduced in \cite{CraiuMeng2005}. It was showed in \cite{Craiu2006} that ILH iterations represent an Iterated Function System with probabilities (IFSP) $\left(\left[0,1\right]^{d}, \left(w_{\pi}\right),p_{\pi}\right)$ with  $w_{\pi} $ similitudes with contraction ratio $\dfrac{1}{d}$ associated to each permutation $\pi$ of $\left\{1,\ldots,d\right\}$
\begin{eqnarray} 
w_{\pi}\left(\mathbf{u}\right)= \left(\dfrac{\pi\left(1\right)-1}{d} + \dfrac{u_1}{d}, \ldots, \dfrac{\pi\left(d\right)-1}{d} + \dfrac{u_d}{d}   \right).
\end{eqnarray}
We can show that the $3$-dimensional superstar considered in \cite{Ger} can be obtained by using the same IFSP: 

\begin{prop}\label{ILHsperstar}
The $3$-dimensional superstar proposed in  \cite{Ger}:
\begin{eqnarray}
 \mathbf{X}_{t}=f_{k}\left(\mathbf{X}_{t-1}\right)= \dfrac{1}{3} \mathbf{X}_{t-1} + \dfrac{2}{3}\mathbf{V}_k
\end{eqnarray} 
with $\mathbf{V}_k$ a random permutation  of $\left\{-1,0,1\right\}$ and an initial $\mathbf{X}_0$ a $3$-dimensional vector such that $-1\leq \mathbf{X}_{i0}\leq 1$ and $\sum^{3}_{i=1}\mathbf{X}_{i0} =0$, up to a change of support, is generated by the same IFSP of the 3-dimensional version of the Iterated Latin Hypercube construction introduced in \cite{CraiuMeng2005}.
\end{prop}

\begin{prop} \label{ILHseg} The line segment stochastic representation  \eqref{gensimplevar} of the construction in \eqref{ILHiter} has $2d!$ vertexes with coordinates $x_{l,k}= \frac{\pi_k\left(l\right) + 1}{d}$ and $x_{l,d!+ k}= \frac{\pi_k\left(l\right)}{d}$, $k=1,\ldots,d!$ and  edge set:
\begin{eqnarray*}
\mathcal{E}^{LH}=\left\lbrace \left(i,j\right) \in \left\lbrace 1,\ldots,d!\right\rbrace\times\left\lbrace d!+1,\ldots,2d!\right\rbrace \vert j=d!+i 
\right\rbrace
\end{eqnarray*}
\end{prop}

\cite{CraiuMeng2005} obtained standard marginal uniformity for the special case of an initial vector with iid components. In the superstar case, the original distribution is concentrated on points but converges to standard marginal uniforms as showed in \cite{Ger}. In \cite{CraiuMeng2005}, it is also shown that in the limit $t\rightarrow \infty$ the ILH is $d$-CTM . We show that the ILH iterations preserve marginal uniformity and constant sum in the general case.
 
 \begin{coro}\label{unifILH} Let $\mathbf{U}_{t-1}$ be a dependent  random vector of dimension $d$, whose coordinates add up to d/2 (a.s.) and each coordinate has a  $\mathcal{U}[0,1]$ distribution. Then the random vector $\mathbf{U}_t$ in \eqref{ILHiter} has all its coordinates  marginally $\mathcal{U}[0,1]$ distributed and adding up to $d/2$.
\end{coro} 
 The preservation of strict $d$-CTM property raises the possibility of using ILH iterations on Arvidsen and Johnson's and CCV constructions. In addition, using results in section \ref{sec:mrho},  we can give a closed-form expression for the multivariate Kendall's $\tau$ and Spearman's $\rho$ for ILH iteration applied to $\mathbf{V}$ with i.i.d components or in representation \eqref{simplevar}  and rank the obtained random vector in the correlation order. 
In particular, ILH iterations on the Arvisen and Johnson construction and CCV have a constant sum and minimal multivariate Kendall's $\tau$  ( equation \eqref{taumin}).  Spearman's $\rho$ for those constructions can be easily obtained using the deterministic composition and equation \eqref{rhoseg}.
Using the segment representation in Proposition \ref{ILHseg}, the ILH($T$) proposal applied to $\mathbf{V}$ with i.i.d components, can be expressed as a $T$-fold deterministic composition. Each composition expands the cardinality of vertex and edge sets by $d!$, resulting in a vertex set of cardinality $2\left(d!\right)^T$ and an edge set of cardinality $\left(d!\right)^T$ . To maintain a feasible notation we substitute the index $k$ of the different edges with the multi-index $\left\{k_t\right\}^T_{t=1}$ where each $k_t = 1\ldots d!$ :

\begin{eqnarray}\label{ILHTvertex}
\alpha_{l,k_1,\ldots,k_T}=\sum^{T}_{r=1} \dfrac{\pi_{k_r}\left(l\right)}{d^{T-r+1}} ,&&\beta_{l,k_1,\ldots,k_T}=\sum^{T}_{r=1} \dfrac{\pi_{k_r}\left(l\right)}{d^{T-r+1}} +\dfrac{1}{d^T} 
\end{eqnarray}

The following proposition allows for computing multivariate Kendall's $\tau$ and Spearman's $\rho$ for the ILH($T$) case.

\begin{prop} \label{ILHtaurho} 
Let $U_0$ a $d$ dimensional random vector with iid $\mathcal{U}[0,1]$ components and let  $\mathbf{U}_T$ be the vector obtained by applying the transformation $T$ given in \eqref{ILHiter} to it. Then the Kendall's $\tau$ is
\begin{eqnarray*}
\tau&=&\dfrac{1}{2^{d-1}-1}\left( \dfrac{1}{\left(d!\right)^T}-1\right)
\end{eqnarray*}
and the Spearmans' $\rho$ has coefficient $\xi^{*}$ of Equation \eqref{eq:rhovref} equal to
\begin{eqnarray*}
\xi^{*} &=& \resizebox{0.9\linewidth}{!}{$\displaystyle \sum_{m_0+m_1+\ldots+m_T=d} {d\choose m_0,m_1,\dots ,m_T} \prod^{T}_{t=1}\left\{{d\choose m_t}^{-1} \sum_{{\tiny\left\{i_{1},\ldots,i_{m_t}\right\}\in \mathcal{P}_{\mathcal{D},m_t}} } \prod^{m_t}_{l=1} \dfrac{2\left(i_l-1\right)}{d^{T-t+1}}   \right\}\left(\dfrac{1}{d^T}\right)^{m_0}$}   
\end{eqnarray*}
where $\mathcal{C}_{\mathcal{D},m_2}$ denotes the set of combinations of the elements of $\mathcal{D}$ with $m_t$ elements.
\end{prop}

In the one iteration case, i.e., $T=1$, we obtain the following expression:
\begin{eqnarray*}
\xi^{*} &=&  \prod^{d-1}_{l=0} \dfrac{(2l  +1)}{d}= \prod^{d}_{l=1}\dfrac{\left(2l  -1\right)}{d}=\dfrac{\left(2d  -1\right)!}{d^d \left(2^{d-1} \left(d-1\right)!\right)} =  \dfrac{\prod^{2d-1}_{l=d} l}{d^d2^{d-1} }
\end{eqnarray*}
and using the arithmetic and geometric means inequality we obtain the bound:
\begin{eqnarray*}
\xi^{*}_1 &\leq&   \dfrac{\left(\frac{1}{d}\sum^{2d-1}_{l=d} l\right)^d}{d^d2^{d-1} }\leq  2\left(\dfrac{3d-1}{4d}\right)^{d}\leq 1.
\end{eqnarray*}

For the sake of comparison, in figure \ref{fig:mrho} we summarize Kendall's $\tau$ and Spearman's $\rho$ for all the constructions discussed in this section.  Since Kendall's $\tau$ and Spearman's $\rho$ are permutation invariant, the same ranking applies to the exchangeable versions of the constructions. 
\begin{figure}[t]\label{fig:mrho}
\begin{center}
\includegraphics[trim=70 50 50 0,clip,scale=0.5]{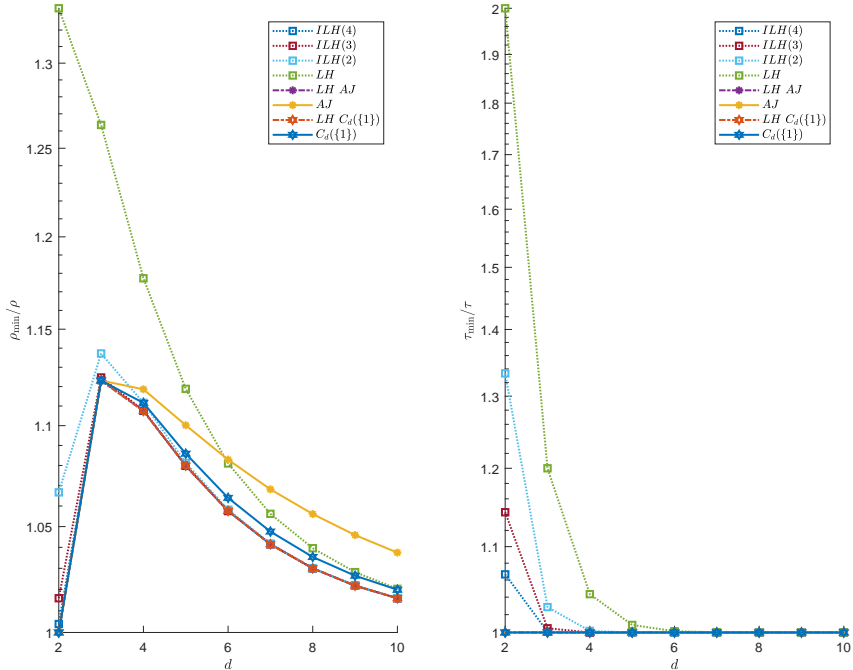}
\end{center}
\caption{Ratio of minimum value attainable by multivariate association measures and the values attained by different antithetic random vectors (vertical axis) as a function of dimension $d$ (horizontal axis). Left: the Spearman's $\rho$ measure. Right: the Kendall's $\tau$ measure. Note: in each plot, larger values indicate constructions farther from the minimum.}
\end{figure}

\section{A Central Limit Theorem}\label{sec:CLT}

In this section, we study the Central Limit Theorem for our best performing classes of variates. 

\begin{definition}[Generalized Latin Hypercube Sample] Let $\pi_{i}$, $i=1,\ldots,p$ independent random permutations of $0,\ldots,d-1$ and  $\mathbf{V}^{i}=\left( V_1^i,\ldots,V_d^i\right)$ $i=1,\ldots,p$  random vectors, independent from the $\pi_{i}$ and from each other,  identically distributed with probability measure $\mu$.  A $p\times d$ matrix $\mathbf{U}$ is a Generalized Latin Hypercube Sample if has stochastic representation
\begin{eqnarray}\label{lssample}
U_l^i =\frac{\pi_{i}\left(l\right) + 1}{d} V^{i}_{l} +  \frac{\pi_{i}\left(l\right)}{d}\left(1-V^{i}_{l}\right)\, i=1,\ldots,p , l=1,\ldots,d 
\end{eqnarray}
\end{definition}

We remark that the most negative dependent constructions ILH(T), LH$C_d\left(\left\{1\right\}\right)$, and LHAJ generate Generalized Latin Hypercube Samples. 

The following lemma introcuce the irrelevance of the distribution on the $\mathbf{V}^{i}$, $i=1,\ldots,p$.

\begin{lemm}\label{muirrelevance}
Consider  a $a$-LH sample and a  $b$-LH sample where  $a$ and $b$ are two different Radon measures. The following relationship holds for every function $f$ locally integrable with respect to both measures :
\begin{eqnarray*}
\mathbb{E}_{a-LH}\left[\left(\dfrac{1}{d}\sum^{d}_{l=1}f\left(\mathbf{U}_l\right)\right)^r \right]-\mathbb{E}_{b-LH}\left[\left(\dfrac{1}{d}\sum^{d}_{l=1}f\left(\mathbf{U}_l\right)\right)^r \right]= o\left(1\right)
\end{eqnarray*}
\end{lemm}

Given the previous lemma we are able to show that the asymptotic distribution is the same of the ordinary Latin Hypercube. In particular \cite{Stei:87:LSP} express the variance of the latin hypercube using the ANOVA decomposition of the function $f$:
\begin{eqnarray}\label{ANOVA}
f\left(\mathbf{u}\right) &=& \mathbb{E}_{IID}\left[f\left(\mathbf{U}\right)\right] + \sum^{p}_{i=1} f_i\left(u_i\right) + r\left(\mathbf{u}\right)\\\nonumber
f_i\left(u_i\right) &=& \mathbb{E}_{IID}\left[ f\left(\mathbf{u}\right) - \mathbb{E}_{IID}\left[f\left(\mathbf{U}\right)\right] \left\vert U_i =u_i\right.\right] 
\end{eqnarray}
\begin{theo}\label{CLT} Let 
$\bar{X}= \dfrac{1}{d}\sum^{d}_{l=1}f\left(\mathbf{U}_l\right)$
with $\mathbf{U}_l$ $l\in\mathcal{D}$ from a $\mu$-LH sample with $\mu$ Radon and $f$ bounded and locally integrable with respect to $\mu$. Then $\sqrt{d} \left( \bar{X}- E_{IID}\left(X\right)\right)$ converges in distribution to $\mathcal{N}\left(0, \int_{\left[0,1\right]^p} r\left(\mathbf{u}\right)^2d\mathbf{u}\right)$, where $ r\left(\mathbf{u}\right)$ is introduced in \eqref{ANOVA}.
\end{theo}
\begin{rema} The hypothesis of independence $\mathbf{U}^{j}$ independent from $\mathbf{U}^{k}$ for all $j\neq k$ $j,k=1,\ldots,p$ in the definition of Generalized Latin Hypercube Sample is not restrictive as it seems. In practice, one can use the inverse  Rosenblatt transform (see for example \cite{ruschendorf2013}, Theorem 1.12) to obtain samples from a generic distribution. Then, if the composition of $f$ with the inverse Rosenblatt transform is bounded, we are still under the incidence of Theorem \ref{CLT}. A multivariate version of the central limit theorem can be obtained as in Corollary 1 of \cite{Owen:92:CLT}. Concerning the boundedness assumption on $f$, it is probably too restrictive because in \cite{loh} a multivariate Berry-Essen type Bound for the standardized multivariate version of $\bar{X}$, in the case of Latin Hypercube, is obtained under the assumption that the multivariate function $\mathbf{f}$ involved is Lebesgue measurable and $\mathbb{E}\left\Vert \mathbf{f} \right\Vert^3<\infty$.  
\end{rema}

\section{Numerical Illustrations}\label{sec:numerics}
We illustrate the performance of the methods presented in this paper using several simulation exercises involving standard Monte Carlo, Markov chain Monte Carlo and Sequential Monte Carlo algorithms \footnote{Replication codes for the examples in the paper can be found at https://github.com/Frattalol/Livingontheedge}.

One of the critical dimension used to rank our countermonotonic vectors is sampling time, as shown in Figure \ref{fig:samplingtimes1}. It is obtained by averaging over 1000 independent replications. In each experiment we sampled 5000 antithetic vectors of dimension $2\le d \le 20$. Samplings schemes have been implemented in Matlab on a Windows 10 laptop with an Intel i7-6500U CPU and 8 GB of RAM.
 
 \begin{figure}[t]\label{fig:samplingtimes1}
 \begin{center}
\includegraphics[trim=50 50 50 40,clip, scale=0.5]{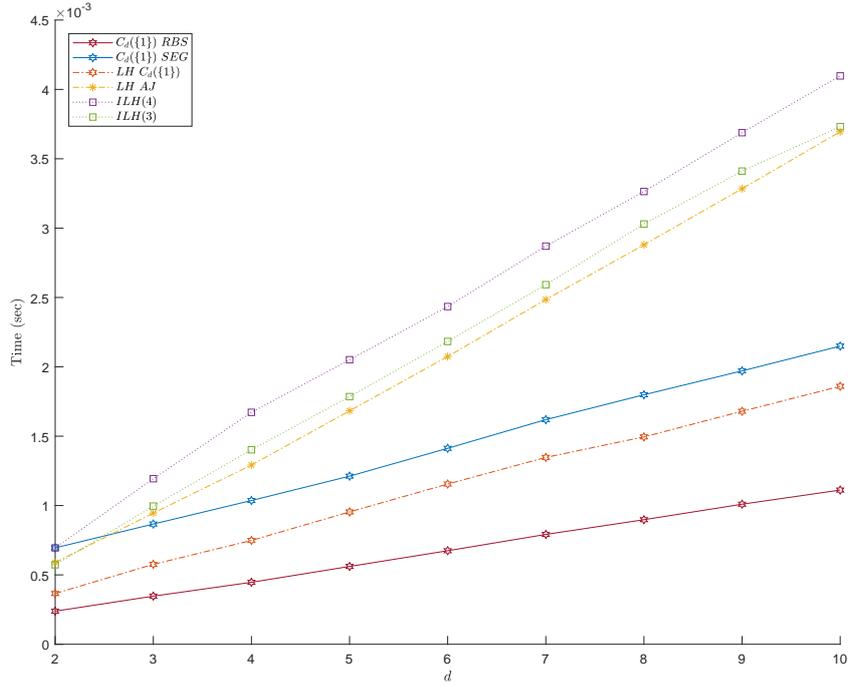}
\end{center}
\caption{Sampling times (vertical axis) as function of the sampling dimension $d$ (horizontal axis) for different methods (different symbols and colors). All estimates are averages over 1000 experiments. In each experiments 5000 antithetic vectors of dimensions $d$ are sampled following a given method.}
\end{figure}

Permuted $C_d\left(\left\{1\right\}\right)$ outperforms the other competitors. We report the times for the segmental version that randomly permutes stochastic representation \eqref{simplevar} and the RBS version that uses the equivalence with random balanced sampling described in Proposition \ref{RBScirc}. The latter choice is faster and it is also used as the base for the Latin Hypercube iteration in $LH$ $C_d\left(\left\{1\right\}\right)$.

Consequently, all the experiment in the section  use the exchangeable version of LH $C_d\left(\left\{1\right\}\right)$. In fact, LH $C_d\left(\left\{1\right\}\right)$ vector reach the lower value of multivariate Kendall's $\tau$ and  Spearman's $\rho$ (fig.\ref{fig:mrho} ), but is  also a faster sampling scheme.

\subsection{Monte Carlo Integration}
The first evaluation of our methodology is for the Monte Carlo integration on the unit hypercube,
with integration points chosen according to the three competing schemes. In standard Monte-Carlo (MC), the sampling points are iid, and each sampling point corresponds to a random vector of dimension equal to the number of variables in the function. In our antithetic method (LH $C_{d}\left(\left\{1\right\}\right)$) the variables used to populate the same coordinate for different sampling points are from an antithetic vector. Across coordinates, these antithetic vectors are independent of one another. For the Quasi Monte Carlo's Sobol scheme (QMC Sobol'), a deterministic sequence of points uniformly cover the hypercube of dimension equal to the number of variables.
In the evaluation of the integration problem complexity the effective dimension plays an important role. We refer to \cite{owen2003}, \cite{wangfang2003} and \cite{wangsloan2005} for the formal definition of truncation $p_t$ and superposition $p_s$ dimensions. \cite{owen2003} shows that a low 
$p_s$ is necessary for QMC to surpass the computational efficiency of MC when the sample sizes are at practical levels. \cite{wangfang2003} and \cite{wangsloan2005} show that the integrands commonly used in option pricing have  $p_t\simeq p$ and  $p_s\leq 2$ and explain, using those results, QMC's good performance in this domain. According to our Theorem \ref{CLT} for integrands with $p_s=1$, i.e. for functions that are well approximated by sums of one dimensional functions, LH $C_{d}\left(\left\{1\right\}\right)$ should be efficient in reducing the variance. Theorem \ref{CLT} also guarantees that our method cannot perform worst than MC, asymptotically in the number of points. To investigate the role of effective dimension in the relative performance of the three competing methods, we use the two parameter function introduced in \cite{wangsloan2005}:
\begin{eqnarray}\label{wangsloan}
f\left(\mathbf{x}\right)=\prod^{p}_{i=1}\left(1+ a\tau^{i}\left(x_i-1/2\right)\right).
\end{eqnarray} 
Varying the parameter $a$ has more effect on $p_t$ than on $p_s$ and varying the parameter $\tau$ has the opposite effect. We consider an high dimensional function ($p=100$) and different specifications of effective dimensions, according to the parameters reported in Table \ref{tab:wangsloan}.

\begin{table}[htbp]
  \centering
  \caption{\cite{wangfang2003} effective dimensions (truncation and superposition dimensions, $p_t$  and $p_s$, respectively) of the integrand function in Eq. \eqref{wangsloan} for dimension $p=100$ and different parameter settings (columns).}\label{tab:wangsloan}
  \setlength{\tabcolsep}{4pt}
    \begin{tabular}{|c|c c c c c| c c c c c| c c c c  c|}
    \hline
  $a$ &0.1	&0.1	&0.1	&0.1	&0.1	&1	&1	&1	&1	&1 &	10&	10	&10	&10	&10
   \\
  \hline
  $\tau$ & 0.1	&0.5&	0.8&	0.9&	1&	0.1&	0.5&	0.8&	0.9&	1&	0.1&	0.5&	0.8&	0.9&	1\\\hline
  $p_t$ & 2	&4	&11	&22	&100	&2	&4	&11	&23	&100	&2	&5	&17	&39	&100\\\hline
  $p_s$ &1&	1&	1&	1&	2&	1&	1&	2&	3&	14&	1&	3&	8&	15&	96
\\\hline  
    \end{tabular}
  \label{tab:effectivedim}%
\end{table}%

Effective dimensions are computed using the methods for multiplicative functions introduced in  \cite{wangfang2003}.

Figure \ref{fig:MCgenz_1} shows the mean square error (MSE) (vertical axis) and the computing time (horizontal axis) of Monte Carlo (red), QMC (black) and circulant variates LH-$C_d(\{1\})$ (blue) sampling, for the different effective dimensions $p_t$ and $p_s$ (different plots) given in Table \ref{tab:wangsloan}. Our LH-$C_d(\{1\})$ method has the best performance when the superposition dimension is equal to 1. The performance is decreasing in the truncation dimension. For $p_s =2,3$, QMC performs better than the method proposed here when the number of points used in the integration is high. QMC advantage increases in the truncation dimension. At moderate $p_s$  QMC dominates. In those cases, our proposal is slower than MC, but reaches the same MSE. In the extreme case of $f$ almost full dimensional (right lower corner),  $C_{d}\left(\left\{1\right\}\right)$ is performing as bad as MC, but QMC is doing orders of magnitudes worst. These numerical results are in line with the result in Theorem \ref{CLT} and indicate that our method should be used when the superposition dimension is low and when there is no information about effective dimensions of the integrand since in the worst case, it  reproduce the precision of standard MC estimates.

%
%
\begin{figure}[t]
\begin{center}
\includegraphics[trim=80 30 80 50,clip,scale=0.6]{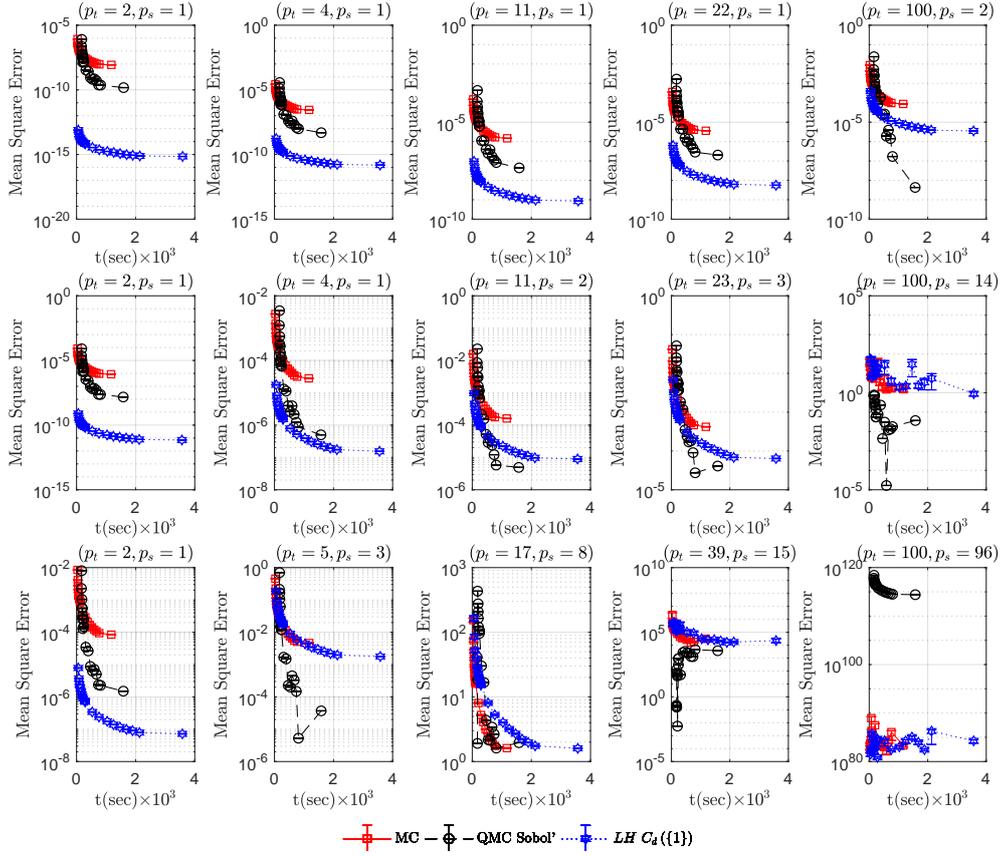}
\end{center}
\caption{Mean square error (vertical axis) and computing time (horizontal axis) for different effective dimensions $p_t$ and $p_s$ (different plots). In each plot: Monte Carlo (red), QMC (black) and circulant variates LH-$C_d(\{1\})$ (blue) sampling. For each line: different number of samples from 10 to 1000 (circles). For each setting, all statistics are averages over 10000 experiments.}
\label{fig:MCgenz_1}
\end{figure}
\subsection{Markov Chain Monte Carlo}
\subsubsection{Bayesian inference on Probit \cite{Meng2001}}
The data used are taken from \cite{Meng2001} and represent the clinical characteristics summarized by two covariates of 55 patients, of which 19 diagnosed with lupus. The disease indicator is modelled as independent $Y_i\sim \mathcal{B}er\left(\Phi\left(x_i^{T} \beta \right)\right)$ where $\Phi$ is the standard normal CDF and $\beta=\left(\beta_0,\beta_1,\beta_2\right)^{T}$ is a the vector of parameters. The objective is to sample from the posterior distribution corresponding to the flat prior for $\beta$.
We adopt the standard Gibbs sampler with latent variables $\Psi_i\sim \mathcal{N}\left(x_i^T\beta,1\right)$ of which we consider only the sign.  To obtain draws from the posterior we repeat the following alternating two steps. First we sample from $\left.\beta\right\vert \psi \sim \mathcal{N}\left(\tilde{\beta}, \left(X^{T}X\right)^{-1}\right)$ with $\tilde{\beta}=\left(X^{T}X\right)^{-1}X^{T}\psi$ with $X$ the data matrix whose $i$-th rows is $x_i$. Then from $\psi_i\left\vert\beta,Y_i\right.\sim \mathcal{TN}\left(x_i^{T}\beta,1,Y_i\right)$ where $\mathcal{TN}\left(\mu,\sigma^2,Y\right)$ is a the normal distribution with mean $\mu$ and variance $\sigma^2$, truncated to be positive if $Y>0$ or negative otherwise. Further details of the algorithms can be found in \citep{CraiuMeng2005}. 

\begin{figure}[t]
\includegraphics[trim=80 30 80 30,clip,scale=0.6]{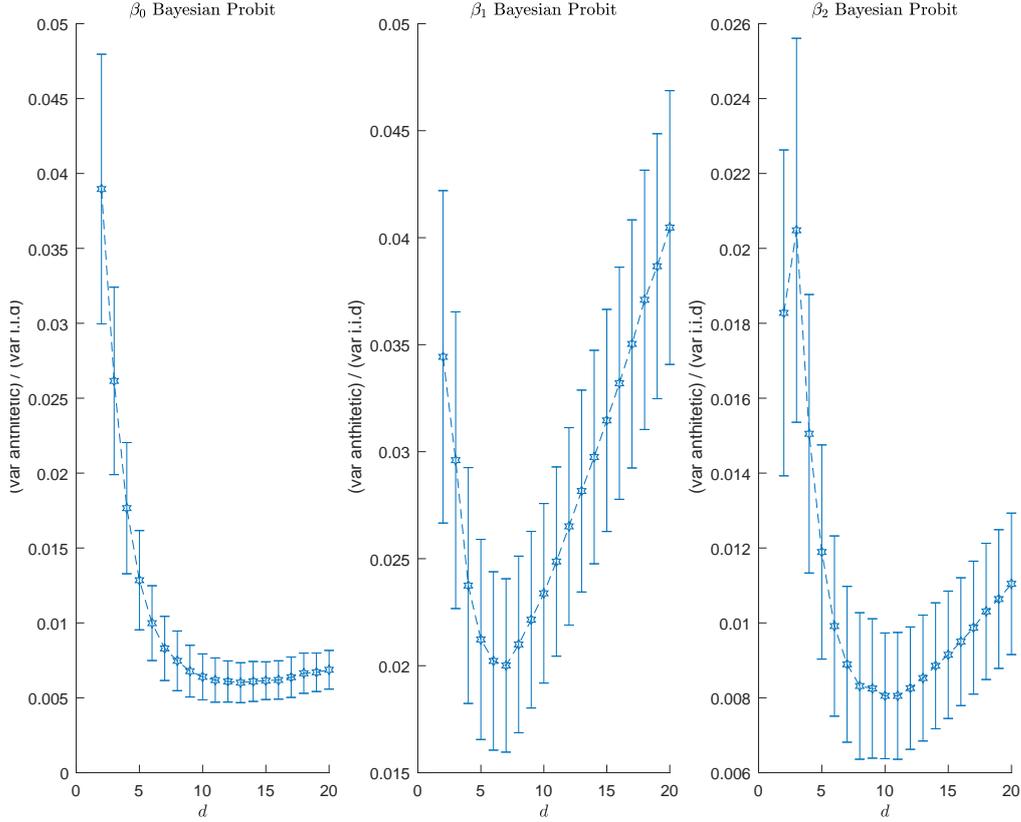}
\caption{Asymptotic variance ratio between antithetic and iid sampling (vertical axis) for different dimension $d$ (horizontal axis) of antithetic variates, and parameters of the Bayesian probit model (different plots). In each plot: the average ratio (dots) and its range (vertical segments). Note: all estimates are based on 1000 independently-replicated experiments. In each experiment 5000 Gibbs iterations with antithetic vectors of dimensions $d$ are used.}
\label{fig:BP}
\end{figure}
\subsubsection{Metropolis within Gibbs Hierarchical Poisson \cite{Gelfand1990}}
This second example concerns the  counts of failures $s=\left(s_1,\ldots,s_n\right)$ for $n=10$ pumps in a nuclear power  plant. The time $t=\left(t_1,\ldots,t_n\right)$ of operation are known. The model assumes $s_k\sim\mathcal{P}oi\left(\lambda_k t_k\right) $ and $\lambda_k\sim\mathcal{G}a\left(\alpha,\beta\right)$. The objective of the inference are draws from the posterior distribution of the parameters $\alpha$  to which we assign an exponential prior with mean 1  and $\beta$ with a Gamma prior distribution $\mathcal{G}a\left(0.1,1\right)$. Using the conjugate priors it is easy to obtain a Gamma distribution for $\lambda_1\left.\vert \lambda_2,\ldots\lambda_n,\beta\right.$ and those variables can be easily sampled using a Gibbs step. Sampling of $\alpha$ is slightly more difficult. In fact we have :
\begin{eqnarray*}
\mathbb{P}\left(\alpha\left.\vert \lambda_1,\ldots\lambda_n,\beta\right.\right)\propto \exp\left[\alpha\left(n\log\beta + \sum^n_{k=1}\log\lambda_k-1\right)-n\log\Gamma\left(\alpha\right)\right]
\end{eqnarray*}
We then sample $\alpha$ with a random walk Metropolis Hastings (MH) step with a deterministic scan. We show results for the case when we antithetically couple also the uniform draws for acceptance rejection choice and when we are not doing it. The former case is the one for which \cite{ Frigessi2000} report the worst performance of the usual two variates antithetic coupling of chains, in agreement with our results. 
\begin{figure}[t]
\includegraphics[trim=80 30 80 30,clip,scale=0.6]{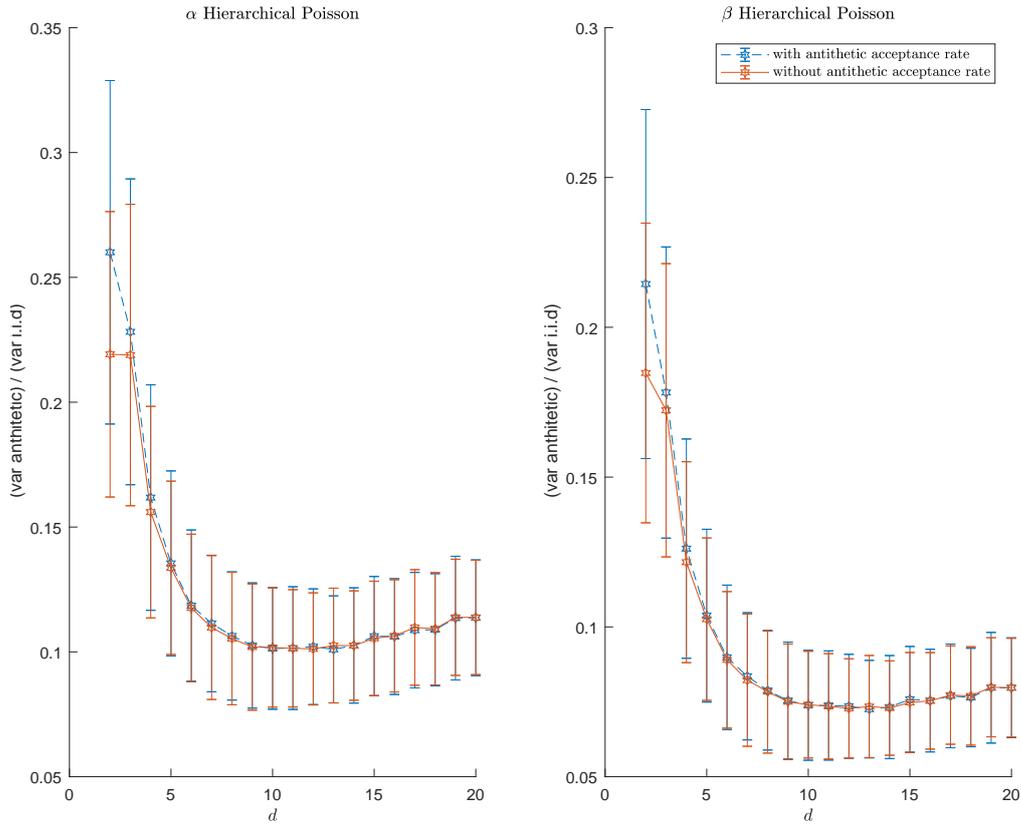}
\caption{Asymptotic variance ratio between antithetic and iid sampling (vertical axis) for different dimension $d$ (horizontal axis) of antithetic variates, and parameters of the Bayesian Hierarchical Poisson model (different plots). In each plot: the average ratio for the MH with (blue dots) and without (red dots) antithetic acceptance rule and their ranges (vertical segments). Note: all estimates are based on 1000 independently-replicated experiments. In each experiment 5000 Gibbs iterations with antithetic vectors of dimensions $d$ are used.}
\label{fig:HP}
\end{figure}

\subsubsection{Pseudo Marginal Metropolis Hastings Stochastic Volatility \cite{SQMC}}
The last application, target a state of the art methodology, the Pseudo Marginal MH (PCMH) proposed by \cite{andrieu2009} and which is able to estimate models with intractable likelihoods that are approximated using a particle filter. In particular, following \cite{SQMC}, we consider the bivariate stochastic volatility model introduced in \cite{Chan2006}:
\begin{eqnarray*}
\mathbf{y}_t &=& S_t^{1/2} \mathbf{\epsilon}_t \\
\mathbf{x}_t &=& \mathbf{\mu} + \Phi\left(\mathbf{x}_{t-1}-\mathbf{\mu}\right)+ \Psi^{1/2}\nu_t\\
S_t &=& \hbox{diag}\left(\exp\left(x_{1t},x_{2t}\right)\right)\\
\left(\mathbf{\epsilon}_t,\mathbf{\nu}_t\right)&\sim & \mathcal{N}\left(\mathbf{0}_4,C\right)
\end{eqnarray*}
with $\Phi$ and $\Psi$ diagonal matrices and $C$ a correlation matrix.
Following those authors we take the prior:
\begin{eqnarray}
\phi_{ii}&\sim & \mathcal{U}[0,1] \\
\dfrac{1}{\psi_{ii}}&\sim & \mathcal{G}a\left(10 \exp\left(-10\right)/, 10\exp\left(-3\right)\right)
\end{eqnarray}
where $\phi_{ii}$ and $\psi_{ii}$ denote respectively the diagonal elements of $\Phi$ and $\Psi$, and a flat prior for
$\mu$. In addition, we assume that C is uniformly distributed on the space of correlation matrices. To sample from the posterior distribution of the parameters, we use a Gaussian random-walk MH algorithm with covariance matrix calibrated by \cite{SQMC} so that the acceptance probability of the algorithm becomes, as $N$ tends to infinity, close to 25\%. We consider the mean-corrected daily returns on the Nasdaq and Standard and Poor's 500 indices for the period ranging from January 3rd, 2012, to
October 21st, 2013, so that the data set contains 452 observations. Figure \ref{fig:PMMH}
\begin{figure}[h!]
\centering
\includegraphics[trim=00 0 0 0,clip,scale=0.5]{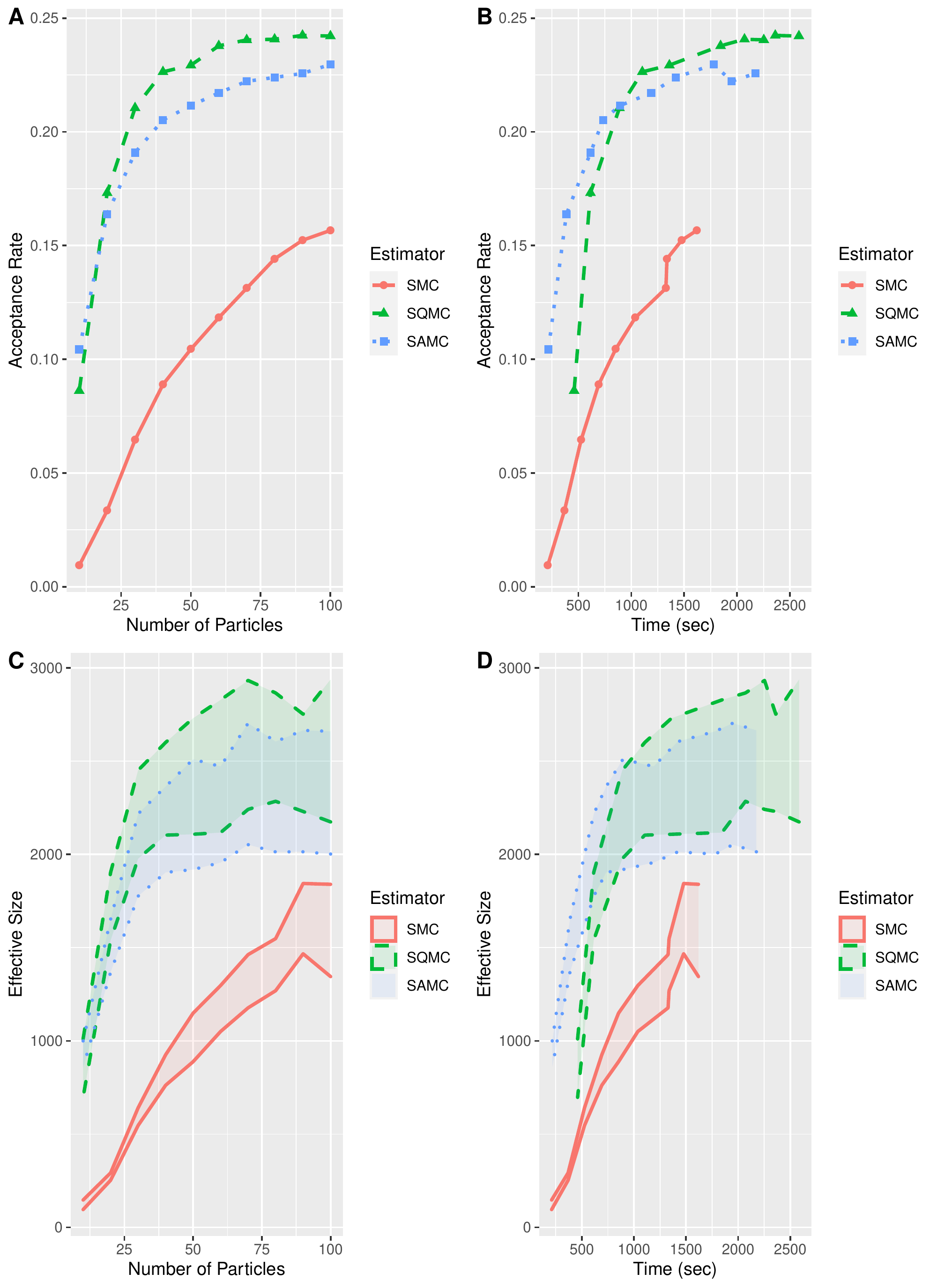}
\caption{PCMH using Sequential Monte Carlo (SMC), Sequential Quasi-Monte Carlo (SQMC) and Sequential Antithetic Monte Carlo (SAMC) (different colors). Acceptance rate of the Metropolis step (vertical axis) versus number of particles (horizontal axis, Panel A) and computing time (horizontal axis, Panel B). Maximum and minimum effective sample size (vertical axis) versus number of particles (horizontal axis, Panel C) and computing time (horizontal axis, Panel D).}
\label{fig:PMMH}
\end{figure}
PCMH algorithms using sequential quasi-Monte Carlo and antithetic Monte Carlo are equivalent in acceptance rate and effective sample size (Panels A and C) when a low number of particles (up to 20) is used. Nevertheless, antithetic Monte Carlo achieves larger acceptance rates (AR) and effective sample size (ESS) with a lower computing time (Panel B and D). When a larger number of particles is used (above 20), the performances are equivalent in terms of ESS, whereas SQMC is better in terms of AR.

\section{Discussion}\label{sec:discussion}
The development of antithetic constructions has generated a rich class of methods to accompany the evolution of Monte Carlo sampling algorithms. We enrich this class with a new antithetic method, the circulant variates (CCV), that satisfies the countermonotonicity, exchangeability, and marginal uniformity conditions. In particular, the marginal uniformity condition is linked to the Kullback-Leibler optimality.

The principle behind the proposal, relying on sampling on segments, leads to a unification of several classical antithetic constructions: rotation sampling, Latin hypercube, permuted displacement, and random balanced sampling.

Within this common framework, we evaluate theoretically the resulting distributions of the antithetic vectors and their concordance measures. The latter allows us to rank the methods within the class of sampling on segments. We also demonstrate a central limit theorem in the case of asymptotically increasing vector size.

Leveraging on iterative Latin Hypercube (LH) properties, we combine the two methods by using the CCV construction to initialize the LH construction and reduce the number of iterations. This reduces the simulation cost and improves the concordance lower bound. The numerical experiments include MCMC, Sequential MC, Quasi-MC, and classical MC integration. The methods proposed outperform standard implementations and are competitive with Quasi-Monte Carlo methods in low effective-dimension scenarios. In addition, Theorem \ref{CLT}
is confirmed by our numerical experiments, which show variance reduction at least as good as standard MC.

Future work includes possible extensions of the theory for KL optimality beyond the marginal univariate case. An investigation of the relationship between the line segment representation and orthogonal array-based latin hypercubes \cite{tang1993} could lead to an improvement of the performance for superposition dimensions bigger than 1. We notice that different methods that satisfy countermonotonicity, exchangeability, and marginal uniformity yield different variance reductions in practice. A more general and ambitious aim is to propose a mathematical framework
directed at identifying the additional features that produce these
differences. Our sampling method has been successfully applied within the Bayesian estimation framework of the European Commission's multi-country model \citep{albonico2019}. We expect that the proposed simulation technique will find direct application in other fields of computational mathematics and statistics.

\section*{Acknowledgements}
The views expressed by Lorenzo Frattarolo are the author's alone and do not necessarily correspond to those of the European Commission. We thank for the useful comments the participants at: the \textit{1st Italian-French Statistics Workshop}, University Ca' Foscari of Venice, October 2017, Venice, Italy; the \textit{32nd European Statisticians Meeting}, Palermo, Italy, 2019; the \textit{Workshop  "Monte Carlo methods, and approximate dynamic programming"}, ESSEC, October 2019, Cergy, France; the seminar series at the University of Surrey. Radu V. Craiu's research was supported by NSERC of Canada grant RGPIN-249547. Roberto Casarin acknowledges support from the Venice Centre for Risk Analytics (VERA) at the University Ca' Foscari of Venice.

\appendix
\setcounter{figure}{0}
\renewcommand{\thefigure}{\thesection.\arabic{figure}}
\section{Proof of Theorems and Lemmas}
 \subsection{Proof of Lemma \ref{notparrallel}}
\begin{proof}
By \eqref{simplevar}, conditionally on being on the $k$-th edge  $e_k=\left(i\left(k\right),j\left(k\right)\right)$,  $U_l$ is a standard uniform on $\left[\alpha_{lk},\beta_{lk}\right]$ if the interval has length bigger than 0 and a Dirac mass centered in $x_{lj\left(k\right)}$ otherwise i.e. has a density: 
\begin{eqnarray*} f\left(\right.u\left\vert K=k\right)=\left\{\begin{array}{ccc}\dfrac{1}{\left(\beta_{lk}-\alpha_{lk}\right)}\mathbb{I}_{\left\{ \left[\alpha_{lk},\beta_{lk}\right]\right\}}\left(u_l\right)& \hbox{if} & x_{li\left(k\right)} \neq x_{lj\left(k\right)}
\\&&\\
\mathbb{I}_{\left\{x_{lj\left(k\right)}\right\}}\left(u_l\right)& \hbox{if} & x_{li\left(k\right)} = x_{lj\left(k\right)}
\end{array} \right.\end{eqnarray*}

We get rid of  the case $x_{li\left(k\right)} = x_{lj\left(k\right)}$ by excluding segments contained in hyperplanes parallel to the hyperfaces of the hypercube, obtaining uniformity on $\left[\alpha_{lk},\beta_{lk}\right]$.

\end{proof}

\subsection{Proof of Theorem \ref{unif}}
Before stating the validity of the result in Theorem \ref{unif} we show the following lemma.
\begin{lemm}\label{unif1} 

Let $U_l$ be the $l$-th component of random vector $\mathbf{U}=(U_1,\ldots,U_d)$ in the stochastic representation \eqref{simplevar} and let $$\mathcal{K}_{u_l}=\left\{k \in \left\{1,\ldots, \left\vert\mathcal{E}\right\vert\right\}: u_l\in \left[a_{l,m_{\alpha}\left(k\right)},a_{l,m_{\beta}\left(k\right)}\right]
\right\}$$ for each $u_l\in \left[0,1\right]$.
Under Assumption \ref{assAdmiss} and \ref{assRange} if
\begin{eqnarray}\nonumber
 1 &=&\dfrac{1}{\left\vert\mathcal{E}\right\vert}\displaystyle \sum_{k\in \mathcal{K}_{u_l}}  \dfrac{1}{ a_{l,m_{\beta}\left(k\right)}-  a_{l,m_{\alpha}\left(k\right)}},\label{eqam}
\end{eqnarray} 
  then $U_l$ has standard uniform marginal density. If $u_l$ belongs to the element $A_{l,m}$ of the partition defined in Equations we get $\mathcal{K}_{u_l}=\mathcal{K}_{l,m}$.
\end{lemm}
\begin{proof}
When $u_l\in\left[0,1\right]$ the marginal standard uniform is obtained by equating the piece-wise uniform PDF of $U_l$ to $1$:
\begin{eqnarray*}
 f\left(u_l\right)&=& \dfrac{1}{\left\vert\mathcal{E}\right\vert}\displaystyle\sum^{\left\vert\mathcal{E}\right\vert}_{k= 1}
 f\left(\left. u_l\right\vert  K=k \right)
\\&=& \dfrac{1}{\left\vert\mathcal{E}\right\vert}\displaystyle\sum^{\left\vert\mathcal{E}\right\vert}_{k= 1} \dfrac{1}{\left(\beta_{lk}-\alpha_{lk}\right)}\mathbb{I}_{\left\{ \left[\alpha_{lk},\beta_{lk}\right]\right\}}\left(u_l\right)=1,
\end{eqnarray*}
which is satisfied by choosing $\mathcal{S}=(\mathcal{G},\mathbf{X})$ and $\alpha_{lk}$ and $\beta_{lk}$ as in the definition of admissibility of  given in Remark \ref{remadmis}.
\end{proof}
In the previous lemma the elements in the sum depend on $u_l$ and the equation \eqref{eqam} has to be checked for infinitely many cases. The proof of theorem \ref{unif} apply lemma \ref{unif1} and
the properties of the finite partition \eqref{partition} of the interval $\left[a_{l,1}=0,a_{l,n}=1\right]$, to show that  condition  \eqref{eqam} need to be verified only for a finite number of points.
\begin{proof}
Consider the interval $A_{l,m}=\left[a_{l,m-1},  a_{l,m}\right)$,
 There are two possibilities $A_{l,m}\cap \left[a_{l,m_{\alpha}\left(k\right)},a_{l,m_{\beta}\left(k\right)}\right]=\emptyset$, when $a_{l,m_{\alpha}\left(k\right)}>= a_{l,m}$ or $a_{l,m_{\beta}\left(k\right)}< a_{l,m-1}$, and $A_{l,m}\subseteq\left[a_{l,m_{\alpha}\left(k\right)},a_{l,m_{\beta}\left(k\right)}\right]$ when we have jointly  $a_{l,m_{\alpha}\left(k\right)}<= a_{l,m-1}$ and  $a_{l,m_{\beta}\left(k\right)}>= a_{l,m}$.

If $u_l \in A_{l,m}$ and $m_{\alpha}\left(k\right)+1\leq m\leq m_{\beta}\left(k\right)$ then $u_l \in \left[a_{l,m_{\alpha}\left(k\right)},a_{l,m_{\beta}\left(k\right)}\right]$. Thus $u_l \in A_{l,m}$ implies $\mathcal{K}_{u_l}=\mathcal{K}_{l,m}$ and $\mathcal{K}_{u_l}$ does not depend on $u_l$.
 
\begin{figure}

\center
\begin{tikzpicture}[x=0.25cm,y=1cm]

  \coordinate (2/m)  at ($(-4,0)+(0,2)$);
  \coordinate (2/m+1)  at ($(6,0)+(0,2)$);

  \coordinate (1/m)  at ($(-4,0)+(0,1)$);
  \coordinate (1/m+1)  at ($(6,0)+(0,1)$);
  
  \coordinate (0/m) at (-4,0);
  \coordinate (0/m+1) at (6,0);

  \foreach 
    \pt 
    in 
    {0,1,2} 
    { 
      \draw ($(\pt/m)-(4em,0)$) -- ($(\pt/m+1)+(4em,0)$);
    } 

  \node[draw,fill=white,inner sep=2pt,circle] (1/m)  at (1/m) {};
  \node[draw, fill=white, inner sep=2pt,circle]      (2/m+1) at (2/m+1) {};
  \node[draw,fill,inner sep=2pt,circle]      (0/m+1) at (0/m+1) {};
  \node[draw,fill,inner sep=2pt,circle]      (0/m) at (0/m) {};

  \draw[line width=1.5pt,arrows=->]      (1/m)  --  ($(1/m)-(4em,0)$);
  \draw[line width=1.5pt,arrows=->]       (2/m+1) -- ($(2/m+1)+(4em,0)$);
  \node at ($(2/m+1)+(-60:3ex)$) {$a_{l,m_{\beta}\left(k\right)}$};  
  \node at ($(0/m+1)+(-60:3ex)$)    {$a_{l,m}$};  
  \node at ($(1/m)+(-140:4ex)$) {$a_{l,m_{\alpha}\left(k\right)}$};
  \node at ($(0/m)+(-140:4ex)$)     {$a_{l,m-1}$};
\end{tikzpicture} 
 \end{figure}

Let us show that if \eqref{eqam} is valid for the all the $A_{l,m}$, then it is valid for every other partition of the unit interval.  We start by verifying that if condition \eqref{eqam} is verified by each disjoint interval in the partition it is also verified by intervals from coarser partitions. Consider partitions where we join two different intervals $B= A_{l,m}\cup A_{l,m+1}$. We should verify \eqref{eqam} for all $u_l\in B$ but since the intervals are disjoint it suffice to verify separately for all $u_l\in A_{l,m}$ and for all $u_l\in A_{l,m+1}$ and those conditions are satisfied by assumption. Let us show, now, that if condition \eqref{eqam} is verified by each disjoint interval in the partition it is also verified by intervals from finer partitions. Lets divide $A_{l,m}$ in two disjoint subsets by choosing an arbitrary point $a_{*}\in A_{l,m}$.
\begin{equation*}
A^1_{l,m}=\left[a_{l,m},  a_{*}\right),\quad
A^2_{l,m}=\left[a_{*},  a_{l,m+1}\right).
\end{equation*}
Let us call $\mathcal{K}_i=\mathcal{K}_{u_l}$ when $u_l\in A^i_{l,m}$, $i=1,2$.

Since no segment in $\mathcal{S}$ has a vertex with $l$-th coordinate $a_{*}$, when $k\in\mathcal{K}_{l,m}$,  we have  $a_{l,m_{\alpha}\left(k\right)}\leq a_{l,m}< a_{*} $ and  $a_{l,m_{\beta}\left(k\right)}\geq a_{l,m+1}>a_{*}$. Then, we have jointly  $a_{l,m_{\alpha}\left(k\right)}< a_{*}$ and  $a_{l,m_{\beta}\left(k\right)}\leq a_{l,m+1}$  and $a_{l,m_{\alpha}\left(k\right)}\leq a_{l,m}$ and  $a_{l,m_{\beta}\left(k\right)}> a_{*}$ . This means:
\begin{eqnarray*}
\mathcal{K}_{1}=\mathcal{K}_{2}=\mathcal{K}_{l,m},
\end{eqnarray*}
verifying \eqref{eqam} for $A^1_{l,m}$ and $A^2_{l,m}$ i.e. for the finer partition. Every other partition can be obtained by iteratively splitting and joining intervals. 

Finally we show that condition \eqref{eqam} is satisfied on the last interval $A_{l,n_l}$ if it is verified on every other interval and that the probability mass is 1. Since $U_l$ is a mixture of uniform random variables on the range $\left[0,1\right]$ the PDF integrates to one and exploiting $\left[0,1\right]=(\left[0,1\right]/A_{l,n_l})\cup A_{l,n_l}$ one obtains:
\begin{eqnarray}
1= \int^{1}_{0}f\left(u_l\right)du_l = \int^{a_{l,n_l-1}}_{0}f\left(u_l\right)du_l +\int^{1}_{a_{l,n_l-1}}f\left(u_l\right)du_l
\end{eqnarray}
If \eqref{eqam} is verified for $\left\{A_{l,m}\right\}^{n_l-1}_{m=2}$, then:
\begin{eqnarray*}
\int^{a_{l,n_l-1}}_{0}f\left(u_l\right)du_l= a_{l,n_l-1},\\
1= a_{l,n_l-1}  + \int^{1}_{a_{l,n_l-1}}\dfrac{1}{\left\vert\mathcal{E}\right\vert}\displaystyle \sum_{k\in \mathcal{K}_{l,n_l}}  \dfrac{1}{ 1-  a_{l,m_{\alpha}\left(k\right)}}du_l\\
1= a_{l,n_l-1}  + \left(1 -a_{l,n_l-1}\right)\dfrac{1}{\left\vert\mathcal{E}\right\vert}\displaystyle \sum_{k\in \mathcal{K}_{l,n_l}}  \dfrac{1}{ 1-  a_{l,m_{\alpha}\left(k\right)}},
\end{eqnarray*}
or equivalently
\begin{eqnarray*}
\dfrac{1}{\left\vert\mathcal{E}\right\vert}\displaystyle \sum_{k\in \mathcal{K}_{l,n_l}}  \dfrac{1}{ 1-  a_{l,m_{\alpha}\left(k\right)}}=1.
\end{eqnarray*}

\end{proof}

\subsection{Proof of Theorem \ref{theoopt}}

In the proof we will use the following lemma.
\begin{lemm}\label{lemmsym} The system in \eqref{uniformseg} is equivalent to the system :

\begin{eqnarray}\nonumber\label{sysunisym}
G_{l,m}\left( \mathbf{a}_l\right)&=& \displaystyle -\dfrac{1}{\left\vert\mathcal{E}\right\vert}\sum^{n_l}_{m^{\prime}=1}\dfrac{n^l_{\left(m,m^{\prime}\right)}}{\left\vert a_{l,m^{\prime}}-  a_{l,m}\right\vert}=0\\
m&=&2,\ldots,n_l-1\\\nonumber
\end{eqnarray}
\end{lemm}

\begin{proof}
Recall the definition in Eq. \eqref{uniformseg}
$$
F_{l,m}\left( \mathbf{a}_l\right)=\dfrac{1}{\left\vert\mathcal{E}\right\vert }\displaystyle \sum_{k\in \mathcal{K}_{l,m}}  \dfrac{1}{ a_{l,m_{\beta}\left(k\right)}-  a_{l,m_{\alpha}\left(k\right)}}-1,
$$
then $G_{l,m}(\mathbf{a}_l)$ writes as:
\begin{eqnarray}\label{FG}\nonumber
G_{l,m}\left( \mathbf{a}_l\right)&=&F_{l,m}\left( \mathbf{a}_l\right)-F_{l,m+1}\left( \mathbf{a}_l\right) \\
m&=&2,\ldots,n_l-1\\\nonumber
G_{l,n_l}\left( \mathbf{a}_l\right)&=&   F_{l,n_l}\left( \mathbf{a}_l\right).
\end{eqnarray}

Then: 
\begin{eqnarray*}
G_{l,m}\left( \mathbf{a}_l\right)&=&\dfrac{1}{\left\vert\mathcal{E}\right\vert }\displaystyle \sum_{k\in \mathcal{K}_{l,m}}  \dfrac{1}{ a_{l,m_{\beta}\left(k\right)}-  a_{l,m_{\alpha}\left(k\right)}}-\dfrac{1}{\left\vert\mathcal{E}\right\vert }\sum_{k\in \mathcal{K}_{l,m+1}}  \dfrac{1}{ a_{l,m_{\beta}\left(k\right)}-  a_{l,m_{\alpha}\left(k\right)}}=0.
\end{eqnarray*}

\begin{figure}[t]\label{fig:sumcancellation}
\centering

\begin{tikzpicture}[x=0.25cm,y=1cm]

 \coordinate (4/m-1)  at ($(-12,0)+(0,4)$);
  \coordinate (4/m)  at ($(-4,0)+(0,4)$);
  \coordinate (4/m+1)  at ($(6,0)+(0,4)$);

 \coordinate (3/m-1)  at ($(-12,0)+(0,3)$);
  \coordinate (3/m)  at ($(-4,0)+(0,3)$);
  \coordinate (3/m+1)  at ($(6,0)+(0,3)$);
  
  \coordinate (2/m-1)  at ($(-12,0)+(0,2)$);
  \coordinate (2/m)  at ($(-4,0)+(0,2)$);
  \coordinate (2/m+1)  at ($(6,0)+(0,2)$);
  
  \coordinate (1/m-1)  at ($(-12,0)+(0,1)$);
  \coordinate (1/m)  at ($(-4,0)+(0,1)$);
  \coordinate (1/m+1)  at ($(6,0)+(0,1)$);
  
   \coordinate (0/m-1)  at ($(-12,0)+(0,0)$);  
  \coordinate (0/m) at (-4,0);
  \coordinate (0/m+1) at (6,0);

  \foreach 
    \pt 
    in 
    {0,1,2,3,4} 
    { 
      \draw ($(\pt/m-1)-(4em,0)$) -- ($(\pt/m+1)+(4em,0)$);
    } 

  \node[draw,fill,inner sep=2pt,circle] (1/m)  at (1/m) {};
  \node[draw,fill=white,inner sep=2pt,circle]      (2/m+1) at (2/m+1) {};
  \node[draw,fill,inner sep=2pt,circle] (3/m-1)  at (3/m-1) {};
  \node[draw,fill=white,inner sep=2pt,circle]      (4/m) at (4/m) {};
  \node[draw,fill,inner sep=2pt,circle]      (0/m+1) at (0/m+1) {};
  \node[draw,fill,inner sep=2pt,circle]      (0/m) at (0/m) {};
    \node[draw,fill,inner sep=2pt,circle]      (0/m-1) at (0/m-1){};

  \draw[line width=1.5pt,arrows=->]      (1/m)  --  ($(1/m-1)-(4em,0)$);
  \draw[line width=1.5pt,arrows=->]       (2/m+1) -- ($(2/m+1)+(4em,0)$);
   \draw[line width=1.5pt,arrows=->]      (3/m-1)  --  ($(3/m-1)-(4em,0)$);
  \draw[line width=1.5pt,arrows=->]       (4/m) -- ($(4/m+1)+(4em,0)$);  
  
  \node at ($(2/m+1)+(-60:3ex)$) {$a_{l,m_{\beta}\left(k\right)}$};  
  \node at ($(0/m+1)+(-60:3ex)$)    {$a_{l,m+1}$};  
  \node at ($(4/m)+(-60:3ex)$) {$a_{l,m_{\beta}\left(k\right)}$};  
  \node at ($(0/m-1)+(-60:3ex)$)    {$a_{l,m-1}$};  
  \node at ($(1/m)+(-140:4ex)$) {$a_{l,m_{\alpha}\left(k\right)}$};
    \node at ($(3/m-1)+(-140:4ex)$) {$a_{l,m_{\alpha}\left(k\right)}
    $};
    \node at (-20,3.5) {$\mathcal{K}_{l,m}$};
        \node at (-20,1.5) {$\mathcal{K}_{l,m+1}$};

  \node at ($(0/m)+(-140:4ex)$)     {$a_{l,m}$};
\end{tikzpicture}
\caption{Computation of $G_{l,m}$}
\end{figure}
In Figure \ref{fig:sumcancellation} the terms corresponding to overlapping arrows are equal in the in the two sums  and cancel out.

Let us define the set of edges starting and terminating in $a_{l,m}$:
\begin{eqnarray*}
\mathcal{K}^{\alpha}_{l,m}=\left\{ k\in \left\{1,\ldots,\left\vert\mathcal{E}\right\vert\right\}: m_{\alpha}\left(k\right)=m\right\}\\
\mathcal{K}^{\beta}_{l,m}=\left\{ k\in \left\{1,\ldots,\left\vert\mathcal{E}\right\vert\right\}: m_{\beta}\left(k\right)=m\right\}
\end{eqnarray*}

We obtain:
\begin{eqnarray*}
G_{l,m}\left( \mathbf{a}_l\right)&=&\displaystyle -\dfrac{1}{\left\vert\mathcal{E}\right\vert }\sum_{k\in \mathcal{K}^{\alpha}_{l,m}}  \dfrac{1}{ a_{l,m_{\beta}\left(k\right)}-  a_{l,m}}+\dfrac{1}{\left\vert\mathcal{E}\right\vert }\sum_{k\in \mathcal{K}^{\beta}_{l,m}}  \dfrac{1}{ a_{l,m}-  a_{l,m_{\alpha}\left(k\right)}}\\
&=& -\dfrac{1}{\left\vert\mathcal{E}\right\vert }\sum^{n_l}_{m^{\prime}=m+1}  \dfrac{n^l_{\left(m,m^{\prime}\right)}}{ a_{l,m^{\prime}}-  a_{l,m}}+\dfrac{1}{\left\vert\mathcal{E}\right\vert } \sum^{m-1}_{m^{\prime}=1}  \dfrac{n^l_{\left(m,m^{\prime}\right)}}{ a_{l,m}-  a_{l,m^{\prime}}}\\
&=& -\dfrac{1}{\left\vert\mathcal{E}\right\vert }\sum^{n_l}_{m^{\prime}=1}  \dfrac{n^l_{\left(m,m^{\prime}\right)}}{ \left\vert a_{l,m^{\prime}}-  a_{l,m}\right\vert}.
\end{eqnarray*}


If we call $\mathbf{F}_{l}$ and $\mathbf{G}_{l}$ the vectors of components $F_{l,m}$ and $G_{l,m}$ ,$m\in\left\{1,\ldots,n_l\right\}$ , by \eqref{FG} we have $\mathbf{G}_{l}= T_{l}\mathbf{F}_{l}$ with 
 $T_l$ as the $\left[n_l \times n_l \right]$ matrix with the only non zero elements
\begin{eqnarray*}
\begin{array}{ccccc}
T_{l\, i,i}=1,&
T_{l\, i,i+1}=-1
\end{array}
\end{eqnarray*}
It is easy to show that $T_l$ has full rank $n_l$. The two systems are then equivalent because of the linear independence of rows of $T_l$. 
Finally by the same argument at the end of the proof of Theorem \ref{unif} we can get rid of the $n_l$-th equation.
\end{proof}

We are ready to present the proof of of Theorem \ref{theoopt}.
\begin{proof}
By Lemma \ref{lemmsym} solutions to $\mathbf{F}_{l}=\mathbf{0}$ are also solution to $\mathbf{G}_{l}=\mathbf{0}$. We, then, look for a function whose gradient is equivalent to  $\mathbf{G}_{l}$. This requirement can be expressed in the language of differential forms.
A differential form $\boldsymbol{\alpha}$ is closed if its exterior derivative is zero $d\boldsymbol{\alpha}=0$.  A differential form $\boldsymbol{\alpha}$ is exact if exist a differential form $\boldsymbol{\beta}$ of degree less than the one of $\boldsymbol{\alpha}$ such that $\boldsymbol{\alpha}=d\boldsymbol{\beta}$.
If the $1$-form:
\begin{eqnarray}
\boldsymbol{\gamma}_l = \sum^{n_l-1}_{m=2} G_{l,m}\left(\mathbf{a}_l\right)da_{l,m}
\end{eqnarray}
is exact in some domain, there exists a $0$-form i.e. a smooth function $\Psi_l$ such that
\begin{eqnarray}
d\Psi_l = \sum^{n_l-1}_{m=2} \dfrac{\partial \Psi_l}{\partial a_{l,m}}\left(\mathbf{a}_l\right)da_{l,m}= \boldsymbol{\gamma}_l
\end{eqnarray} 
i.e.
\begin{eqnarray}\label{partialpsi}
G_{l,m} &=&  \dfrac{\partial \Psi_l}{\partial a_{l,m}}\left(\mathbf{a}_l\right),\quad m\in\left\{2,\ldots,n_l-1\right\}.
\end{eqnarray} 

 Let us pick an open ball $B$ in $\mathbb{R}^{n_l-2}$ not including the points $\left(a_{l,j}=a_{l,i}\right)$, $i,j\in\left\{1,\ldots,n_l\right\}$ where $\mathbf{G}_l$ is not defined. By Poincare lemma if  $\boldsymbol{\gamma}_l$ is closed in  $B$ is also exact in $B$. Thus if
\begin{eqnarray*}
d\boldsymbol{\gamma}_l = \sum^{n_l-1}_{m,m^{\prime}=2} \left\{\dfrac{\partial G_{l,m}}{\partial a_{l,m^{\prime}}}-\dfrac{\partial G_{l,m^{\prime}}}{\partial a_{l,m}}\right\}da_{l,m}da_{l,m^{\prime}}=0
\end{eqnarray*}
then
\begin{eqnarray*}
\left\{\dfrac{\partial G_{l,m}}{\partial a_{l,m^{\prime}}}-\dfrac{\partial G_{l,m^{\prime}}}{\partial a_{l,m}}\right\}=0,
\end{eqnarray*}
and there exists $\Psi_l$ such that $G_{l,m}$ is its gradient. We have:
\begin{eqnarray*}
H_{i,i}=\dfrac{\partial G_{l,i}}{\partial a_{l,i}}&=&\displaystyle \sum^{n_l}_{m=1}  \dfrac{n^l_{\left(i,m\right)}}{ \left(a_{l,m}-  a_{l,i}\right)^2}\\
H_{i,j}=\dfrac{\partial G_{l,i}}{\partial a_{l,j}} &=&-\displaystyle  \dfrac{n^l_{\left(i,j\right)}}{ \left(a_{l,i}-  a_{l,j}\right)^2}
\end{eqnarray*}
$i\neq j\in\left\{2,\ldots,n_l-1\right\}$. Since $n^l_{\left(i,j\right)}=n^l_{\left(j,i\right)}$ $\boldsymbol{\gamma}_l$ is exact whenever $\mathbf{G}_l$ is defined and $\Psi_l$ exists. After some algebra we obtain:
\begin{eqnarray*}
\Psi_l\left(\mathbf{a}_l\right)  &=&-\dfrac{1}{2 \left\vert\mathcal{E}\right\vert}\displaystyle \sum^{n_l}_{ m,m^{\prime}=1}n^l_{\left(m,m^{\prime}\right)} \log\left\vert a_{l,m^{\prime}}-  a_{l,m}\right\vert. 
\end{eqnarray*}

Since $H$ is a symmetric diagonally dominant matrix, it is everywhere semi-positive definite, this is equivalent to the convexity of $\Psi_l\left(\mathbf{a}_l\right)$. In addition, we are restricting the domain to the unit hypecube that is a convex domain so the problem is a convex one.

 \end{proof}
\subsection{Proof of Theorem \ref{KL}}
\begin{proof}
In order to apply Radon--Nikodym theorem $\mathbb{P}$ needs to be absolutely continuous with respect to $\mathbb{Q}$.
This assumption amounts to impose the restrictions $$\displaystyle a_{l,1}=\min_{k\in \left\{1,\ldots,\left\vert \mathcal{E}\right\vert\right\}}\alpha_{l,k}=0$$
and
$$\displaystyle a_{l,n_l}=\max_{k\in \left\{1,\ldots,\left\vert \mathcal{E}\right\vert\right\}}\beta_{l,k}=1$$ (Assumption \ref{assRange}). We have: 
 \begin{eqnarray*}
 \mathbb{P}\left(u_l\in A, k \in B\right)&=& \sum_{k\in B} \int_{u_l\in A} \dfrac{1}{\left\vert\mathcal{E}\right\vert}f\left(u_l\right.\left\vert K=k\right)d u_l\\&=& \sum_{k\in B} \int_{u_l\in A} \dfrac{d\mathbb{P}}{d\mathbb{Q}}\left(u_l,k\right) \dfrac{1}{\left\vert\mathcal{E}\right\vert} \mathbb{I}_{\left[0,1\right]}\left(u_l\right)du_l
 \end{eqnarray*}
By the condition on the a's above, we have $f\left(u_l\right\vert \left.K=k\right)=f\left(u_l\right\vert \left.K=k\right)\mathbb{I}_{\left[0,1\right]}\left(u_l\right)$ and we can read off the derivative :
\begin{eqnarray*}
\dfrac{d\mathbb{P}}{d\mathbb{Q}}\left(u_l,k\right) = f\left(u_l\right.\left\vert K=k\right)
\end{eqnarray*}
\begin{eqnarray*}
D_{KL}\left(\mathbb{P}\vert\vert \mathbb{Q}\right)&=& \mathbb{E}_{\mathbb{P}}\left[\log\left(\dfrac{d\mathbb{P}}{d\mathbb{Q}}\left(u_l,k\right)\right)\right]\\&=& \displaystyle\sum^{\left\vert\mathcal{E}\right\vert}_{k= 1} \int^{1}_{0} \log\left(f\left(u_l\right.\left\vert K=k\right)\right) \dfrac{1}{\left\vert\mathcal{E}\right\vert} f\left(u_l\right.\left\vert K=k\right)du_l
\\&=& \displaystyle\sum^{\left\vert\mathcal{E}\right\vert}_{k= 1} \int^{1}_{0} \log\left(\dfrac{1}{\left(\beta_{lk}-\alpha_{lk}\right)}\mathbb{I}_{\left\{ \left[\alpha_{lk},\beta_{lk}\right]\right\}}\left(u_l\right)\right) \dfrac{1}{\left\vert\mathcal{E}\right\vert}\dfrac{1}{\left(\beta_{lk}-\alpha_{lk}\right)}\mathbb{I}_{\left\{ \left[\alpha_{lk},\beta_{lk}\right]\right\}}\left(u_l\right)du_l
\\&=& \displaystyle\sum^{\left\vert\mathcal{E}\right\vert}_{k= 1} \int^{1}_{0} \log\left(\dfrac{1}{\left(\beta_{lk}-\alpha_{lk}\right)}\right) \dfrac{1}{\left\vert\mathcal{E}\right\vert}\dfrac{1}{\left(\beta_{lk}-\alpha_{lk}\right)}\mathbb{I}_{\left\{ \left[\alpha_{lk},\beta_{lk}\right]\right\}}\left(u_l\right)du_l
\\&+& \displaystyle\sum^{\left\vert\mathcal{E}\right\vert}_{k= 1} \int^{1}_{0} \log\left(\mathbb{I}_{\left\{ \left[\alpha_{lk},\beta_{lk}\right]\right\}}\left(u_l\right)\right) \dfrac{1}{\left\vert\mathcal{E}\right\vert}\dfrac{1}{\left(\beta_{lk}-\alpha_{lk}\right)}\mathbb{I}_{\left\{ \left[\alpha_{lk},\beta_{lk}\right]\right\}}\left(u_l\right)du_l
\end{eqnarray*}
As usually done in this kind of computation we assume by convention that $0\log\left(0\right)=0$, then:
\begin{eqnarray*}
D_{KL}\left(\mathbb{P}\vert\vert \mathbb{Q}\right)&=& \displaystyle\sum^{\left\vert\mathcal{E}\right\vert}_{k= 1} \int^{1}_{0} \log\left(\dfrac{1}{\left(\beta_{lk}-\alpha_{lk}\right)}\right) \dfrac{1}{\left\vert\mathcal{E}\right\vert}\dfrac{1}{\left(\beta_{lk}-\alpha_{lk}\right)}\mathbb{I}_{\left\{ \left[\alpha_{lk},\beta_{lk}\right]\right\}}\left(u_l\right)du_l
\\&=& \dfrac{1}{\left\vert\mathcal{E}\right\vert} \displaystyle\sum^{\left\vert\mathcal{E}\right\vert}_{k= 1} -\log\left(\left(\beta_{lk}-\alpha_{lk}\right)\right)
\\&=& \dfrac{1}{\left\vert\mathcal{E}\right\vert} \displaystyle\sum^{\left\vert\mathcal{E}\right\vert}_{k= 1} -\log\left(\left(a_{l,m_{\beta}\left(k\right)}-a_{l,m_{\alpha}\left(k\right)}\right)\right)
\\&=& \dfrac{1}{\left\vert\mathcal{E}\right\vert} \displaystyle\sum_{m>m^{\prime}} -n_{\left(m,m^{\prime}\right)}\log\left(a_{l,m}-a_{l,m^{\prime}}\right)
\\&=& \dfrac{1}{2\left\vert\mathcal{E}\right\vert} \displaystyle\sum^{n_l}_{m,m^{\prime}=1} -n_{\left(m,m^{\prime}\right)}\log\left\vert a_{l,m}-a_{l,m^{\prime}}\right\vert.
\end{eqnarray*}

\end{proof}

\subsection{Proof of Lemma \ref{detcompunif}}
\begin{proof}
Since $\mathbf{U}$ is uniform on $\mathcal{S}^{(x)}$ the random variables $\left\{U_l,l\in\mathcal{D}\right\}$ are standard uniforms, and since $\mathbf{Y}$ is a solution to the standard uniform on $\mathcal{S}^{(y)}=(\mathcal{G}^{(y)},\mathbf{Y})$, the $\left\{W_l, l\in\mathcal{D}\right\}$ are standard uniform as well.
$W_l= \left( y_{lI}x_{lI^{\prime}} + y_{lJ}\left(1-x_{lI^{\prime}}\right)\right)V\\
+\left( y_{lI}x_{lJ^{\prime}} + y_{lJ}\left(1-x_{lJ^{\prime}}\right)\right)\left(1-V\right)$
and $\mathbf{W}$ and $\left(I^{\prime},J^{\prime}\right)$ are sampled with uniform probability from $\mathcal{E}^1$
\end{proof}

\subsection{Proof of Lemma\ref{CTMcomp}}
\begin{proof}
\begin{eqnarray*}
 \dfrac{d}{2}= \sum^{d}_{l=1} W_l&=&  y_{lI} U_l + \left(1-U_l\right)y_{lJ}\\
&=& \sum^{d}_{l=1} \left(y_{lI} - y_{lJ}\right) U_l + \sum^{d}_{l=1} y_{lJ}\\
&=& c_1\sum^{d}_{l=1} U_l + c_2= \dfrac{d}{2} c_1 + c_2
\end{eqnarray*}
and $\sum^{d}_{l=1} W_l=\dfrac{d}{2}$ if 
\begin{equation}
\left(1-c_1\right) \dfrac{d}{2}=c_2.
\end{equation}
\end{proof}

\subsection{Proof of Lemma \ref{stochcomp}}
\begin{proof}

 The coordinates of the vertexes (i.e. the columns of $\left(\mathbf{X}^1,\mathbf{X}^2\right)$) inherit the constant sum from \eqref{segdcm} and allow us to define the following subsets of $\mathcal{E}^{t}$
\begin{eqnarray*}
\mathcal{K}^1_{u_l}&\equiv&\left\{k \in \left\{1,\ldots, \left\vert\mathcal{E}^1\right\vert\right\}: u_l\in \left[a_{l,m_{\alpha}\left(k\right)},a_{l,m_{\beta}\left(k\right)}\right]
\right\}\\
\mathcal{K}^2_{u_l}&\equiv&\left\{k \in \left\{1,\ldots, \left\vert\mathcal{E}^1\right\vert\right\}: u_l\in \left[a_{l,m_{\alpha}\left(k\right)},a_{l,m_{\beta}\left(k\right)}\right]
\right\}\\
\mathcal{K}_{u_l}&\equiv&\left\{k \in \left\{1,\ldots, \left\vert\mathcal{E}^1\sqcup \mathcal{E}^2\right\vert\right\}: u_l\in \left[a_{l,m_{\alpha}\left(k\right)},a_{l,m_{\beta}\left(k\right)}\right]
\right\}\\ &=& \mathcal{K}^1_{u_l}\sqcup\mathcal{K}^2_{u_l}.
\end{eqnarray*} 
The sufficient conditions for standard marginal uniformity on $\mathcal{S}^{\prime}$ given in \eqref{uniformseg} become 
\begin{eqnarray*}
&&\dfrac{1}{\left\vert\mathcal{E}^1\sqcup \mathcal{E}^2\right\vert }\displaystyle \sum_{k\in \mathcal{K}_{u_l} } \dfrac{1}{ a_{l,m_{\beta}\left(k\right)}-  a_{l,m_{\alpha}\left(k\right)}}=\\&=&
\dfrac{\left\vert\mathcal{E}^1\right\vert}{\left\vert\mathcal{E}^1\sqcup \mathcal{E}^2\right\vert }  \displaystyle \dfrac{1}{\left\vert\mathcal{E}^1\right\vert } \sum_{k\in \mathcal{K}^1_{u_l}}  \dfrac{1}{ a_{l,m_{\beta}\left(k\right)}-  a_{l,m_{\alpha}\left(k\right)}} \\
&+& \dfrac{\left\vert \mathcal{E}^2\right\vert}{\left\vert\mathcal{E}^1\sqcup \mathcal{E}^2\right\vert }\displaystyle \dfrac{1}{\left\vert\mathcal{E}^2\right\vert }
\sum_{k\in \mathcal{K}^2_{u_l}}  \dfrac{1}{ a_{l,m_{\beta}\left(k\right)}-  a_{l,m_{\alpha}\left(k\right)}}\\
&=&\dfrac{\left\vert \mathcal{E}^1\right\vert+\left\vert \mathcal{E}^2\right\vert}{\left\vert\mathcal{E}^1\sqcup \mathcal{E}^2\right\vert }=1.
\end{eqnarray*}
\end{proof}

\subsection{Proof of Lemma \ref{permdctm}}
\begin{proof}
The random vector $\mathbf{W}\in [0,1]^{d}$ with elements
\begin{eqnarray}
W_{1}&=& x_{{\pi\left(1\right)}I} V + \left(1-V\right)x_{{\pi\left(1\right)}J}\nonumber\\
&\vdots&\\
W_{d}&=& x_{{\pi\left(d\right)}I} V + \left(1-V\right)x_{{\pi\left(d\right)}J}\nonumber
\end{eqnarray}
defines a permutation of $\mathbf{U}$.
If $\mathbf{U}$ is strict $d$-CTM on $\mathcal{S}=(\mathcal{G},\mathbf{X})$ then the permuted vector is strict $d$-CTM on $\mathcal{S}^{\prime}=(\mathcal{G},P_{\pi}\mathbf{X})$, because the sum is preserved by permutations. 
\end{proof}

\subsection{Proof of Corollary \ref{genunif}}
\begin{proof}
Follows the same lines as the proof of Theorem \ref{unif} by replacing the common $V$ by the coordinate-specific $V_l\sim\mathcal{U}[0,1]$.
\end{proof}

\subsection{Proof of Theorem \ref{th:condjoint}}
\begin{proof}
Let us define the random variables
\begin{eqnarray}
V_{l,k}&=& \left\lbrace\begin{array}{ccc}  V_l & & \hbox{if }l \in \mathcal{L}^{+}_{k}\\ &&\\
1-V_l & & \hbox{if } l \in \mathcal{L}^{-}_{k}\end{array}\right.,
\end{eqnarray}
the event $\left\{U_l\leq u_l \left\vert K=k\right.\right\}$ can be rewritten
\begin{eqnarray*}\label{preimage}
\left\{U_l\leq u_l \left\vert K=k\right.\right\}&=&\left\{x_{li\left(k\right)} V_l + \left(1-V_l\right)x_{lj\left(k\right)} \leq u_l \left\vert K=k\right.\right\} \\&=& \left\{V_{l,k} \leq F_{U_l\left\vert K\right.}\left(u_l;k\right)\left\vert K=k\right.\right\}
\end{eqnarray*}
Let $v_{l,k}=F_{U_l\left\vert K\right.}\left(u_l;k\right)$ for $l=1,\ldots,d$, $k=1,\ldots,\left\vert\mathcal{E} \right\vert$ 
then:
\begin{eqnarray}\nonumber
& &\mathbb{P}\left(U_1 \leq u_1,\ldots,U_d\leq u_d\left\vert K=k \right.\right) = \mathbb{P}\left(V_{1,k} \leq v_{1,k},\ldots,V_{d,k} \leq v_{d,k}\left\vert K=k\right.\right)
\\\nonumber &=&\mathbb{P}\left(  \left.\left(\bigcap_{l\in \mathcal{L}^{-}_{k}} \left\{\left(1-V_l\right)\leq v_{l,k} \right\}\right)\cap \left(\bigcap_{l\in \mathcal{L}^{+}_{k}} \left\{V_l\leq v_{l,k}\right\}\right)\right\vert K=k\right)
\\\nonumber &=&\mathbb{E}\left[\prod_{l\in \mathcal{L}^{-}_{k}}\mathbb{I}_{\left(1-v_{l,k},1\right]}\left(V_l\right)\prod_{l\in \mathcal{D}\setminus \mathcal{L}^{-}_{k}}\mathbb{I}_{\left[0,v_{l,k}\right]}\left(V_l\right)\right]
\\\nonumber&=&\mathbb{P}\left(R_{\mathcal{L}^{-}_{k},\frac{1}{2}}(\mathbf{V})\leq \mathbf{v}_k\right).
\end{eqnarray}

\end{proof}

\subsection{Proof of Corollary \ref{th:permutecond}}
\begin{proof}
Assuming $V_l=V$ for $l=1,\ldots,d$ in the proof of Theorem \ref{th:condjoint}, it follows that
\begin{eqnarray}\nonumber
& &\mathbb{P}\left(U_1 \leq u_1,\ldots,U_d\leq u_d\left\vert K=k \right.\right) = \mathbb{P}\left(V_{1,k} \leq v_{1,k},\ldots,V_{d,k} \leq v_{d,k}\left\vert K=k\right.\right)
\\\nonumber &=&\mathbb{P}\left(  \left.\left(\bigcap_{l\in \mathcal{L}^{-}_{k}} \left\{\left(1-V\right)\leq v_{l,k} \right\}\right)\cap \left(\bigcap_{l\in \mathcal{L}^{+}_{k}} \left\{V\leq v_{l,k}\right\}\right)\right\vert K=k\right)
\\\label{minmax}&=& \mathbb{P}\left(1-v^{-}_{ k}<V\leq v^{+}_{ k}\left\vert K=k\right.\right)
\\ \nonumber &=& \mathbb{I}_{\left[0,v^{+}_{ k}\right]}\left(1-v^{-}_{ k}\right)\int^{v^{+}_{ k}}_{1-v^{-}_{ k}} dv
\\ \nonumber\label{unifcop}&=& \max\left( v^{+}_{k}  + v^{-}_{k} -1 , 0 \right),
\end{eqnarray}
where  in \eqref{minmax} we used the identity 
$\displaystyle \max_{l\in\mathcal{L}^{-}_{k}} \left(1- v_{l,k}\right)=1- \displaystyle \min_{l\in\mathcal{L}^{-}_{k}}  v_{l,k}$

The expression for the bivariate conditional distributions is obtained by substituting $u_{l^{\prime \prime}}=1$ in $v_{l^{\prime \prime},k}=F_{U_{l^{\prime \prime}}\left\vert K\right.}\left(u_{l^{\prime \prime}};k\right)$ for all $l^{\prime \prime}\neq l^{ \prime}$ and $l^{\prime \prime}\neq l$ where $l^{\prime \prime}, l^{ \prime}, l\in\mathcal{D}$, thus obtaining:

\begin{eqnarray*}
& v^{+}_{k}=\min\left(v_{l,k},v_{l^{\prime},k}\right), v^{-}_{k}=1  &\hbox{ if } l^{ \prime}, l\in\mathcal{L}^{+}_{k}\\
&v^{-}_{k}=\min\left(v_{l,k},v_{l^{\prime},k}\right), v^{+}_{k}=1  &\hbox{ if }l^{ \prime}, l\in\mathcal{L}^{-}_{k}\\
& v^{+}_{k}=v_{l,k}, v^{-}_{k}=v_{l^{\prime},k} & \hbox{ if } l^{ \prime}\in\mathcal{L}^{-}_{k}, l\in\mathcal{L}^{+}_{k}\\
& v^{-}_{k}=v_{l,k}, v^{+}_{k}=v_{l^{\prime},k} & \hbox{ if } l^{ \prime}\in\mathcal{L}^{+}_{k}, l\in\mathcal{L}^{-}_{k}.
\end{eqnarray*}

\end{proof}
\subsection{Proof of Proposition \ref{prop:taubound}}
\begin{proof}
By  corollary \ref{distr} we obtain:
\begin{eqnarray*}
\mathbb{P}\left(\mathbf{U}\leq\mathbf{W}\right)= \dfrac{1}{\left\vert\mathcal{E}\right\vert}\sum^{\left\vert\mathcal{E}\right\vert}_{k_{\mathbf{U}}=1} \mathbb{P}\left(R_{\mathbf{V},\mathcal{L}^{-}_{k_{\mathbf{U}}}}\left(\mathbf{V}\right)\leq\mathbf{Y}_{k_{\mathbf{U}}}\right)
\end{eqnarray*}

with $\mathbf{Y}_{k_{\mathbf{U}}}=\left(F_{U_1\left\vert K_{\mathbf{U}}\right.}\left(W_1;k_{\mathbf{U}}\right),\ldots, F_{U_d\left\vert K_{\mathbf{U}}\right.}\left(W_d;k_{\mathbf{U}}\right)\right)$.

\end{proof}
\subsection{Proof of Proposition \ref{prop:rhobound}}
\begin{proof}
Using the definition of the reflection $R_{\mathcal{L},\frac{1}{2}}\left(\mathbf{V}\right)$, we obtain:
\begin{eqnarray*}
\rho\left(F_{\mathbf{V},\mathcal{D}\setminus\mathcal{L}}\right)&=& \dfrac{2^d\left(d+1\right)}{2^d-\left(d+1\right)}\left(\mathbb{E}\left[\prod_{l\in \mathcal{L}}V_l\prod_{l\in \mathcal{D}\setminus\mathcal{L}}\left(1-V_l\right) \right] -\dfrac{1}{2^d}\right).
\end{eqnarray*}

For our general representation in \eqref{gensimplevar} the expected value of the product is: 
\begin{eqnarray*}
\end{eqnarray*}
\begin{eqnarray*}
\mathbb{E}\left[\prod^{d}_{l=1}U_l\right]&=& \dfrac{1}{\left\vert \mathcal{E}\right\vert}\sum^{\left\vert \mathcal{E}\right\vert}_{k=1}\mathbb{E}\left[\left.\prod^{d}_{l=1}U_l\right\vert K=k\right]=\dfrac{1}{\left\vert \mathcal{E}\right\vert}\sum^{\left\vert \mathcal{E}\right\vert}_{k=1}\mathbb{E}\left[\left.\prod^{d}_{l=1}U_l\right\vert K=k\right]\\
&=& \dfrac{1}{\left\vert \mathcal{E}\right\vert}\sum^{\left\vert \mathcal{E}\right\vert}_{k=1}\mathbb{E}\left[\prod^{d}_{l
=1}\left(x_{l,i\left(k\right)} V_l +x_{l,j\left(k\right)}\left(1-V_l\right)\right) \right]
\\
&=& \sum^{d}_{m=0}\sum_{\tiny\begin{array}{c}\mathcal{L}_m\subseteq\mathcal{D}\\ \left\vert \mathcal{L}_m\right\vert=m\end{array}} \xi_{\mathcal{L}_m}\mathbb{E}\left[\prod_{l\in \mathcal{L}_m}V_l\prod_{l\in \mathcal{D}\setminus\mathcal{L}_m}\left(1-V_l\right) \right],
\end{eqnarray*}

with 

\begin{equation*}
    \xi_{\mathcal{L}_m}=\dfrac{1}{\left\vert \mathcal{E}\right\vert}\sum^{\left\vert \mathcal{E}\right\vert}_{k=1}\left(\prod_{l\in \mathcal{L}_m}x_{l,i\left(k\right)}  \prod_{l\in \mathcal{D}\setminus\mathcal{L}_m}x_{l,j\left(k\right)}\right).
\end{equation*}

The Spearman's $\rho\left(F_{\mathbf{U}}\right)$ can be expressed as a function of the $\rho$s on the different  $F_{\mathbf{V},\mathcal{D}\setminus\mathcal{L}_m}$ :

\begin{eqnarray}\label{rhosegment}
 \rho\left(F_{\mathbf{U}}\right)
 &=& \sum^{d}_{m=0}\sum_{\tiny\begin{array}{c}\mathcal{L}_m\subseteq\mathcal{D}\\ \left\vert \mathcal{L}_m\right\vert=m\end{array}}\xi_{\mathcal{L}_m}\rho\left(F_{\mathbf{V},\mathcal{D}\setminus\mathcal{L}_m}\right)
+\dfrac{2^d\left(d+1\right)}{2^d-\left(d+1\right)}\dfrac{1}{2^d}\left(\xi^{*} -1 \right)\label{rhobound}
\end{eqnarray}
with
\begin{eqnarray*}
\xi^{*}&=&\sum^{d}_{m=0}\sum_{\tiny\begin{array}{c}\mathcal{L}_m\subseteq\mathcal{D}\\ \left\vert \mathcal{L}_m\right\vert=m\end{array}} \xi_{\mathcal{L}_m} =\dfrac{1}{\left\vert \mathcal{E}\right\vert}\sum^{\left\vert \mathcal{E}\right\vert}_{k=1} \prod^{d}_{l=1} \left(x_{l,i\left(k\right)} + x_{l,j\left(k\right)} \right).\label{eq:RhoSeg}
\end{eqnarray*}

If $\mathbf{V}$ is reflection invariant then $\bar{F}_{\mathbf{V},\mathcal{D}\setminus\mathcal{L}_{m}}=\bar{F}_{\mathbf{V}}$ in eq. \eqref{eq:RhoSeg}. The condition $0<\xi^*\leq 1$ implies
%
$\underline{\rho}\leq\rho\left(F_{\mathbf{U}}\right)<\bar{\rho}$ where  $\underline{\rho}=\min\{\rho\left(F_{\mathbf{V}}\right),-(d+1)/(2^d-\left(d+1\right))\}$ and $\bar{\rho}=\max\{\rho\left(F_{\mathbf{V}}\right),-(d+1)/(2^d-\left(d+1\right))\}$. Since  
$\rho_{\min}\leq\rho\left(F_{\mathbf{V}}\right)$ if we are able to show that $-(d+1)/(2^d-(d+1))\leq\rho_{\min}$ we obtain the result. Let $\mathbf{V}^{\ast}$ such that $\rho(F_{\mathbf{V}^{\ast}})=\rho_{\min}$
\begin{eqnarray}
&&-\dfrac{\left(d+1\right)}{2^d-\left(d+1\right)}\leq  \dfrac{2^d\left(d+1\right)}{2^d-\left(d+1\right)}\left(\mathbb{E}\left[\prod^{d}_{l=1}V^{\ast}_l\right] -\dfrac{1}{2^d}\right)\\
&&\Longleftrightarrow
2^{d}\mathbb{E}\left[\prod^{d}_{l=1}V^{\ast}_l\right]\geq 0
\end{eqnarray}
which is satisfied since $V_l^{\ast}$ is a.s. positive for all $l=1,\ldots,d$.
\end{proof}

\subsection{Proof of Proposition \ref{RBScirc}}
\begin{proof}
  $U_1\sim\mathcal{U}[0,1]$ and 
\begin{eqnarray*}
U_l &=& \gamma_l - \dfrac{U_1}{d-1},\quad
\gamma_l = \dfrac{l -1 }{d-1}
\end{eqnarray*}
or 
\begin{eqnarray*}
U_l &=& \dfrac{l}{d-1}U_1 +\dfrac{l-1}{d-1}\left(1-U_1\right)\sim \mathcal{U}\left[\dfrac{l-1}{d-1},\dfrac{l}{d-1}\right].
\end{eqnarray*}
It could be verified that  $U_l$ has the same distribution of extracting on the edge $\left(d,1\right) \in C_{d}\left(\left\{1\right\}\right)$ . Random permutation permute simultaneously the coordinate of the vertexes that remain one the first  circular permutation of the other , The exchangeable version then has uniform support on the edges that start from a permutation of the first vertex and link the vertex resulting from an addition first circular permutation. This is exactly the support of the  the exchangeable version of  CCV with dependence graph $C_{d}\left(\left\{1\right\}\right)$  
\end{proof}

\subsection{Proof of Proposition \ref{rotsampseg}}
\begin{proof}
Let $m-1=\left\lfloor d U\right\rfloor\in \left[0,\ldots, d-1\right]$ and $V= dU- \left\lfloor d U\right\rfloor = \left(dU\right)\mathrm{mod}\,1\sim \mathcal{U}[0,1]$ .
We can write \eqref{rotsamp} as:
\begin{eqnarray*}
U_{l}&=& \left(\dfrac{l-1+m-1 +V}{d} \right)\mathrm{mod}\,1,
\\&=&  \left(\left(\dfrac{l+m-2}{d}\right)\mathrm{mod}\,1 +\dfrac{V}{d} \right)\mathrm{mod}\,1
\\&=& \left\{\begin{array}{ccc} \dfrac{l+m-2}{d} + \dfrac{V}{d}&\hbox{if}& m < d+2-l 
\\\dfrac{l+m-2-d}{d} + \dfrac{V}{d}  & \hbox{if}& m \geq d+2-l  \end{array}\right.
\\&=& \left\{\begin{array}{ccc} \dfrac{l+m-1}{d}V+ \dfrac{l+m-2}{d}\left(1-V\right) &\hbox{if}& m < d+2-l 
\\\dfrac{l+m-1-d}{d}V +\dfrac{l+m-2-d}{d}\left(1-V\right)  & \hbox{if}& m \geq d+2-l  \end{array}\right.
\end{eqnarray*}
\end{proof}

\subsection{Proof of Corollary \ref{CTMAJ}}
\begin{proof}

The unique values for coordinates over the $l$-th dimension are:

\begin{eqnarray*}
n_l&=& d+1\\
\mathbf{a}_l&=& \dfrac{1}{d}\left(0,\ldots,d\right)^{T}
\end{eqnarray*}

In the projected graph, only nearby vertexes are connected by edges, the difference in projected coordinates joined is constant and the number of edges is $d$. We have:
\begin{eqnarray*}
n^{l}_{\left(m,m^{\prime}\right)}&=&\left\{\begin{array}{ccc}1 && \left\vert m-m^{\prime}\right\vert =1\\
&&\\
0 && \hbox{otherwise}\end{array}
\right.
\\ \mathcal{K}_{l,m}&=& \left(m-1,m\right)\\
a_{l m}- a_{l,m-1}&=& \dfrac{1}{d}
\end{eqnarray*}
with $m,m'=1,\ldots,n_l$. Given those preliminary computation we are ready to show standard marginal uniformity using the conditions in Theorem \ref{unif}:
\begin{eqnarray*}
\dfrac{1}{\left\vert\mathcal{E}\right\vert}\displaystyle \sum_{k\in \mathcal{K}_{l,m}}  \dfrac{1}{ a_{l,m_{\beta}\left(k\right)}-  a_{l,m_{\alpha}\left(k\right)}}= \dfrac{1}{\left\vert\mathcal{E}\right\vert} \dfrac{n^{l}_{\left(m,m^{\prime}\right)}}{a_{l m}- a_{l,m-1}}= \frac{d}{d}=1
\end{eqnarray*}

Concerning the constant sum constraint, it is sufficent to sum over $l$ , $x_{l,m}$.
\begin{eqnarray*}
\sum^{d}_{l=1} x_{l,m}&=&  \sum^{d+2-m-1}_{l=1} \dfrac{l+m-1}{d} + \sum^{d}_{l=d+2-m} \dfrac{l+m-1-d}{d}
\\&=& \dfrac{d+1}{2}.
\end{eqnarray*}

\end{proof}

\subsection{Proof of Proposition \ref{propAJ}}
\begin{proof}

In the $3$-dimensional case the proposal of  \cite{Arvi:John:82:VRT} is
\begin{eqnarray*}
U_1 &=& U  \qquad U \sim \mathcal{U}[0,1]
\\
U_2&=& \left\lbrace U + 1/2\right\rbrace\\
U_3 &=& 1 - \left\lbrace 2U \right\rbrace
\end{eqnarray*}
Case $U\in\left[0,1/2\right]$
\begin{eqnarray*}
U_1 &=& U  
\\
U_2&=&  U + 1/2 \\
U_3 &=& 1 - 2U 
\end{eqnarray*}
Case $U\in\left[1/2,1\right]$
\begin{eqnarray*}
U_1 &=& U 
\\
U_2&=&  U - 1/2\\
U_3 &=& 1 - \left( 2U -1\right)
\end{eqnarray*}
The construction in Eq (2.4) of \cite{Ruschendorf2002} is obtained by setting $W_l=-1+2 U_l$. 
\end{proof}

\subsection{Proof of Proposition \ref{basedAVseg}}
\begin{proof}
Let $V= \left\lbrace b^{d-2} U_1\right\rbrace$ and $m=  \left\lfloor b^{d-2} U_1\right\rfloor$. Then $V$ is a standard uniform random variable and $m$ takes values on the integers $\left\lbrace 0,\ldots,d-1\right\rbrace$ with equal probabilities. We have:
\begin{eqnarray*}
U_1&=& \dfrac{V}{b^{d-2}} + \dfrac{m}{b^{d-2}} = V \dfrac{m +1}{b^{d-2}} + \left(1-V\right) \dfrac{m}{b^{d-2}} \\&=&  V z_{1,m+1} + \left(1-V\right)  y_{1,m+1}.
\end{eqnarray*}
For $i=2,\ldots,d-1$:
\begin{eqnarray*}
U_i&=&\left\lbrace b^{i-2} \left(V \dfrac{m +1}{b^{d-2}} + \left(1-V\right) \dfrac{m}{b^{d-2}}\right) + 1/2 +1/2 V - 1/2 V \right\rbrace\\  
&=&\left\lbrace V \left( b^{i-2}\dfrac{m +1}{b^{d-2}} +1/2\right) + \left(1-V\right) \left(b^{i-2}\dfrac{m}{b^{d-2}} + 1/2 \right)\right\rbrace\\ 
&=& \left\lbrace V \left( \left\lfloor b^{i-2}\dfrac{m}{b^{d-2}} +1/2 \right\rfloor  +\left\lbrace b^{i-2}\dfrac{m}{b^{d-2}} +1/2 \right\rbrace + \dfrac{b^{i-2}}{b^{d-2}}\right)\right.
\\ &+& \left. \left(1-V\right) \left(\left\lfloor b^{i-2}\dfrac{m}{b^{d-2}} +1/2 \right\rfloor  +\left\lbrace b^{i-2}\dfrac{m}{b^{d-2}} +1/2 \right\rbrace\right)  \right\rbrace 
\\
&=& \left\lbrace V \left( \left\lbrace b^{i-2}\dfrac{m}{b^{d-2}} +1/2 \right\rbrace + \dfrac{b^{i-2}}{b^{d-2}}\right)+ \left(1-V\right) \left(  \left\lbrace b^{i-2}\dfrac{m}{b^{d-2}} +1/2 \right\rbrace\right)  \right\rbrace 
\\
&=&  V z_{i,m+1}+ \left(1-V\right) y_{i,m+1}, \quad i=1,\ldots,d-1  
\end{eqnarray*}
\begin{eqnarray*}
U_d&=& 1- \left\lbrace  \left(V + m\right) \right\rbrace = 1-V= V z_{d,m+1} + \left(1-V\right) y_{d,m+1}
\end{eqnarray*}
 
The final result for $U_i$ follows from $z_{i,m+1}\leq 1 $. Then the convex combination of $z_{i,m+1}$ and $y_{i,m+1}$ is in $\left[0,1\right]$.
\end{proof}
\subsection{Proof of Corollary \ref{CTMAJ}}
\begin{proof}
We now show that the base-$b$ Ardvisen and Johnson random vector has standard uniform marginals for each $b\in\mathbb{N}$ using conditions in Theorem \ref{unif}. This would also showcase the use of our notation in a non trivial example. Let us derive another expression for the $\mathbf{y}$'s that would simplify the coming computations
\begin{eqnarray*}
y_{1,m}&=& \dfrac{m-1}{b^{d-2}}\\
y_{k,m}&=& \left( b^{k-2} \dfrac{m-1}{b^{d-2}}+ \dfrac{1}{b}\right) \mathrm{mod}\,1 =\left( \dfrac{m-1 + b^{d-k-1}}{b^{d-k}}\right) \mathrm{mod}\,1\\&= &\dfrac{\left(m-1 + b^{d-k-1}\right)\mathrm{mod}\, b^{d-k}}{b^{d-k}},\quad k\in\left\{2,\ldots,d-1\right\}\\
y_{d,m}&=&1
\end{eqnarray*}
for $m\in\{1,\ldots,b^{d-2}\}$, where the last line follows from the fact that $m-1 + b^{d-k-1}$ is an integer, the inclusion of the case $k=d$ from he fact that $a \left(\hbox{mod }1\right)=0$ for $a$ integer. 
In particular this representation allow an easy computation of key quantities. First of all , coordinates of projected vertexes i.e. the vector of ordered unique values in $\mathbf{y}_l$ and  $\mathbf{z}_l$, $l\in\left\{1,\ldots,d\right\} $ :
\begin{eqnarray*}
n_l&=&b^{d-\max\left(l,2\right)} +1\\
\mathbf{a}_l&=& \dfrac{1}{n_l-1}\left(0,\ldots,n_l-1 \right)^{T}
\end{eqnarray*}

Then, because of the mod function, values of $y_{k,m}$ repeat themselves after $b^{d-\max\left(l,2\right)}$ positions. The multiplicity of each unique value in $\mathbf{y}_l$ is the the length of the vector divided by the period, i.e.: $\frac{b^{d-2}}{b^{d-\max\left(l,2\right)}}=b^{\max\left(l,2\right)-2}$. Considering the fact that all values with the exception of $0$ are also repeated in the same way in $\mathbf{z}_l$, we get:
\begin{eqnarray*}
\left\vert\mathcal{M}_{l,1}\right\vert&=&\left\vert\mathcal{M}_{l,n_l}\right\vert= b^{\max\left(l,2\right)-2}\\
\left\vert\mathcal{M}_{l,j}\right\vert &=&2b^{\max\left(l,2\right)-2}, j=2,\ldots,n_l-1
\end{eqnarray*}

Finally, in the projected graph only nearby vertexes are connected by edges, the difference in projected coordinates joined is constant and the number of edges correspond to the number of replicas in $\mathbf{y}_l$ of such a vertex. We have:
\begin{eqnarray*}
n^{l}_{\left(m,m^{\prime}\right)}&=&\left\{\begin{array}{ccc} b^{\max\left(l,2\right)-2} && \left\vert m-m^{\prime}\right\vert =1\\
&&\\
0 && \hbox{otherwise}\end{array}
\right.
\\ \mathcal{K}_{l,m}&=& \left(m-1,m\right)\\
a_{l m}- a_{l,m-1}&=& \dfrac{1}{b^{d-\max\left(l,2\right)}}
\end{eqnarray*}
with $m,m'=1,\ldots,n_l$. Given those preliminary computation we are ready to show standard marginal uniformity using the conditions in Theorem \ref{unif}:
\begin{eqnarray*}
\dfrac{1}{\left\vert\mathcal{E}\right\vert}\displaystyle \sum_{k\in \mathcal{K}_{l,m}}  \dfrac{1}{ a_{l,m_{\beta}\left(k\right)}-  a_{l,m_{\alpha}\left(k\right)}}= \dfrac{1}{\left\vert\mathcal{E}\right\vert} \dfrac{n^{l}_{\left(m,m^{\prime}\right)}}{a_{l m}- a_{l,m-1}}= \frac{b^{\max\left(l,2\right)-2}b^{d-\max\left(l,2\right)}}{b^{d-2} }=1
\end{eqnarray*}

Concerning the constant sum constraint, it is known since \cite{Arvi:John:82:VRT} (see also \cite{CraiuMeng2005}) that the case $b=2$ satisfies $d$-CTM. We show that the first vertex ($m=1$) sum to $\frac{d}{2}$ only in the case $b=2$, implying that $b\neq2$ cases are not $d$-CTM.
\begin{eqnarray*}
y_{1,1}&=& \dfrac{0}{b^{d-2}}=0\\
y_{k,1}&=& \left\lbrace b^{k-2} \dfrac{0}{b^{d-2}}+ \dfrac{1}{b}\right\rbrace = \dfrac{1}{b}\\
k&\in&\left\{2,\ldots,d-1\right\}\\
y_{d,1}&=&1
\end{eqnarray*}
then we should solve for $b$ the following equation 
\begin{eqnarray*}
\sum^{d}_{k=1} y_{k,1} = \dfrac{d-2}{b} +1 = \dfrac{d}{2}
\end{eqnarray*}
which has as unique solution $b=2$.
\end{proof}

\subsection{Proof of Proposition \ref{ILHsperstar}}
\begin{proof}
 For each $l\in\mathcal{D}$
 \begin{eqnarray*}
U_{l,t}&=& \dfrac{1}{d}\left(\pi_k\left(l\right) + U_{l,t-1}\right) 
\\&=&  \dfrac{1}{d}\left(\pi_k\left(l\right) + U_{l,t-1}\right)  + \dfrac{\pi_k\left(l\right)}{d}U_{l,t-1} - \dfrac{\pi_k\left(l\right)}{d}U_{l,t-1}
\\&=&  \dfrac{\pi_k\left(l\right)+1 }{d} U_{l,t-1} +  \dfrac{\pi_k\left(l\right)}{d}\left(1-U_{l,t-1}\right)
\end{eqnarray*}
\end{proof}

\subsection{Proof of Proposition \ref{ILHseg}}
\begin{proof}
For each t $-1\leq X_{it}\leq 1$ let us define $\mathbf{U}_t=\dfrac{\mathbf{X}_t+1}{2}$, we have:
\begin{eqnarray}
 \mathbf{U}_{t}= \dfrac{1}{3} \mathbf{U}_{t-1} + \dfrac{2}{3}\dfrac{\mathbf{V}_k +1}{2}
\end{eqnarray} 
and $\mathbf{V}_k +1$ is a permutation of $\left\{0,1,2\right\}$, so that the generalization to $d$ dimension is straightforward:
\begin{eqnarray}
 \mathbf{U}_{t}= \dfrac{1}{d} \mathbf{U}_{t-1} + \left(1-\dfrac{1}{d}\right) \dfrac{\mathcal{D}^{\pi}_t}{d-1} = \dfrac{1}{d}\left(\mathcal{D}^{\pi}_t + \mathbf{U}_{t-1}\right)
\end{eqnarray} 
\end{proof}

\subsection{Proof of Corollary \ref{unifILH}}

\begin{proof}
With the line segment representation, of proposition \ref{ILHseg}
the unique values are  $a_{l,m}= \dfrac{m}{d}$, $m=0,\ldots,d$ and the segments extremes are  $\beta_{l,k}= \frac{\pi_k\left(l\right) + 1}{d}$ and $\alpha_{l,k}= \frac{\pi_k\left(l\right)}{d}$. Then, since $m_{\beta}\left(k\right)= \pi_k\left(l\right) + 1$ and $m_{\alpha}\left(k\right)= \pi_k\left(l\right)$, we obtain:
\begin{eqnarray*}
a_{l,m_{\beta}\left(k\right)}-  a_{l,m_{\alpha}\left(k\right)}&=& \dfrac{1}{d}
\end{eqnarray*}
and 
\begin{eqnarray*}
\mathcal{K}_{l,m}&=&\left\{k \in \left\{1,\ldots, \left\vert\mathcal{E}\right\vert\right\}: \left(m_{\alpha}\left(k\right)\leq m-1 \right) \cap \left(m_{\beta}\left(k\right)\geq m  \right)\right\}
\\&=&\left\{k \in \left\{1,\ldots, \left\vert\mathcal{E}\right\vert\right\}: \pi_k\left(l\right) =m-1
\right\}
\end{eqnarray*}

leading to  $\left\vert \mathcal{K}_{l,m}\right\vert$ being equal to the number of permutations with the l-th element equal to $m-1$. We obtain then   $\left\vert \mathcal{K}_{l,m}\right\vert= \left(d-1\right)!$, which implies for each $m =2,\ldots,d$:

\begin{eqnarray*}
F_{l,m}\left( \mathbf{a}_l\right)&=&\dfrac{1}{\left\vert\mathcal{E}\right\vert }\displaystyle \sum_{k\in \mathcal{K}_{l,m}}  \dfrac{1}{ a_{l,m_{\beta}\left(k\right)}-  a_{l,m_{\alpha}\left(k\right)}}
\\&=&\dfrac{1}{d! }\displaystyle \sum_{k\in \mathcal{K}_{l,m}} d = \dfrac{d}{d! }\left\vert \mathcal{K}_{l,m}\right\vert=1
\end{eqnarray*} 

we now investigate if the constant sum is preserved over iterations. As previously shown, the line segment representation of ILH satisfies the first constraint in equations \eqref{c1c2} with $c_1=\dfrac{1}{d}$. The second constraint is:
\begin{equation*}
\sum^{d}_{l=1} \dfrac{\pi_{k}\left(l\right)}{d}= \sum^{d-1}_{l=0} \dfrac{l}{d} = \dfrac{1}{d}\left(d-1\right)\dfrac{d}{2}= \left(1-\dfrac{1}{d}\right)\dfrac{d}{2}= \left(1-c_1\right)\dfrac{d}{2}
\end{equation*}
and if  $U_{l,t-1}$ has constant sum then $U_{l,t}$ in \eqref{ILHiter} has constant sum by Lemma \ref{CTMcomp}.
\end{proof}

\subsection{Proof of Proposition \ref{ILHtaurho}}

\begin{proof}
In the ILH($T$) case we can apply the deterministic composition. By ILH construction the $\mathbf{V}$ vector has independent components, and in the segments belong to non-intersecting hypercubes of smaller size. Let's consider a generic construction with the latter characteristics.
\begin{eqnarray*}
&&\mathbb{P}\left(\mathbf{U}\leq\mathbf{W}\right)=\\
&= &\dfrac{1}{\left\vert\mathcal{E}\right\vert}\sum^{\left\vert\mathcal{E}\right\vert}_{k_{\mathbf{U}}=1} \mathbb{P}\left(R_{\mathbf{V},\mathcal{L}^{-}_{k_{\mathbf{U}}}}\left(\mathbf{V}\right)\leq\mathbf{Y}_{k_{\mathbf{U}}}\right)\\
&=& \dfrac{1}{\left\vert\mathcal{E}\right\vert}\sum^{\left\vert\mathcal{E}\right\vert}_{k_{\mathbf{U}}=1} \mathbb{P}\left(\mathbf{V}\leq\mathbf{Y}_{k_{\mathbf{U}}}\right)\\
&=& \dfrac{1}{\left\vert\mathcal{E}\right\vert}\sum^{\left\vert\mathcal{E}\right\vert}_{k_{\mathbf{U}}=1} \mathbb{E}_{\mathbf{W}}\left[\prod^d_{l=1}\dfrac{\min\left(\max\left(W_l,\alpha_{l,k_{\mathbf{U}}}\right),\beta_{l,k_{\mathbf{U}}}\right)-\alpha_{l,k_{\mathbf{U}}}}{\beta_{l,k_{\mathbf{U}}}-\alpha_{l,k_{\mathbf{U}}}}\right]
\\&=&
\dfrac{1}{\left\vert\mathcal{E}\right\vert}\sum^{\left\vert\mathcal{E}\right\vert}_{k_{\mathbf{U}}=1} \dfrac{1}{\left\vert\mathcal{E}\right\vert}\sum^{\left\vert\mathcal{E}\right\vert}_{k_{\mathbf{W}}=1}\mathbb{E}_{\mathbf{V}}\left[\prod^d_{l=1}\dfrac{\min\left(\max\left(\alpha_{l,k_{\mathbf{W}}}+ \left(\beta_{l,k_{\mathbf{W}}}-\alpha_{l,k_{\mathbf{W}}}\right)V_l,\alpha_{l,k_{\mathbf{U}}}\right),\beta_{l,k_{\mathbf{U}}}\right)-\alpha_{l,k_{\mathbf{U}}}}{\beta_{l,k_{\mathbf{U}}}-\alpha_{l,k_{\mathbf{U}}}}\right]
\\&=&
\dfrac{1}{\left\vert\mathcal{E}\right\vert}\sum^{\left\vert\mathcal{E}\right\vert}_{k_{\mathbf{U}}=1} \dfrac{1}{\left\vert\mathcal{E}\right\vert}\sum^{\left\vert\mathcal{E}\right\vert}_{k_{\mathbf{W}}=1}\prod^d_{l=1}\dfrac{\mathbb{E}_{V_l}\left[\min\left(\max\left(\alpha_{l,k_{\mathbf{W}}}+ \left(\beta_{l,k_{\mathbf{W}}}-\alpha_{l,k_{\mathbf{W}}}\right)V_l,\alpha_{l,k_{\mathbf{U}}}\right),\beta_{l,k_{\mathbf{U}}}\right)\right]-\alpha_{l,k_{\mathbf{U}}}}{\beta_{l,k_{\mathbf{U}}}-\alpha_{l,k_{\mathbf{U}}}}
\end{eqnarray*}
Since the boxes with different segments do not overlap 
in the sum for $k_{\mathbf{W}}$ the only contributions come from $k_{\mathbf{U}}=k_{\mathbf{W}}$:

\begin{eqnarray*}
\mathbb{P}\left(\mathbf{U}\leq\mathbf{W}\right)&= &
\dfrac{1}{\left\vert\mathcal{E}\right\vert^2}\sum^{\left\vert\mathcal{E}\right\vert}_{k_{\mathbf{U}}=1} \prod^d_{l=1}\dfrac{\mathbb{E}_{V_l}\left[\alpha_{l,k_{\mathbf{U}}}+ \left(\beta_{l,k_{\mathbf{U}}}-\alpha_{l,k_{\mathbf{U}}}\right)V_l\right]-\alpha_{l,k_{\mathbf{U}}}}{\beta_{l,k_{\mathbf{U}}}-\alpha_{l,k_{\mathbf{U}}}}\\
&=& \dfrac{1}{\left\vert\mathcal{E}\right\vert^2}\sum^{\left\vert\mathcal{E}\right\vert}_{k_{\mathbf{U}}=1} \prod^d_{l=1}\dfrac{\left(\beta_{l,k_{\mathbf{U}}}-\alpha_{l,k_{\mathbf{U}}}\right)/2}{\beta_{l,k_{\mathbf{U}}}-\alpha_{l,k_{\mathbf{U}}}}
\\&=& \dfrac{1}{\left\vert\mathcal{E}\right\vert} \dfrac{1}{2^d}
\end{eqnarray*}

The number of vertexes in the ILH($T$) case is $\left(d!\right)^T$ and the result for $\tau$ follows.

Denoting $\xi^{*}_t$, the case for $t\leq T$ iterations we compute:
\begin{eqnarray*}
\xi^{*}_1 &=&  \dfrac{1}{d!}\sum^{d!}_{k_1=1}\prod^d_{l=1}\dfrac{\left(2\pi_{k_1}\left(l\right)  +1\right)}{d}\\
\xi^{*}_2 &=&  \dfrac{1}{d!}\sum^{d!}_{k_2=1}\dfrac{1}{d!}\sum^{d!}_{k_1=1} \prod^d_{l=1}\left(\dfrac{ \pi_{k_1}\left(l\right)}{d^2} +\dfrac{ \pi_{k_2}\left(l\right)}{d} +\dfrac{ \pi_{k_1}\left(l\right)  +1}{d^2} +\dfrac{ \pi_{k_2}\left(l\right) }{d}\right) 
\\&=& \dfrac{1}{d!}\sum^{d!}_{k_2=1}\dfrac{1}{d!}\sum^{d!}_{k_1=1} \prod^d_{l=1}\left(\sum^2_{t=1}\dfrac{ 2\pi_{k_t}\left(l\right)}{d^{2-t+1}} + \dfrac{1}{d^2}\right)
\\&\vdots&
\\\xi^{*}_T &=&\dfrac{1}{d!}\sum^{d!}_{k_T=1}\cdots\dfrac{1}{d!}\sum^{d!}_{k_1=1} \prod^d_{l=1}\left(\sum^T_{t=1}\dfrac{ 2\pi_{k_t}\left(l\right)}{d^{T-t+1}} + \dfrac{1}{d^T}\right)
\end{eqnarray*}
We simplify the expression for $T=2$ by expanding the product
\begin{eqnarray*}
\xi^{*}_2 &=&  \dfrac{1}{d!}\sum^{d!}_{k_2=1}\dfrac{1}{d!}\sum^{d!}_{k_1=1} 
\sum^{d}_{m_2=0} \sum_{{\tiny\begin{array}{c}\mathcal{L}_{m_2}\subseteq\mathcal{D}\\ \left\vert\mathcal{L}_{m_2}\right\vert=m_2 \end{array}}}\left\{\prod_{l\in\mathcal{L}_{m_2} }\dfrac{ 2 \pi_{k_2}\left(l\right)}{d}   \right\}\left\{  \prod_{l\in\mathcal{D}\setminus\mathcal{L}_{m_2}}\dfrac{ 2\pi_{k_1}\left(l\right)  +1}{d^2} \right\}
\\&=& \sum^{d}_{m_2=0} \sum_{{\tiny\begin{array}{c}\mathcal{L}_{m_2}\subseteq\mathcal{D}\\ \left\vert\mathcal{L}_{m_2}\right\vert=m_2 \end{array}}}\dfrac{1}{d!} \sum^{d!}_{k_2=1}\left\{\prod_{l\in\mathcal{L}_{m_2} }\dfrac{ 2 \pi_{k_2}\left(l\right)}{d}   \right\} \dfrac{1}{d!}\sum^{d!}_{k_1=1} \left\{ \prod_{l\in\mathcal{D}\setminus\mathcal{L}_{m_2}}\dfrac{ 2\pi_{k_1}\left(l\right)  +1}{d^2} \right\}
\end{eqnarray*}

The sums over permutations can be reduced. Consider the first sum, the product is independent from $\pi_{k_2}\left(l\right)$ for $l\in\mathcal{D}\setminus\mathcal{L}_{m_2}$ and invariant for permutations of the product itself. Then the product is the same for $\left(d-m_2\right)! m_2!$ permutations.  We consider only the ordered products multiplied by   $\left(d-m_2\right)! m_2!$. This amounts to summing over the set of combinations of $\mathcal{D}$ with $m_2$ elements that we denote by $\mathcal{C}_{\mathcal{D},m_2}$:
\begin{eqnarray*}
\dfrac{1}{d!} \sum^{d!}_{k_2=1}\prod_{l\in\mathcal{L}_{m_2} }\dfrac{ 2 \pi_{k_2}\left(l\right)}{d} &=& {d\choose m_2}^{-1} \sum_{{\tiny\left\{i_{1},\ldots,i_{m_2}\right\}\in \mathcal{C}_{\mathcal{D},m_2}} } \prod^{m_2}_{l=1} \dfrac{2\left(i_l-1\right)}{d}
\end{eqnarray*}
\begin{eqnarray*}
\xi^{*}_2 &=& \resizebox{\linewidth}{!}{$\displaystyle \sum^{d}_{m_2=0} {d\choose m_2, d-m_2} \left\{{d\choose m_2}^{-1} \sum_{{\tiny\left\{i_{1},\ldots,i_{m_2}\right\}\in {\mathcal{D}\choose m_2}} } \prod^{m_2}_{l=1} \dfrac{2\left(i_l-1\right)}{d}   \right\} \left\{{d\choose d-m_2}^{-1} \sum_{{\tiny\left\{i_{1},\ldots,i_{d-m_2}\right\}\in \mathcal{C}_{\mathcal{D},d-m_2}} } \prod^{d-m_2}_{l=1} \dfrac{2\left(i_l-1\right) +1}{d^2}   \right\}$}
\\&=& \resizebox{\linewidth}{!}{$\displaystyle \sum_{m_2+m_1+m_0=d} {d\choose m_2, m_1, m_0}\prod^{2}_{t=1} \left\{{d\choose m_t}^{-1} \sum_{{\tiny\left\{i_{1},\ldots,i_{m_t}\right\}\in \mathcal{C}_{\mathcal{D},m_t}} } \prod^{m_t}_{l=1} \dfrac{2\left(i_l-1\right)}{d^{2-t+1}}   \right\} \left(\dfrac{1}{d^2}\right)^{m_0}$}
\end{eqnarray*}

The same line of reasoning leads to:
\begin{eqnarray*}
\xi^{*}_T &=& \!\!\!\resizebox{\linewidth}{!}{$\displaystyle \sum_{m_0+m_1+\ldots+m_T=d} {d\choose m_0,m_1,\dots ,m_T} \prod^{T}_{t=1}\left\{{d\choose m_t}^{-1} \sum_{{\tiny\left\{i_{1},\ldots,i_{m_t}\right\}\in \mathcal{C}_{\mathcal{D},m_2}} } \prod^{m_t}_{l=1} \dfrac{2\left(i_l-1\right)}{d^{T-t+1}}   \right\}\left(\dfrac{1}{d^T}\right)^{m_0}$}   
\end{eqnarray*}

\end{proof}

\subsection{Proof of Lemma \ref{muirrelevance}}

\begin{proof}

Let $\mathbf{U}_{l}=\left( U_l^1,\ldots,U_l^p\right), l\in\mathcal{D}$, $\mathbf{U}_{l}$ is a vector of standard uniform independent random variables. We define $\mathbf{m}_l=\left\{m^1_l,\ldots,m^p_l\right\}, m^i_l=1,\ldots,n_l$, $A^{\mathbf{m}_l}= A_{m^1_l}\times \ldots\times A_{m^p_l}$, $A_{m^l_1}=\left[a_{l,m^i_l-1},a_{l,m^i_l}\right]$ , $A^{\mathbf{m}}= A_{\mathbf{m}_1}\times \ldots\times A_{\mathbf{m}_d}$. By Corollary 1 
we have 
 $\cup_{\mathbf{m}}A^{\mathbf{m}} = \left[0,1\right]^{dp}$ and 
$A^{\mathbf{m}}\cap A^{\mathbf{m}^{\prime}}=\emptyset$ if $\mathbf{m}\neq \mathbf{m}^{\prime}$.
We introduce also a notation for the marginal boxes, given a subset $\mathcal{R}_s\subset \mathcal{D}$, $\left\vert\mathcal{R}_{s}\right\vert=s\leq d$ we call $A_{\left\{m^i_l\right\}_{l\in\mathcal{R}_s}}$ the box $ A_{m^i_{t_1}}\times \ldots\times A_{m^i_{t_s}}$ and $A_{\left\{\mathbf{m}_l\right\}_{l\in\mathcal{R}_s}}$ the box $ A_{\mathbf{m}_{t_1}}\times \ldots\times A_{\mathbf{m}_{t_s}}$ .
The multinomial theorem implies:
\begin{eqnarray*}
 &&\mathbb{E}_{a-LH}\left[\left(\dfrac{1}{\sqrt{d}}\sum^{d}_{l=1}f\left(\mathbf{U}_l\right)\right)^r \right]-\mathbb{E}_{b-LH}\left[\left(\dfrac{1}{\sqrt{d}}\sum^{d}_{l=1}f\left(\mathbf{U}_l\right)\right)^r \right]
\\&=&\displaystyle d^{-r/2}\sum_{r_1+\ldots +r_d = r} \dfrac{r!}{ r_1! \ldots r_d!}\left\{ \mathbb{E}_{a-LH}\left[\prod^d_{l=1} 	 f\left(\mathbf{U}_l\right)^{r_l}\right] - \mathbb{E}_{b-LH}\left[\prod^d_{l=1} 	 f\left(\mathbf{U}_l\right)^{r_l}\right]\right\}
\end{eqnarray*}

For each expectation inside the sum, since $\sum^d_{l=1}r_l = r$ does not vary with $d$, in the limit a different subset $r_{t_1},\ldots,r_{t_s}$, $t_1,\ldots,t_s \in\mathcal{R}_s\subset \mathcal{D}$, $s\leq r$ , of exponents will be different from zero. Then for $c=a,b$ we rewrite:
\begin{eqnarray*}
 \lim_{d\rightarrow \infty}\mathbb{E}_{c-LH}\left[\prod^d_{l=1} 	 f\left(\mathbf{U}_l\right)^{r_l}\right]= \lim_{d\rightarrow \infty}\mathbb{E}_{c-LH}\left[\prod^{s}_{l=1} 	 f\left(\mathbf{U}_{t_{l}}\right)^{r_{t_l}}\right]
\end{eqnarray*}
For $c=a,b$ we obtain:
\begin{eqnarray*}
&&\mathbb{E}_{c-LH}\left[\mathbb{E}_{c-LH}\left[\prod^s_{l=1} 	 f\left(\mathbf{U}_{t_l}\right)^{r_{t_l}}\left\vert \left\{\mathbf{U}_l\in A^{\mathbf{m}_{l}}\right\}_{l\in\mathcal{R}_s}\right.\right]\right]\\&=&
 \sum_{\left\{\mathbf{m}_{l}\right\}_{l\in\mathcal{R}_s}} \mu_{c-LH}\left(A_{\left\{\mathbf{m}_l\right\}_{l\in\mathcal{R}_s}}\right)\mathbb{E}_{c-LH}\left[\prod^s_{l=1} 	 f\left(\mathbf{U}_{t_l}\right)^{r_{t_l}}\left\vert \left\{\mathbf{U}_l\in A^{\mathbf{m}_{l}}\right\}_{l\in\mathcal{R}_s}\right.\right]
\end{eqnarray*}

\begin{eqnarray*}
&&\mu_{c-LH}\left(A_{\left\{\mathbf{m}_l\right\}_{l\in\mathcal{R}_s}}\right)=\prod^{p}_{i=1}\mathbb{P}_{c-LH}\left(\left\{U^{i}_l\in \left(a_{l,m^i_l-1},a_{l,m^i_l}\right]\right\}_{l\in\mathcal{R}_s}\right) \\&=& \prod^{p}_{i=1}\dfrac{1}{\left\vert \mathcal{E}\right\vert}\sum_{k_i\in\mathcal{E}} \mathbb{P}_{c-LH}\left(\left.\left\{U^{i}_l\in \left(a_{l,m^i_l-1},a_{l,m^i_l}\right]\right\}_{l\in\mathcal{R}_s}\right\vert K_i=k_i\right)
\\&=& \resizebox{0.9\linewidth}{!}{$\displaystyle\prod^{p}_{i=1} \dfrac{1}{\left\vert \mathcal{E}\right\vert}\sum_{k_i\in\mathcal{E}} \left[\prod^{s}_{l=1} \mathbb{I}_{\mathcal{K}_{m^i_{t_l}}}\left(k_i\right)\right] \mathbb{E}_{c}\left[\prod^{s}_{l=1}\mathbb{I}_{\left(\dfrac{a_{l,m^i_{t_l}-1}-\alpha_{t_l,k_i}}{\beta_{t_l,k_i}-\alpha_{t_l,k_i}},\dfrac{a_{t_l,m^i_l}-\alpha_{t_l,k_i}}{\beta_{t_l,k_i}-\alpha_{t_l,k_i}}\right]}\left(V^{i}_{t_l,k_i}\right) \right]$}
\\&=& \prod^{p}_{i=1} \dfrac{1}{\left\vert \mathcal{E}\right\vert}\sum_{k_i\in\bigcap^{s}_{l=1}\mathcal{K}_{m^i_{t_l}}} \mathbb{E}_c\left[\prod^{s}_{l=1}\mathbb{I}_{\left(\dfrac{\alpha_{t_l,k_i}-\alpha_{t_l,k_i}}{\beta_{t_l,k_i}-\alpha_{t_l,k_i}},\dfrac{\beta_{t_l,k_i}-\alpha_{t_l,k_i}}{\beta_{t_l,k_i}-\alpha_{t_l,k_i}}\right]}\left(V^{i}_{t_l,k_i}\right)\right]
\\&=& \prod^{p}_{i=1} \dfrac{1}{\left\vert \mathcal{E}\right\vert}\sum_{k_i\in\bigcap^{s}_{l=1}\mathcal{K}_{m^i_{t_l}}} \mathbb{E}_c\left[\prod^{s}_{l=1}\mathbb{I}_{\left(0,1\right]}\left(V^{i}_{t_l,k_i}\right)\right]
\\&=& \prod^{p}_{i=1} \dfrac{1}{\left\vert \mathcal{E}\right\vert}\sum_{k_i\in\bigcap^{s}_{l=1}\mathcal{K}_{m^i_{t_l}}} \mu_c\left(\left(0,1\right]^{s} \times \left(0,1\right]^{d-s}\right)
\\&=& \prod^{p}_{i=1} \dfrac{\left\vert \bigcap^{s}_{l=1}\mathcal{K}_{m^i_{t_l}}\right\vert}{\left\vert \mathcal{E}\right\vert}= \left\{\begin{array}{ccc}
\left(\dfrac{\left(d-s\right)!}{d!}\right)^p &\hbox{if}& m^i_{t_1}\neq  m^i_{t_s} \hbox{ for each } i=1,\ldots,p \\
0 && \hbox{ otherwise } 
\end{array}\right.
\end{eqnarray*}

where $\left\vert \bigcap^{d}_{l=1}\mathcal{K}_{m^i_l}\right\vert=\left(d-s\right)!$ if $m^i_{t_1}\neq \ldots\neq m^i_{t_s}$ and $0$ otherwise because, in the LH case, it represents the number of permutation of the $d$-dimensional vector $\left(0,\ldots,d-1\right)$ when a number $s$ of the components are held fixed and there is no permutation with two equal indexes.
 
 We showed that $\mu_{c-LH}\left(A_{\left\{\mathbf{m}_l\right\}_{l\in\mathcal{R}_s}}\right)$ is independent from the measure of $\mathbf{V}^i$, then if 
\begin{eqnarray}\label{diffE}
\resizebox{0.95\linewidth}{!}{$\displaystyle\mathbb{E}_{a-LH}\left[\prod^s_{l=1} 	 f\left(\mathbf{U}_{t_l}\right)^{r_{t_l}}\left\vert \left\{\mathbf{U}_l\in A^{\mathbf{m}_{l}}\right\}_{l\in\mathcal{R}_s}\right.\right] - \mathbb{E}_{b-LH}\left[\prod^s_{l=1} 	 f\left(\mathbf{U}_{t_l}\right)^{r_{t_l}}\left\vert \left\{\mathbf{U}_l\in A^{\mathbf{m}_{l}}\right\}_{l\in\mathcal{R}_s}\right.\right]=o\left(1\right)$}
\end{eqnarray}

we obtain the statement of the Lemma.

We denote as $B_{\left\{\mathbf{m}_l\right\}_{l\in\mathcal{R}_s}}$ the circumscribed ball  of $A_{\left\{\mathbf{m}_l\right\}_{l\in\mathcal{R}_s}}$. The measure of $B_{\left\{\mathbf{m}_l\right\}_{l\in\mathcal{R}_s}}$ is bounded above by the measure of the union of the boxes adjacent to $A_{\left\{\mathbf{m}_l\right\}_{l\in\mathcal{R}_s}}$ and  $A_{\left\{\mathbf{m}_l\right\}_{l\in\mathcal{R}_s}}$ itself. In the worst case when none of the adjacent boxes has two or more $m$s equal,  we obtain:
\begin{eqnarray*}
\mu_{c-LH}\left(A_{\left\{\mathbf{m}_l\right\}_{l\in\mathcal{R}_s}}\right)\leq \mu_{c-LH}\left(B_{\left\{\mathbf{m}_l\right\}_{l\in\mathcal{R}_s}}\right)\leq \left(2s+1\right) \mu_{c-LH}\left(A_{\left\{\mathbf{m}_l\right\}_{l\in\mathcal{R}_s}}\right)
\end{eqnarray*}

The result follows from Lebesgue-Besicovitch Differentiation Theorem \citep{besicovitch_1945}. We refer to section 1.7 of \cite{evans1991} for a contemporary treatment. In particular, we write, for $c=a,b$ and $\left(\mathbf{u}^{\star}_1,\ldots,\mathbf{u}^{\star}_d\right) $, the center of the box $A_{\left\{\mathbf{m}_l\right\}_{l\in\mathcal{R}_s}}$:
\begin{eqnarray*}
&&\lim_{d\rightarrow\infty} \left\vert \mathbb{E}_{c-LH}\left[\prod^s_{l=1} 	 f\left(\mathbf{U}_{t_l}\right)^{r_{t_l}}
\left\vert \left\{\mathbf{U}_l\in A^{\mathbf{m}_{l}}\right\}_{l\in\mathcal{R}_s}\right.\right]
- \prod^s_{l=1} f\left(\mathbf{u}^{\star}_l\right)^{r_{t_l}} \right\vert\\&\leq&
\lim_{d\rightarrow\infty} \mathbb{E}_{c-LH}\left[\left\vert\prod^s_{l=1} 	 f\left(\mathbf{U}_{t_l}\right)^{r_{t_l}}- \prod^s_{l=1} 	 f\left(\mathbf{u}^{\star}_l\right)^{r_{t_l}}\right\vert\left\vert \left\{\mathbf{U}_l\in A^{\mathbf{m}_{l}}\right\}_{l\in\mathcal{R}_s}\right.\right]\\&\leq& \lim_{d\rightarrow\infty}\dfrac{\displaystyle\int_{A_{\left\{m_l\right\}_{l\in\mathcal{R}_s}} }\left\vert\prod^s_{l=1} 	 f\left(\mathbf{U}_{t_l}\right)^{r_{t_l}}-\prod^s_{l=1} 	 f\left(\mathbf{u}^{\star}_l\right)^{r_{t_l}}\right\vert d\mu_{c-LH}}{\mu_{c-LH}\left(A_{\left\{m_l\right\}_{l\in\mathcal{R}_s}}\right)} \\&\leq&  \lim_{d\rightarrow\infty}\dfrac{\left(2s+1\right)}{\mu_{c-LH}\left(B_{\left\{m_l\right\}_{l\in\mathcal{R}_s}}\right) } \int_{A_{\left\{m_l\right\}_{l\in\mathcal{R}_s}} }\left\vert\prod^s_{l=1} 	 f\left(\mathbf{U}_{t_l}\right)^{r_{t_l}}-\prod^s_{l=1} 	 f\left(\mathbf{u}^{\star}_l\right)^{r_{t_l}}\right\vert d\mu_{c-LH}
\\&\leq& \lim_{d\rightarrow\infty} \dfrac{\left(2s+1\right)}{\mu_{c-LH}\left(B_{\left\{m_l\right\}_{l\in\mathcal{R}_s}}\right) } \int_{B_{\left\{m_l\right\}_{l\in\mathcal{R}_s}} }\left\vert\prod^s_{l=1} 	 f\left(\mathbf{U}_{t_l}\right)^{r_{t_l}}-\prod^s_{l=1} 	 f\left(\mathbf{u}^{\star}_l\right)^{r_{t_l}}\right\vert d\mu_{c-LH}\leq 0
\end{eqnarray*}
 The last passage follows from  $\lim_{d\rightarrow\infty} \left(2s+1\right)\mu_{c-LH}\left(A_{\left\{\mathbf{m}_l\right\}_{l\in\mathcal{R}_s}}\right)=0$ and corollary 1 in section 1.7 of \cite{evans1991}. Because the limit is independent from the measure of $\mathbf{V}^i$, equation \eqref{diffE} is verified.
\end{proof}
\subsection{Proof of Theorem \ref{CLT}}

\begin{proof}
 We apply lemma \ref{muirrelevance} substituting by the method of moments the measure $\mu$ with the independence measure obtaining an ordinary hypercube sample. Since $f$ is bounded we can apply Theorem 1 in \cite{Owen:92:CLT} .
\end{proof}

\end{document}